\pdfoutput=1
\documentclass[acmsmall, nonacm, authorversion, screen]{acmart}

\usepackage{booktabs}   
\usepackage{subcaption} 
\usepackage{sidecap}

\allowdisplaybreaks[4]
\usepackage{bm}
\usepackage{mathtools}
\usepackage{hyperref}
\usepackage{enumitem}
\usepackage{array}
\usepackage{wrapfig}

\usepackage{pifont}

\usepackage{stmaryrd}
\usepackage{float}
\usepackage{caption}
\usepackage{subcaption}

\usepackage{setspace}
\usepackage{xcolor}
\usepackage{diagbox}
\usepackage{makecell}
\usepackage{mathtools}

\renewcommand{\subsubsection}[1]{\paragraph{ #1}}

\newcommand{\ket}[1]{\lvert #1\rangle}
\newcommand{\bra}[1]{\langle #1\rvert}
\newcommand{\ketbra}[3][]{\ket{#2}_{#1}\bra{#3}}
\newcommand{\braket}[2]{\langle #1\vert #2\rangle}
\newcommand{\code}[1]{{\texttt{#1}}}

\newcommand{\qtrans}[2]{{T_{#2}}}
\newcommand{\qd}[1]{{\textstyle\frac{\partial}{\partial #1}}}
\DeclareMathOperator{\tr}{tr}
\def\qvar{\bm{qVar}}
\def\abs#1{\left\lvert #1\right\rvert}
\def\norm#1{\left\lVert #1\right\rVert}
\def\figref#1{Fig.~\ref{#1}}

\def\qwhile{\mathbf{while\ }}
\def\qskip{\mathbf{skip}}
\def\qif{\mathbf{if\ }}
\def\qfi{\mathbf{\ fi}}
\def\qdo{\mathbf{\ do\ }}
\def\qod{\mathbf{\ od}}
\def\dinit{\mathbf{Dinit}}
\def\qinit#1{{#1\coloneqq \ket{0}}}
\def\qutp#1#2#3{{#3\coloneqq e^{-i#1 #2}[#3]}}
\def\qut#1#2{{#2\coloneqq #1[#2]}}
\def\qinitd#1#2{{#2\coloneqq #1}}
\def\qqwhile#1#2{\qwhile M[#1] = 1 \qdo #2 \qod}
\def\sem#1{\left\llbracket #1 \right\rrbracket}
\def\overto#1{\overset{#1}{\to}}

\def\cD{\mathcal{D}}
\def\cE{\mathcal{E}}
\def\cH{\mathcal{H}}
\def\cS{\mathcal{S}}
\def\cT{\mathcal{T}}

\def\cL{\mathcal{L}}
\def\cG{\mathcal{G}}
\def\cX{\mathcal{X}}
\def\cY{\mathcal{Y}}

\def\cDa{\cD(\cH_{all})}
\def\multiset#1{\left\lbrace\!\left\lvert #1 \right\rvert\!\right\rbrace}
\DeclareMathOperator{\Span}{span}
\DeclareMathOperator{\supp}{supp}

\newtheorem{theorem}{Theorem}[section]

\newtheorem{lemma}[theorem]{Lemma}
\newtheorem{remark}[theorem]{Remark}

\newtheorem{example}[theorem]{Example}
\newtheorem{cor}[theorem]{Corollary}

\usepackage{tikz,ifthen}
\usetikzlibrary{decorations.pathreplacing,decorations.pathmorphing,shapes.geometric}
\usetikzlibrary{positioning, external}
\tikzset{location/.style={circle, fill=black,inner sep=0pt, minimum size=1mm}}
\tikzset{transition/.style={fill=white, inner sep=0pt, midway, font=\tiny}}

\def\COMPILETIKZ{0}

\tikzset{
    *|/.style={
        to path={
            (perpendicular cs: horizontal line through={(\tikztostart)},
                                 vertical line through={(\tikztotarget)})
            -- (\tikztotarget) \tikztonodes
        }
    }
}

\ifthenelse{\COMPILETIKZ=1}{
\tikzexternalize[prefix=figures/]
\tikzset{external/system call={pdflatex \tikzexternalcheckshellescape -halt-on-error -interaction=batchmode -jobname "\image" "\texsource"}}}

\newcommand{\tikzinput}[2]{
\ifthenelse{\COMPILETIKZ=1}
{\tikzsetnextfilename{#1}\input{#2}}
{\includegraphics{./figures/#1.pdf}}
}

\tikzstyle{every picture}+=[remember picture]

\newsavebox{\tempbox}

\makeatletter
\newcommand{\biggg}{\bBigg@{3}}
\def\bigggl{\mathopen\biggg}

\def\bigggr{\mathclose\biggg}
\newcommand{\Biggg}{\bBigg@{3.5}}

\newcommand{\bigggg}{\bBigg@{4}}
\def\biggggl{\mathopen\bigggg}

\def\biggggr{\mathclose\bigggg}
\newcommand{\Bigggg}{\bBigg@{4.5}}
\def\Biggggl{\mathopen\Bigggg}

\def\Biggggr{\mathclose\Bigggg}
\makeatother
\numberwithin{equation}{section}

\begin{document}

\title[Differentiable Quantum Programming with Unbounded Loops]{Differentiable Quantum Programming with Unbounded Loops}


\author[W. Fang]{Wang Fang}
\affiliation{
  \department{State Key Laboratory of Computer Science}
  \institution{Institute of Software, Chinese Academy of Sciences}
  \country{China}
}
\email{fangw@ios.ac.cn}
\affiliation{
  \institution{University of Chinese Academy of Sciences}
  \country{China}
}

\author[M. Ying]{Mingsheng Ying}
\affiliation{
  \department{State Key Laboratory of Computer Science}
  \institution{Institute of Software, Chinese Academy of Sciences}
  \country{China}
}
\email{mingshengying@gmail.com}
\affiliation{
  \department{Department of Computer Science and Technology}
  \institution{Tsinghua University}
  \country{China}
}

\author[X. Wu]{Xiaodi Wu}
\affiliation{
  \department{Department of Computer Science and Joint Center for Quantum Information and Computer Science}
  \institution{University of Maryland}
  \country{United States}
}
\email{xwu@cs.umd.edu}

\begin{abstract}
  The emergence of variational quantum applications has led to the development of automatic differentiation techniques in quantum computing.
Recently, \citet{Zhu2020Differentiable} have formulated differentiable quantum programming with bounded loops, providing a framework for scalable gradient calculation by quantum means for training quantum variational applications. 
However, promising parameterized quantum applications, e.g., quantum walk and unitary implementation, cannot be trained in the existing framework due to the natural involvement of unbounded loops.  
To fill in the gap,  we provide the first differentiable quantum programming framework with unbounded loops, including a newly designed  differentiation rule, code transformation, and their correctness proof. 
Technically, we introduce a randomized estimator for derivatives to deal with the infinite sum in the differentiation of unbounded loops, whose applicability in classical and probabilistic programming is also discussed. 
We implement our framework with Python and Q\#, and demonstrate a reasonable sample efficiency.
Through extensive case studies, we showcase an exciting application of our framework in automatically identifying close-to-optimal parameters for several parameterized quantum applications.

\end{abstract}



\keywords{quantum programming languages, differentiable programming, quantum machine learning, unbounded loops}  

\maketitle

\section{Introduction}
Inspired by the advantage of neural networks with program features (e.g. controls) over the plain ones~\cite{deep-mind, Grefenstette:2015:LTU:2969442.2969444}, the notion of \emph{differentiable programming} has been introduced~\cite{gaunt2017differentiable,abadi2019simple, wang2019demystifying, baydin2017automatic} as a new programming paradigm, where programs become parameterized and differentiable, and has recently stimulated active investigation (e.g., ~\cite{Damiano2021Automatic,Benjamin2021lambdas,Mathieu2020Diffeologies,Carol2021Probabilistic}). 
Specifically, many efforts have been devoted to the development of \emph{automatic differentiation} (e.g.,  \cite{Griewank:2000:EDP:335134,Corliss:2000:ADA:571034}) for various program constructs. 

Quantum programming is one specific type of programming that would benefit from the study of automatic differentiation. 
With the availability of the 50$\sim$100-qubit machines, near-term Noisy Intermediate-Scale Quantum (NISQ) machines~\cite{Preskill18} have become the major platform for quantum applications. 
Parameterized (or variational) quantum circuits, introduced as a quantum machine learning model with remarkable expressive power~\cite{benedetti2019parameterized}, are one compelling candidate of NISQ applications, including examples like variational quantum eigensolver (VQE)~\cite{kandala2017hardware, peruzzo2014variational}, quantum neural networks~\cite{farhi2018classification, beer2020training,cong2019quantum}, to the quantum approximate optimization algorithm (QAOA)~\cite{farhi2014quantum,farhi2016quantum, hadfield2019quantum}. 
Similar to classical machine learning, gradient-based methods are employed to train the loss functions, which, however, now depend on the read-outs of quantum computation.  
Thus, the "quantum" gradient calculation has a similar complexity of simulating quantum circuits, which is infeasible for classical computation. 

Automatic differentiation (AD) on quantum programs, which would enable the ability of computing quantum gradients efficiently by quantum computation, is thus critical for the scalability of variational quantum applications. 
However, it is a priori unclear whether the AD technique could extend to the quantum setting at all due to a few fundamental differences between quantum and classical. 
First, an appropriate formulation of the differentiation in quantum computing is important because the outcomes of quantum programs are quantum states rather than classical variables. 
Second, the quantum no-cloning theorem~\cite{wootters1982single} prohibits the duplication of intermediate states in quantum programs, which prohibits the natural extension of the classical forward-mode and reverse-mode differentiation~\cite{abadi2019simple} to quantum.

Fortunately, a series of recent research on analytical formulas of "quantum" gradients~\cite{beer2020training,schuld2019evaluating,mitarai2018quantum, farhi2018classification, guerreschi2017practical} has helped (partially) overcome these difficulties and thus enabled AD on quantum circuits, which has already been adopted in major quantum machine learning platforms, including Tensorflow Quantum~\cite{broughton2020tensorflow} and    PennyLane~\cite{bergholm2018pennylane}. 

\citet{Zhu2020Differentiable} provides the first rigorous formalization of the AD technique for quantum programs with bounded loops beyond quantum circuits.  
They also leveraged their framework in the training of a VQC instance with controls, which has superior performance than normal VQCs for certain machine learning tasks.

\vspace{1mm}\noindent\textbf{Quantum Applications with Unbounded Loops.}
Most existing AD results in quantum computing have been focusing on applications of variational quantum circuits (or their variants) for a few designated tasks, which misses the opportunity to investigate more sophisticated quantum algorithms. 
For example, parameterized quantum programs with \emph{unbounded loops} can describe a rich family of quantum algorithms with a few unspecified parameters, which 
could be trained to help quantum programs meet the run-time requirement, e.g., achieving quantum speedup in the examples of quantum walk~\cite{ambainis2005coins} and amplitude amplification~\cite{brassard2002quantum}, or generating desired unitaries which are unknown beforehand in the example of the repeat-until-success unitary implementation~\cite{PhysRevLett.114.080502}. 
Analytical derivation of these parameters, if ever possible, would likely require domain knowledge of the underlying problem, and is usually done in a case-by-case fashion. 
Instance-driven gradient-based search of these parameters is a promising alternative, which is only possible with the AD technique for unbounded loops. 

Let us dive into a simple example based on amplitude amplification (AA).
A direct adoption of the textbook AA would be written as a for-loop with a given number of iterations. 
However, this number of iterations is often hard to determine beforehand, which makes it desirable to write the algorithm as a while-loop (i.e., an unbounded loop) and let the program decide when to terminate.

\begin{wrapfigure}{r}{0.32\textwidth}
  \centering
  \vspace{-15pt}
  \begin{minipage}[t]{\linewidth}
      \begin{align*}
          & \qwhile M[r] = 1 \qdo \\
          & \qquad q: = \text{GroverRotation}[q]; \\
          & \qquad q, r: = \text{Coupling}_{\theta}[q, r] \\
          & \qod
      \end{align*}
  \end{minipage}
  \vspace{-10pt}
  \caption{Parameterized AA}\label{fig:paa}
  \vspace{-15pt}
\end{wrapfigure}

To that end, a framework with while-loops and parameterized weak measurements has been introduced~\cite{mizel2009critically, andres2020quantum} as the \emph{parameterized} AA program in \figref{fig:paa}, where parameter $\theta$ controls the coupling strength between the search variable $q$ and the measure variable $r$.
The choice of $\theta$ is critical in achieving quantum speedups. 
While an analytical solution of $\theta$ exists for certain quantum speedup~\cite{andres2020quantum}, its optimal choice that minimizes the expected run-time is still unknown. 

As shown in our case study, gradient-based methods in differentiable quantum programming could automatically identify a better choice of $\theta$ than existing literature ~\cite{andres2020quantum,mizel2009critically}, without domain knowledge about the AA algorithm, the promise of which also extends to parameterized quantum random walks and repeat-until-success unitary implementation.
This provides a strong motivation to develop the AD technique for \emph{unbounded} quantum loops.
However, there is no general AD solution for unbounded loops even in classical (imperative) programs~\cite{GP18}, which questions the feasibility of our goal. 

Indeed, as we elaborate on in Sect.~\ref{sec:challenge}, unbounded loops introduce serious challenges in AD for classical, probabilistic, and quantum programs. 
Moreover, unique features of quantum programs, like the no-cloning theorem and the branching induced by measurements, further restrict the available AD techniques for quantum programs with unbounded loops. 

\noindent \textbf{Contributions.} We overcome these challenges and develop a differentiable quantum programming framework for unbounded loops, with the following contributions:   

\begin{itemize}[leftmargin=3mm]
  \item A formulation of parameterized quantum \textbf{while}-language with unbounded loops and a new parameterized unitary operation called the density operator exponentiation $e^{-it\sigma}$ of any density operator $\sigma$ 
  which allows the inclusion of more unitary gates. (Sect.~\ref{sec:para_lang})
  \item A sufficient condition (i.e., finite-dimensional program state space) for the differentiability of quantum programs with unbounded loops (Theorem~\ref{thm:diff-exist}). We also exhibit an example of non-differentiable infinite-dimensional quantum programs (Example~\ref{eg:non-diff}) to demonstrate the difference between finite and infinite dimensional quantum programs for differentiation. (Sect.~\ref{sec:para_lang})
  \item An AD scheme for quantum programs with unbounded loops with two components:
  (1) \emph{Differentiation on a Single-Occurrence of Parameter} (DSOP) for quantum circuits with respect to a parameter with a single occurrence;  
  and (2) \emph{Extension to Unbounded Loops} (EUL) for unbounded loops based on any DSOP.  We contribute a new DSOP technique, called the \emph{commutator-form rule} inspired by~\cite{Lloyd_2014} for general $e^{-i\theta H}$, with a more general applicability\footnote{It removes the limitation of the parameter-shift rule that is only applicable when $H$ has at most two distinct eigenvalues.}. 
  We also develop the code transformation and establish its correctness (Theorem~\ref{thm:exact}). (Sect.~\ref{sec:auto_diff})
  \item Implementation of our AD scheme with Python and Q\# and discussion of its relevant sample efficiency, in which we provide an upper bound that matches the one of~\citet{Zhu2020Differentiable} when there is no unbounded loop. (Sect.~\ref{sec:implement})
  \item Extensive case study on the gradient-based approach for automatically identifying unknown parameters in quantum algorithm design, which includes the parameterized AA algorithm, quantum walk with parameterized shift operator on 2-d grids, and unitary implementation with parameterized repeat-until-success algorithms. (Sect.~\ref{sec:case_study})
\end{itemize}

\vspace{1mm} \noindent \textbf{Related Work.}
There is a rich literature on differentiation rules of quantum circuits~\cite{schuld2019evaluating,PhysRevLett.118.150503,mitarai2018quantum,Banchi2021measuringanalytic,vidal2018calculus,wierichs2021general,izmaylov2021analytic,kyriienko2021generalized}.
These researches focus on how to use quantum hardware to derive the derivative of the expectation function of a parametrized quantum circuit.
For Pauli rotations $U(\theta) = e^{-i\theta \Delta/2}, \Delta=X,Y,Z$, \citet{PhysRevLett.118.150503} and \citet{mitarai2018quantum} first proposed a formula that only need to run the initial circuit twice with different parameters to find the derivative.
Then, \citet{schuld2019evaluating} named this formula the ``parameter-shift rule'' and expanded it to a general case of $U(\theta) = e^{-i\theta H}$ with Hamiltonian $H$ having at most two distinct eigenvalues.
Recently, independent developments of variants of the parameter-shift rules (general parameter-shift rules)~\cite{kyriienko2021generalized,izmaylov2021analytic,wierichs2021general} were proposed for general Hamiltonian $H$.
Their works can be traced back to an observation that the expectation function of a PQC with a single parameter is a finite Fourier series~\cite{vidal2018calculus}. Our commutator-form rule, which is applicable to $e^{-i\theta H}$ for general $H$, is based on a very different technique and has a simple form compared to general parameter-shift rules.

Most existing AD techniques in quantum computing~\cite{bergholm2018pennylane,luo2020yao,di2022quantum,kottmann2021tequila,nguyen2022qsun,broughton2020tensorflow} work with simple languages describing quantum circuits without any control flow. Some of these results, e.g., \texttt{Yao.jl}~\cite{luo2020yao}, also apply classical AD techniques to classical programs that simulate quantum circuits, which is unfortunately not scalable for real quantum applications.  

The only exception and also the closest work to ours is \cite{Zhu2020Differentiable}, which proposed differentiable quantum programming with bounded loops beyond quantum circuits. 
Although the syntax in~\cite{Zhu2020Differentiable} supports general parameterized gates, 
its code transformation only supports Pauli rotation gates based on a variant of the parameter-shift rule.
To handle AD of bounded quantum loops, \citet{Zhu2020Differentiable} used a finite collection of quantum programs and added up their outputs for the derivative.
As elaborated on in Sect.~\ref{sec:challenge}, one cannot extend this approach to a collection of unbounded sizes like unbounded loops in this paper.
The correctness and feasibility of this paper to deal with infinite summation caused by unbounded loops is the main difficulty, which \citet{Zhu2020Differentiable} didn't encounter.
Moreover, the efficiency of this paper is comparable to~\cite{Zhu2020Differentiable} in the bounded-loop setting.
Thus, this paper strictly improves \cite{Zhu2020Differentiable}.

\section{Challenges and Our Key ideas}
\label{sec:challenge}

Let us first revisit classical, probabilistic, and quantum programs as shown in~\figref{fig:example_programs} with one unbounded loop in order to understand their features, differences, and corresponding difficulties in differentiation. We assume some simple quantum terminology and refer readers to a more detailed preliminary in Sect.~\ref{sec:para_lang}. 

While-loop, an important construct to make an imperative programming language Turing-complete, may cause an arbitrary number of loops or infinite loops (not terminated).
In \emph{classical (deterministic) programs}, such as \figref{fig:example:a}, the variable $r$ is assigned with an integer $k$. When $k \geq 0$, the while-loop will execute the loop body $k$ times and terminate and when $k < 0$, the while-loop will execute the loop body for infinitely many times and not terminate.
Nevertheless, due to the deterministic nature of classical programs, \emph{there is one and only one path of the program execution}.

\begin{figure}[ht]
  \centering
  \begin{subfigure}[b]{.24\linewidth}
    \centering
    \begin{align*}
      &r \coloneqq k; \\
      &\qwhile r \neq 0 \qdo \\
      &\quad r \coloneqq r-1\\
      &\qod
    \end{align*}
    \caption{Classical.}\label{fig:example:a}
  \end{subfigure}
  \hfill
  \begin{subfigure}[b]{.24\linewidth}
    \centering
    \begin{align*}
      &r \coloneqq 0\oplus_{0.5}1; \\
      &\qwhile r \neq 0 \qdo \\
      &\quad r \coloneqq 0\oplus_{0.5}1\\
      &\qod
    \end{align*}
    \caption{Probabilistic.}\label{fig:example:b}
  \end{subfigure}
  \hfill
  \begin{subfigure}[b]{.24\linewidth}
    \centering
    \begin{align*}
      &r \coloneqq (\ket{0}-\ket{1})/\sqrt{2}; \\
      &\qwhile M[r] \neq 0 \qdo \\
      &\quad  r \coloneqq H[r]\\
      &\qod
    \end{align*}
    \caption{Quantum.}\label{fig:example:c}
  \end{subfigure}
  \hfill
  \begin{subfigure}[b]{.24\linewidth}
    \centering
    \begin{align*}
      &r \coloneqq (\ket{0}-\ket{1})/\sqrt{2}; \\
      &\qwhile^{\!\!(k)} M[r] \neq 0 \qdo \\
      &\quad  r\coloneqq H[r]\\
      &\qod
    \end{align*}
    \caption{Quantum, bounded.}\label{fig:example:d}
  \end{subfigure}
  \caption{Four while-loop programs to demonstrate unbounded loops in quantum programs.}\label{fig:example_programs}
\end{figure}

However, there may be an infinite number of execution paths in a \emph{quantum program} or \emph{probabilistic program}.
In~\figref{fig:example:b}, the command $r \coloneqq 0\oplus_{0.5}1$ assigns $0$ to variable $r$ with probability $0.5$ or $1$ to $r$ with probability $0.5$ otherwise. Thus, the probability of the probabilistic program in~\figref{fig:example:b} execute the loop body $k$, $k\geq 0$, times is $0.5^{k+1}$, which means this program has an \emph{infinite} number of execution paths.
Similarly, the quantum program in \figref{fig:example:c} also has an \emph{infinite} number of execution paths.
Let's see the execution of the program in \figref{fig:example:c}:
\begin{itemize}
  \item[(1)] First, $r$ is assigned with state $\ket{-}=(\ket{0}-\ket{1})/\sqrt{2}$. 
  \item[(2)] Second, in the measurement of the while-loop, state $\ket{-}$ is measured with $\{M_0 = \ketbra{0}{0}, M_1 = \ketbra{1}{1}\}$. The measurement outcome will be \emph{$0$ with probability $\bra{-}M_0\ket{-} = 0.5$} and \emph{$1$ with probability $\bra{-}M_1\ket{-} = 0.5$}. When the outcome is $0$, the program will terminate; when the outcome is $1$, the measured state of $r$ will become $\ket{1}$, and the program will enter the loop body (3).
  \item[(3)] In the loop body, applied with $H$, the state of $r$ becomes $\ket{-} = H\ket{1}=\frac{1}{\sqrt{2}}(\ket{0}-\ket{1})$, then the program goes back to the measurement of the while-loop (2).
\end{itemize} 
From (2) above, we can see that the program will terminate with probability $0.5$ or continue the while-loop with probability $0.5$ at each entry of the while-loop.
Therefore, the probability of this program executing loop body ($r \coloneqq H[r]$) for $k$, $k\geq 0$, times is $0.5^{k+1}$, which means this program has an \emph{infinite} number of execution paths.
Note the probability $0.5$ is implicitly implied by the measurement outcome of the quantum program, whereas in probabilistic programs such probabilities are explicitly given by the syntax. 

Another important difference worth highlighting between classical (deterministic or probabilistic) programs and quantum ones is the quantum no-cloning theorem that prohibits the use of the \emph{chain-rule-based}  forward/reverse differentiation in the quantum setting. 
As a result, AD techniques in quantum are somewhat separate from those commonly studied in the classical AD literature. 

Finally, consider a \emph{bounded-loop quantum program} in \figref{fig:example:d}, which is investigated in \cite{Zhu2020Differentiable}.
The number $k > 0$ in $\qwhile^{\!\!(k)}$ limits the iteration times of the loop body upto $k$, which results that the program in \figref{fig:example:d} has at most $k+1$, and hence a \emph{finite} number of  execution paths.

We summarize these comparisons in Table~\ref{tab:difference}. As we will see, dealing with infinitely many execution paths is one major difficulty in differentiation over unbounded loops, either probabilistic or quantum. 
The unusability of the Chain rule further complicates the quantum case. 

\begin{table}[ht]
  \centering
  \scalebox{0.85}{
  \renewcommand{\arraystretch}{1.2}
  \begin{tabular}{lcccc}
      \toprule
      & Classical & Probabilistic & Quantum \\
      \hline
      \# of Execution Paths & one & \textbf{possibly infinitely many} & \textbf{possibly infinitely many} \\
      \hline
      Distribution over Paths & none & explicit by sampling & \textbf{implicit by measurement} \\
      \hline
      Usability of Chain Rule & yes & yes & \textbf{no} \\
      \hline
      Differentiability & \thead{almost everywhere \\ \cite{Damiano2021Automatic}} & \thead{almost everywhere \\ for every path \\ \cite{Carol2021Probabilistic}, \\ sufficient conditions \\ for bounded loops \\ \cite{10.1145/3371084}} &  \thead{sufficient condition \\ for \textbf{unbounded loops}, \\ Theorem~\ref{thm:diff-exist}}\\
      \hline
  \end{tabular}}
  \caption{Comparisons among classical, probabilistic and quantum programs.}\label{tab:difference}
  \vspace{-15pt}
\end{table}

\subsection{Differentiability of Unbounded Quantum Loops}\label{sec:differentiability}
The \emph{differentiability} of unbounded quantum loops should be the first question to address, as it is already quite non-trivial in establishing so in classical and probabilistic functional programs. 

In classical programs, conditionals often lead to piecewise-defined functions and non-differentiable points, which are discontinuous or have different left and right derivatives~\cite{beck1994if}.
Even if the function defined by a program is differentiable, syntactic discontinuity can make AD fails\footnote{
Consider the program presented in~\cite{abadi2019simple,Damiano2021Automatic}:
\[\textsf{SillyId}\ \equiv\  \lambda x^{\mathbb{R}}.\textsf{if}\ x = 0\ \textsf{then}\ 0 \ \textsf{else}\ x.\]
We can see that $\sem{\textsf{SillyId}}(x) = x$, thus $\sem{\textsf{SillyId}}$ has a constant derivative of $1$.
However, general AD would produce the wrong answer $0$ at the point $x = 0$.}.
To resolve the issue of conditionals, \citet{abadi2019simple} adopted the ``partial conditionals'', which ignores the boundary case, and proved the correctness of AD on conditionals and recursion.
\citet{Damiano2021Automatic} characterized the set of ``stable points'', 
the intuition behind which is the point that has an open neighborhood with the same execution trace (``execution path''). They proved that AD is almost everywhere correct under the mild hypothesis. 
However, these arguments are developed for one execution path in classical programs. 

The case of probabilistic programs resembles a lot with quantum programs due to possibly infinitely many execution paths. 
Whether probabilistic programs would lead to non-differentiable densities at some non-measure-zero set has been an important open question in the field~\cite{yang2019some}. 
Recently, \citet{Carol2021Probabilistic} considered higher-order probabilistic programming with recursion and proved that a probabilistic program's density is almost everywhere differentiable under mild hypothesis.
However, it only tells us the differentiability of any trace of sampled values during execution, which implies a fixed execution path.

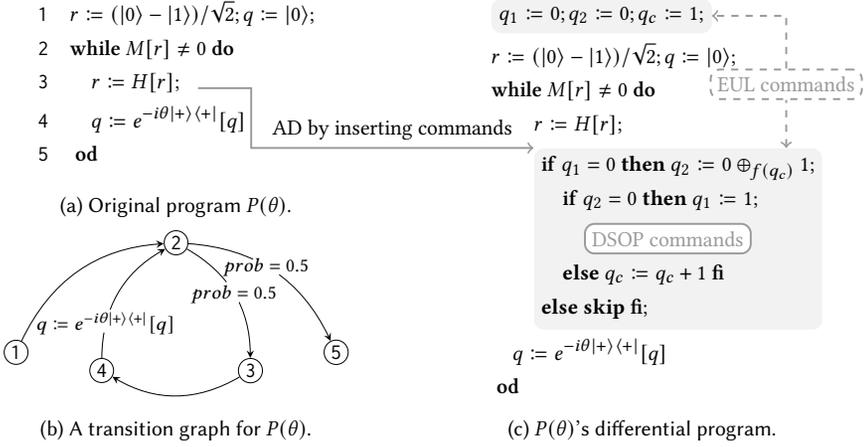
\begin{figure}[ht]
  \small
  \scalebox{0.9}{
    \centering
  \begin{subfigure}[b]{.48\linewidth}
    \centering
    \begin{subfigure}[b]{\linewidth}
        \centering
        \begin{align*}
          \texttt{1}\quad &r \coloneqq (\ket{0}-\ket{1})/\sqrt{2}; q \coloneqq \ket{0}; \\
          \texttt{2}\quad &\qwhile M[r] \neq 0 \qdo \\
          \texttt{3}\quad &\quad \qut{H}{r}; \tikz[baseline]{\node[anchor=base] (runa) {}}\\
          \texttt{4}\quad &\quad \qut{e^{-i\theta\ketbra{+}{+}}}{q}
          \\
          \texttt{5}\quad &\qod
        \end{align*}
        \caption{Original program $P(\theta)$.}\label{fig:running_example:a}
    \end{subfigure}
    
    \begin{subfigure}[b]{\linewidth}
      \centering
      \begin{tikzpicture}[every node/.style={circle,draw,inner sep=1pt}]
        \node (1) {\texttt{1}};
        \node[right=2cm of 1,yshift=1.6cm] (2) {\texttt{2}};
        \node[right=2cm of 2,yshift=-1.6cm] (5) {\texttt{5}};
        \node[below=1.5cm of 2,xshift=-1.1cm] (4) {\texttt{4}};
        \node[below=1.5cm of 2,xshift=1.1cm] (3) {\texttt{3}};
        \path[-stealth]
          (1) edge[bend left] (2)
          (2) edge[bend left] node[fill=white,draw=none,inner sep=4.5pt,pos=0.45] {} node[transition,draw=none,fill=none,pos=0.45] {\footnotesize $prob=0.5$} (5)
          (2) edge[bend left] node[fill=white,draw=none,inner sep=3.5pt,pos=0.48] {} node[transition,draw=none,fill=none] {\footnotesize $prob=0.5$} (3)
          (3) edge[bend left] (4)
          (4) edge[bend left] node[fill=white,draw=none,inner sep=3.5pt,pos=0.27] {} node[transition,draw=none,fill=none,pos=0.25] {\footnotesize $\qut{e^{-i\theta\ketbra{+}{+}}}{q}$} (2);
      \end{tikzpicture}
      \caption{A transition graph for $P(\theta)$.}\label{fig:running_example:b}
    \end{subfigure}
  \end{subfigure}
  \hfill
  \begin{subfigure}[b]{0.5\linewidth}
    \centering
    \begin{subfigure}[b]{\linewidth}
      \centering
      \hskip 12pt $\begin{aligned}
        & \tikz[baseline]{\node[anchor=base,thick,draw=black!5,text=black,fill=black!5,rounded corners=4pt] (cad3) {$q_1\coloneqq 0;q_2\coloneqq 0;q_c\coloneqq 1;$}} \\
        & r \coloneqq (\ket{0}-\ket{1})/\sqrt{2}; q \coloneqq \ket{0}; \\
        & \qwhile M[r] \neq 0 \qdo \\
        & \qquad \qut{H}{r};\\
        & \qquad \tikz[baseline]{\node[anchor=base,text=black,fill=black!5,rounded corners=4pt] (cad5) {$\begin{aligned}
          & \qif q_1 = 0 \textbf{ then } q_2 \coloneqq 0 \oplus_{f(q_c)} 1; \\
          & \quad \qif q_2 = 0 \textbf{ then } q_1 \coloneqq 1; \\
          & \quad \quad \tikz[baseline]{\node[anchor=base,thick,draw=black!40,rounded corners=4pt,text=black!40,fill=white] (cad6) {DSOP commands}} \\
          & \quad \textbf{else } q_c \coloneqq q_c + 1 \qfi \\
          & \textbf{else } \textbf{skip} \qfi;
        \end{aligned}$}} \\
        & \quad \qut{e^{-i\theta\ketbra{+}{+}}}{q} \\
        & \qod
      \end{aligned}$
      \caption{$P(\theta)$'s differential program.}
      \label{fig:running_example:c}
    \end{subfigure}
  \end{subfigure}
  \begin{tikzpicture}[overlay]
    \node[thick,draw=black!40,text=black!40,rounded corners=4pt,above=0.5cm of cad5, xshift=1.6cm,yshift=0.2cm,dashed] (cad8) {EUL commands};
    \path[draw=black!40,->,thick,dashed] (cad8) |- (cad3);
    \path[draw=black!40,<-,thick,dashed] (cad5.north) to[*|] (cad8.south);
    \path[draw=black!40,->,thick] (runa) -- +(0.9,0) |- node[above,pos=0.75] {AD by inserting commands} (cad5.north west);
  \end{tikzpicture}}
  \caption{Running example to demonstrate our AD scheme. Every time the program runs to a statement, e.g., $\qutp{\theta}{\ketbra{+}{+}}{q}$ here,  that containing $\theta$, it will first enter a block of EUL to decide whether to continue $P(\theta)$ or to do a differentiation operation by DSOP then continue $P(\theta)$. $0\oplus_{p}1$ is a probabilistic choice outputs $0$ with probability $p$ and $1$ with probability $1-p$ for any $0\leq p\leq 1$ and $f(n) = \mu(n) /(1-\sum_{j=1}^{n-1}\mu(j))$ is determined by the distribution $\mu$ mentioned in Sect~\ref{sec:eul_difficulty}.}\label{fig:running_example_programs}
  \vspace{-12pt}
\end{figure}

Recall the program in Fig.~\ref{fig:example:c}, on which we add a command with parameter $\theta$ into the loop body as our running example in Fig.~\ref{fig:running_example_programs}.
The transition graph for each line of program $P(\theta)$ in \figref{fig:running_example:a} is shown in \figref{fig:running_example:b}, where the behavior of $P(\theta)$ is the same as in \figref{fig:example:c}: $P(\theta)$ will terminate (goto line \texttt{5}) with probability $0.5$ or continue the
while-loop (goto line \texttt{3}) with probability $0.5$ at each entry (line \texttt{2}) of the while-loop.
The probability of $P(\theta)$ executing $k$ times $\qut{e^{-i\theta\ketbra{+}{+}}}{q}$ is $0.5^{k+1}$.
Consider operations related to variable $q$, $P(\theta)$ induces a semantic function like $f_k(\theta) = \sem{\qut{e^{-i\theta\ketbra{+}{+}}}{q}}^k$ with probability $0.5^{k+1}$.
The differentiability of $f_k(\theta)$ is generally easily   obtainable, while the differentiability of its expectation $F(\theta) = \sum_{k=0}^\infty \frac{1}{2^{k+1}}f_{k}(\theta)$ is unclear.
For this problem of expectation, two conditions are proposed in the probabilistic programming~\cite{10.1145/3371084}:
\begin{itemize}
  \item \emph{Differentiability of expectation (infinite summation)}: \begin{equation}\label{eq:r4}\tag{Condition A1}
    F(\theta) = \sum_k \mu(k)f_{k}(\theta)
  \text{ is differentiable on $\mathbb{R}$}.
  \end{equation}
  \item \emph{Exchangeability between differentiation and infinite summation}:
  \begin{equation}\label{eq:r5}\tag{Condition A2} 
    \text{for all $\theta \in \mathbb{R}$}, \partial_{\theta} \sum_k \mu(k)f_{k}(\theta) = \sum_k \mu(k) \partial_{\theta}(f_{k}(\theta)).
  \end{equation}
\end{itemize}
where $\mu$ is a probability distribution over index $k$.
\ref{eq:r4} states the premise of AD: we wouldn't talk about AD if the function induced by the program was not differentiable. \ref{eq:r5} states that the differential operation on every trace ($\partial_{\theta}f_{k}(\theta)$, which reflects the underlying AD implementation) can be collected into the differential operation on the total program ($\partial_{\theta}\sum_k \mu(k) f_{k}(\theta)$).
To the best of our knowledge, no research has yet investigated what kind of probabilistic programs or quantum programs with unbounded loops meets \ref{eq:r4} and \ref{eq:r5}.

\subsubsection{Our solution:} Fortunately, we identify \emph{finite-dimensional state space}, which is met by all  existing quantum applications, as a sufficient condition to satisfy \ref{eq:r4} and \ref{eq:r5}. 
Technically, under this condition, the probability that an unbounded quantum loop iterates $k$ times has an exponential decay on $k$ (Lemma~\ref{lem:a9}). 
The cornerstone behind this lemma is the \emph{compactness} of finite-dimensional Hilbert spaces.
As in our running example, $f_k(\theta) = \sem{\qut{e^{-i\theta\ketbra{+}{+}}}{q}}^k$, which corresponds to loop $k$ times, appears with probability $0.5^{k+1}$.
Let $g_{\theta} = \sem{\qut{e^{-i\theta\ketbra{+}{+}}}{q}}$, the differential of $f_{k}(\theta)$ is 
$\partial_{\theta}f_k(\theta) = \sum_{j=1}^k (g_{\theta})^{j-1}\circ \left(\partial_\theta g_{\theta}\right)\circ(g_{\theta})^{k-j}$, then
\begin{equation}\label{eq:diff_right}
  \sum_{k=0}^{\infty}\frac{1}{2^{k+1}}\partial_{\theta}f_k(\theta) = \sum_{k=0}^{\infty}\frac{1}{2^{k+1}}\sum_{j=1}^k (g_{\theta})^{j-1}\circ \left(\partial_\theta g_{\theta}\right)\circ(g_{\theta})^{k-j}
\end{equation}
is uniformly convergent as $\frac{1}{2^{k+1}}\sum_{j=1}^k (g_{\theta})^{j-1}\circ \left(\partial_\theta g_{\theta}\right)\circ(g_{\theta})^{k-j} \in \mathcal{O}(\frac{k}{2^{k+1}})$.
This uniform convergence implies both \ref{eq:r4} and \ref{eq:r5}.

Conversely, through Weierstrass' non-differentiable function $S(x) = \sum_{k=0}^{\infty}a^n\sin(b^n x)$~\cite[Theorem 1.31]{hardy1916weierstrass}, we can construct a counterexample (Example~\ref{eg:non-diff}) that is nowhere differentiable on $\mathbb{R}$ for quantum loops with infinite-dimensional space.
Specifically, Example~\ref{eg:non-diff} has two nested loops that induces $f_{k}(\theta) = \sem{\qut{e^{-i\theta\ketbra{+}{+}}}{q}}^{2^k}$ with probability $0.5^{k+1}$.
The differential of $f_k(\theta)$ becomes $\partial_{\theta}f_k(\theta) = \sum_{j=1}^{2^k} (g_{\theta})^{j-1}\circ \left(\partial_\theta g_{\theta}\right)\circ(g_{\theta})^{2^k-j}$, then this summation of $2^k$ terms makes $\sum_{k=0}^{\infty}\frac{1}{2^{k+1}}\partial_{\theta}f_k(\theta) = \sum_{k=0}^{\infty}\frac{1}{2^{k+1}}\sum_{j=1}^{2^k} (g_{\theta})^{j-1}\circ \left(\partial_\theta g_{\theta}\right)\circ(g_{\theta})^{2^k-j}$ divergent everywhere.
Note that the probability $0.5^{k+1}$ comes from the outer loop and the exponent $2^k$ comes from the inner loop.
For the latter to happen, the inner loop, like Example~\ref{eg:non-diff}, needs to have an infinite-dimensional register to record the information of $k$, which goes to positive infinite.
This non-differentiable counterexample can also be represented by the expectation of a probabilistic program (see Example~\ref{eg:non-diff}).  

\subsection{Execution Paths with Infinitely Many Parameter Occurrences}
\label{sec:eul_difficulty}

As we saw in~\figref{fig:running_example_programs}, loops lead to repeated execution of loop body, which means a parameter can appear many times at a single execution path, e.g., the function $f_k(\theta) = \sem{\qut{e^{-i\theta\ketbra{+}{+}}}{q}}^k$ introduced in Sect.~\ref{sec:differentiability} has $k$ occurrences of $\theta$.
These execution paths refer to \emph{quantum circuits with a parameter appearing multiple times}.
In the case of bounded quantum loops, the total number of execution paths is finite, and the multiplicity of parameter occurrences on each path is also bounded. 
Therefore, \citet{Zhu2020Differentiable} proposed additive quantum programs to represent a (finite) collection of quantum programs that compute partial derivatives of all occurrences.  
Conceivably, this approach does not extend to the infinite (and unbounded) case.

\subsubsection{Our solution:}
According to Sect.~\ref{sec:differentiability}, the differential of $F(\theta) = \sum_{k=0}^\infty \frac{1}{2^{k+1}}f_{k}(\theta)$ should be equal to
\begin{equation}\label{eq:all_partial_derivative}
  \sum_{k=0}^{\infty}\frac{\partial_{\theta}f_k(\theta)}{2^{k+1}} = \sum_{k=0}^{\infty}\sum_{j=1}^k\frac{1}{2^{k+1}} (g_{\theta})^{j-1}\circ \left(\partial_\theta g_{\theta}\right)\circ(g_{\theta})^{k-j} = \sum_{j=1}^{\infty}\sum_{k=1}^{\infty}\frac{1}{2^{k+1}}(g_{\theta})^{j-1}\circ \left(\partial_\theta g_{\theta}\right)\circ(g_{\theta})^{k-j}.
\end{equation}
We rewrite (\ref{eq:all_partial_derivative}) as $\sum_{j=1}^{\infty}a_j$ with $a_j = \sum_{k=1}^{\infty}\frac{1}{2^{k+1}}(g_{\theta})^{j-1}\circ \left(\partial_\theta g_{\theta}\right)\circ(g_{\theta})^{k-j}$ is the corresponding summation term for all execution paths with differentiation on $j$-th position.
With the infeasibility of the Chain rule, the existing quantum AD technique can only handle $a_j$ one by one, the cost of which will depend on the number of terms in~(\ref{eq:all_partial_derivative}) that will be unbounded in our case. 

Inspired by the \emph{importance sampling} in statistics, we construct a random variable $X$ such that
\[ \Pr(X = a_j/\mu(j)) = \mu(j), \forall j \in \mathbb{Z}_+\]
to estimate the infinite sum $\sum_{j}^{\infty}a_j$, where $\mu$ is a probability distribution on $\mathbb{Z}_+$.
By construction, the expectation of $X$ is $\sum_{j=1}^{\infty}a_j$.
Let's now focus on Fig.~\ref{fig:running_example:c}, where the role of EUL commands, written with probabilistic pseudocode, is to generate the distribution $\mu$ as the original program runs, and to attribute this probabilistic distribution to DSOP commands (differentiation operations that induce $a_j$) in the execution order.
Precisely, the variable $q_c$ records the number of loops until $q_2 = 0$. Together with the probabilistic choice $q_2 \coloneqq 0\oplus_{f(q_c)}1$ (see definition in Fig.~\ref{fig:running_example_programs}'s caption), the probability of $q_c = j$, $q_2 = 0$ and $q_1=1$, which means DSOP commands (a differentiation operation) are executed, and one term of $a_j$ is evaluated, is $\mu(j)$ when fixing an execution path of $P(\theta)$ with loops' number $n > j$.
Therefore, we associate the probability $\mu(j)$ with $a_j$ in the execution order.
Then, the desired random variable $X$ is natural to construct.

However, for the estimation efficiency of $X$'s expectation, the variance of $X$ should also be bounded.
To that end, we identify a sufficient condition for the distribution  $\mu:\mathbb{Z}_+\to[0,1]$ as 
\begin{equation}
    \label{condition}
    \lim_{n\to\infty} \sqrt[n]{\mu(n)} = 1, \tag{converging-rate condition}
\end{equation}
which would imply the correctness of our code-transformation (Theorem~\ref{thm:exact}) and its efficiency (Theorem~\ref{thm:bound2}).
Another implicit but critical property of our construction of random variable $X$ in Fig.~\ref{fig:running_example:c} is its independence of the underlying execution path.
It allows us to apply a simple and uniform code transformation while keeping all existing quantum branches that lead to all execution paths.

\subsection*{Applicability to classical and probabilistic programs with unbounded loops?}

Ignoring the issue of differentiability, 
our idea of constructing random variables to sample partial derivatives can be applied to classical and probabilistic unbounded loops.
However, it would be less efficient because partial derivatives can be collected and forward/backward-propagated along the execution path by the Chain rule in the classical setting, which is not leveraged by our scheme. 

Consider the example in~\figref{fig:eul_probabilistic} where we show how EUL is applied to a probabilistic program.
Denote the expectation of variable $y$ after executing $P(\theta)$ in~\figref{fig:eul_probabilistic_a} as $E(\theta) = \sum_{j=1}^{\infty} j \sin(\theta)\cdot 0.5^j = 2\sin(\theta)$. Its derivative is $E'(\theta) = 2\cos(\theta)$.
We can directly apply the classical forward-AD with a dual number (e.g., ~\cite{baydin2017automatic}) 
as in~\figref{fig:eul_probabilistic_b}.
Then the expectation of $\hat{y}$ (dual number of $y$) after executing the program in~\figref{fig:eul_probabilistic_b} is $\sum_{j=1}^\infty j \cos(\theta)\cdot 0.5^j = 2\cos(\theta) = E'(\theta)$.

As a comparison, the rewriting of our EUL to the probabilistic program $P(\theta)$ becomes~\figref{fig:eul_probabilistic_c}.
With the probabilistic choice $0\oplus_{f(q_c)}1$ and a variable $q_c$ to count the occurrences of $\theta$, the program in~\figref{fig:eul_probabilistic_c} has at most one time runs into the command $(*)$ and then never increases $q_c$ with probability $\mu(j)\cdot 0.5^{j-1}$ if the output of $q_c$ is $j \geq 1$.
Then, the probability of $\hat{y} = \cos(\theta), q_c = j$ after executing the program in~\figref{fig:eul_probabilistic_c} is $\mu(j)\cdot 0.5^{j-1}$.
We can construct a random variable $X$ with respect to $\hat{y}$ and $q_c$ satisfies that $\Pr(X = \hat{y}/\mu(j)) \equiv \Pr(\hat{y} = \cos(\theta), q_c = j) = \mu(j)\cdot 0.5^{j-1}$.
Hence, the expectation of $X$ is $\sum_{j=1}^{\infty} \cos(\theta)/\mu(j) \cdot \mu(j)\cdot 0.5^{j-1} = \sum_{j=1}^\infty \cos(\theta)\cdot 0.5^{j-1} = 2\cos(\theta)$, which is consistent with the above $E'(\theta)$.
The major difference between \figref{fig:eul_probabilistic_b} and \figref{fig:eul_probabilistic_c} is that \figref{fig:eul_probabilistic_b} can execute the command $(*)$ multiple times, while \figref{fig:eul_probabilistic_c} only executes $(*)$ once. 

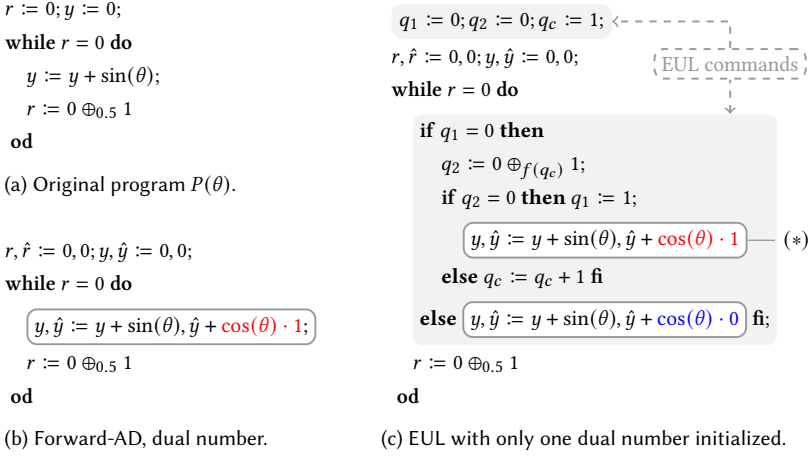
\begin{figure}[ht]
  \small
  \scalebox{0.9}{
    \centering
    \hspace{1cm}
  \begin{subfigure}[b]{0.3\linewidth}
    \captionsetup{singlelinecheck = false, justification=justified}
    \begin{subfigure}[b]{\linewidth}
      $\begin{aligned}
        & r\coloneqq 0; y\coloneqq 0; \\
        & \qwhile r = 0 \qdo \\
        & \quad y \coloneqq y + \sin(\theta); \\
        & \quad r \coloneqq 0 \oplus_{0.5} 1 \\
        & \qod
      \end{aligned}$
      \caption{Original program $P(\theta)$.}
      \label{fig:eul_probabilistic_a}
    \end{subfigure}

    \vspace{0.6cm}

    \begin{subfigure}[b]{\linewidth}
      $\begin{aligned}
        & r,\hat{r}\coloneqq 0,0; y,\hat{y}\coloneqq 0,0; \\
        & \qwhile r = 0 \qdo \\
        & \quad \tikz[baseline]{\node[anchor=base,thick,draw=black!40,rounded corners=4pt,text=black] (cad2) {$y,\hat{y} \coloneqq y+\sin(\theta),\hat{y}+{\color{red}\cos(\theta)\cdot 1};$}} \\
        & \quad r \coloneqq 0 \oplus_{0.5} 1 \\
        & \qod
      \end{aligned}$
      \caption{Forward-AD, dual number.}
      \label{fig:eul_probabilistic_b}
    \end{subfigure}
  \end{subfigure}
  \begin{subfigure}[b]{0.62\linewidth}
    \centering
    \begin{subfigure}[b]{\linewidth}
      \centering
      $\begin{aligned}
        & \tikz[baseline]{\node[anchor=base,thick,draw=black!5,text=black,fill=black!5,rounded corners=4pt] (cad3) {$q_1\coloneqq 0;q_2\coloneqq 0;q_c\coloneqq 1;$}} \\
        & r,\hat{r}\coloneqq 0,0; y,\hat{y}\coloneqq 0,0; \\
        & \qwhile r = 0 \qdo \\
        & \quad \tikz[baseline]{\node[anchor=base,text=black,fill=black!5,rounded corners=4pt] (cad5) {$\begin{aligned}
          & \qif q_1 = 0 \textbf{ then} \\
          & \quad q_2 \coloneqq 0 \oplus_{f(q_c)} 1; \\
          & \quad \qif q_2 = 0 \textbf{ then } q_1 \coloneqq 1; \\
          & \quad \quad \tikz[baseline]{\node[anchor=base,thick,draw=black!40,rounded corners=4pt,text=black,fill=white,pin={[overlay,pin edge={overlay}]right:$(*)$}] (cad6) {$y,\hat{y} \coloneqq y+\sin(\theta),\hat{y}+{\color{red}\cos(\theta)\cdot 1}$}} \\
          & \quad \textbf{else } q_c \coloneqq q_c + 1 \qfi \\
          & \textbf{else } \tikz[baseline]{\node[anchor=base,thick,draw=black!40,rounded corners=4pt,text=black,fill=white] (cad7) {$y,\hat{y} \coloneqq y+\sin(\theta),\hat{y}+{\color{blue}\cos(\theta)\cdot 0}$}} \qfi;
        \end{aligned}$}} \\
        & \quad r \coloneqq 0 \oplus_{0.5} 1 \\
        & \qod
      \end{aligned}$
      \caption{EUL with only one dual number initialized.}
      \label{fig:eul_probabilistic_c}
    \end{subfigure}
  \end{subfigure}
  \begin{tikzpicture}[overlay]
    \node[thick,draw=black!40,text=black!40,rounded corners=4pt,above=0.5cm of cad5, xshift=2cm,dashed] (cad8) {EUL commands};
    \path[draw=black!40,->,thick,dashed] (cad8) |- (cad3);
    \path[draw=black!40,<-,thick,dashed] (cad5.north) to[*|] (cad8.south);
  \end{tikzpicture}}
  \caption{A probabilistic program $P(\theta)$ with forward-AD and our EUL-based AD applied, where $0\oplus_{p}1$ is a probabilistic choice outputs $0$ with probability $p$ and $1$ with probability $1-p$ for any $0\leq p\leq 1$, $f(n) = \mu(n) /(1-\sum_{j=1}^{n-1}\mu(j))$ is determined by the distribution $\mu$ mentioned in Sect.~\ref{sec:eul_difficulty}.}
  \label{fig:eul_probabilistic}
  \vspace{-10pt}
\end{figure}

\section{Parameterized Quantum While-Programs}\label{sec:para_lang}
In this paper, we expand a parameterized extension~\cite{Zhu2020Differentiable} of quantum \textbf{while}-language~\cite{ying2016foundations} to include unbounded loops.
We assume that the reader is familiar with the basic knowledge of quantum computing, which is deferred to Appendix~\ref{sec:prelim}, and provide a summary of notation in Table~\ref{tab:notation}.

\begin{table}[ht]
    \small 
    \setlength{\tabcolsep}{3pt}
    \begin{tabular}{llllll}
        \multicolumn{2}{l}{\textbf{Sets}, \textbf{Spaces}} & \multicolumn{2}{l}{\textbf{States}} & \multicolumn{2}{l}{\textbf{Operators}, \textbf{Operations}}\\
        $\cH,\cH_q$ & Hilbert space & $\ket{\psi},\ket{0},\ket{1},\ket{+},\ket{-}$ & pure state & \makecell[l]{$U,V,e^{-i\theta \sigma}$, \\[-2pt] $H,X,Y,Z,X\otimes X$} & unitary \\
        $\cL(\cH)$ & \makecell[l]{linear operators \\[-2pt] of $\cH$} & $\rho,\sigma,\ketbra{\psi}{\psi}$ & \makecell[l]{density operator} & \makecell[l]{$M,\{M_m\},\{\ketbra{0}{0},\ketbra{1}{1}\}$ \\[-1pt] $O=\sum_{m}\lambda_m\ketbra{\psi_m}{\psi_m}$} & \makecell[l]{measurement \\[-1pt] observable} \\
        $\cD(\cH)$ & \makecell[l]{partial density \\[-2pt] operators of $\cH$} &  $\rho$ with $\tr(\rho) \leq 1$ & \makecell[l]{partial density \\[-2pt] operator} & $\cE:\rho\mapsto\sum_jE_j\rho E_j^{\dagger}$ & \makecell[l]{superoperator}\\
    \end{tabular}
    \caption{A summary of notation used in this paper.}\label{tab:notation}
    \vspace{-1cm}
\end{table}

\subsection{Syntax}\label{sec:syntax}
Let us first define the syntax of our programming language. Similar to~\cite{ying2016foundations, Zhu2020Differentiable}, we assume a countably infinite set $\qvar$ of quantum variables and use the symbols $q, q', q_0, q_1, \ldots \in \qvar$ as metavariables ranging over them.
Each quantum variable $q\in \qvar$ has a type of Hilbert space $\cH_q$ as its state space. A quantum register $\bar{q} = q_1, q_2, \ldots, q_n$ is a finite sequence of distinct quantum variables, and its state space is
$\cH_{\bar{q}} = \bigotimes_{j=1}^n\cH_{q_j}$.

\begin{definition}
    [Syntax]\label{syntax}
    A $k$-parameterized quantum \textbf{while}-program with parameter $\bm{\theta}\in \mathbb{R}^k$,
    is generated by the syntax:
        \begin{align*}
            P(\bm{\theta}) \Coloneqq{} &\qskip \mid \qinit{q} 
            \mid \qut{U}{\bar{q}} \mid \qutp{\theta}{\sigma}{\bar{q}} \mid \\ 
            & P_1(\bm{\theta}); P_2(\bm{\theta}) \mid \qif (\Box m\cdot M[\bar{q}] = m \to P_m(\bm{\theta})) \qfi \mid\\ 
            & \qqwhile{\bar{q}}{P(\bm{\theta})}
        \end{align*} 
\end{definition}

\subsubsection{Explanation of the syntax:}
Statement $\qskip$ does nothing and terminates immediately. 
Statement $\qinit{q}$ sets the quantum variable $q$ to $\ket{0}$. 
Statement $\qut{U}{\bar{q}}$ means perform a unitary $U$ on the quantum register $\bar{q}$. Statement $\qutp{\theta}{\sigma}{\bar{q}}$ gives a special \emph{parameterized} form of unitary---\emph{density operator simulation}---with $\sigma$ a density operator and $\theta$ selected from $\bm{\theta}$. 
Statement $P_1(\bm{\theta}); P_2(\bm{\theta})$ means first executes $P_1(\bm{\theta})$, and when $P_1(\bm{\theta})$ terminates, it executes $P_2(\bm{\theta})$. 
Statement $\qif (\Box m\cdot M[\bar{q}] = m \to P_m(\bm{\theta})) \qfi$ means performs a measurement $M = \lbrace M_m\rbrace$ on $\bar{q}$ and then a subprogram $P_m(\bm{\theta})$ will be performed upon the outcome $m$ of the measurement. 
In $\qqwhile{\bar{q}}{P(\bm{\theta})}$, a binary measurement $M = \lbrace M_0, M_1\rbrace$ is performed, if the measurement outcome is $0$, then the program terminates; otherwise, the program executes the loop body $P(\bm{\theta})$ and continues the loop, potentially for an arbitrary number of rounds. 

\begin{remark}\label{rmk:1}
   We provide some remarks on the above syntax.
    \begin{itemize} [leftmargin=3mm]
        \item We add a statement $\qinitd{\sigma}{\bar{q}}$ as a more general initialization that sets the state of the quantum register $\bar{q}$ to be a representable density operator $\sigma$, where the ``representable'' means the density operator can be generated by a short parameterized quantum \textbf{while}-program $P$ without parameters and while-loop statement.
        \item We use $\qutp{\theta}{\sigma}{\bar{q}}$ to describe a generally parameterized unitary applied on $\bar{q}$. For any unitary $U$, there is a Hermitian operator $H$ such that $U = e^{-iH}$ and for any Hermitian operator $H$, $e^{-iH}$ (the quantum simulation of Hamiltonian $H$) is also a unitary.
        The parameterization we have chosen, i.e., density operator simulation ($e^{-i\theta \sigma}$), can also express general Hamiltonian simulation.
        \footnote{For any $e^{-i\theta H}$, we define a density operator $\sigma_H = (H - \mu I)/\tr(H - \mu I)$, where $\mu$ is the ground eigenvalue of $H$, $I$ the identity and $H\neq \mu I$. Then $e^{-i \theta H}$ is the same as $e^{-i\theta' \sigma_H}$ where $\theta'=\tr(H - \mu I)\theta$, since  
        $e^{-i\theta H}\rho e^{i\theta H} = e^{-i(\tr(H - \mu I) \theta) \sigma_{H}} \rho e^{i(\tr(H - \mu I) \theta) \sigma_{H}}, \forall \rho.$ }
        \item Our parameterization allows expressing many commonly used parameterized quantum gates, such as Pauli rotation gates and two-qubit coupling gates which are universal and can be reliably implemented in near-term quantum machines\footnote{Single-qubit Pauli rotation gates are given in the following form $R_{\Delta}(\theta):= \exp(\frac{-i\theta}{2}\Delta), \Delta\in \{X,Y,Z\}$. One can also extend Pauli rotations to multiple qubits. For example, consider two-qubit coupling gates $\{R_{\Delta\otimes \Delta}:=\exp(\frac{-i\theta}{2}\Delta\otimes \Delta)\}_{\Delta \in \{X,Y,Z\}}$. Note that these two-qubit gates can generate entanglement between two qubits. Combined with single-qubit rotations, they form a \emph{universal gate set} for quantum computation. 
        }.
    \end{itemize}
\end{remark}

\subsection{Denotational Semantics}\label{sec:semantics}
Following the semantics of quantum \textbf{while}-programs~\cite{ying2016foundations}, the denotational semantics of parameterized quantum \textbf{while}-programs can be defined.

\begin{definition}[Structural Representation of Denotational Semantics~\cite{ying2016foundations}]  
    Let $\cH_{all}$ denote the tensor product of the state spaces of all quantum variables and $\rho \in \cD(\cH_{all})$ indicate the (global) state of quantum variables.
    The denotational semantics of a parameterized quantum \textbf{while}-program $P(\bm{\theta})$ is a superoperator $\sem{P(\bm{\theta})}:\cD(\cH_{all})\to\cD(\cH_{all})$ inductively defined as:
    \begin{itemize}
        \item $\sem{\qskip}(\rho) = \rho$;
        \item $\sem{\qinit{q}}(\rho) = \sum_n\ketbra[q]{0}{n}\rho\ketbra[q]{n}{0}$;
        \item $\sem{\qut{U}{\bar{q}}}(\rho) = U\rho U^{\dagger}$;
        \item $\sem{\qutp{\theta}{\sigma}{\bar{q}}}(\rho) = e^{-i\theta\sigma}\rho e^{i\theta\sigma}$;
        \item $\sem{P_1(\bm{\theta)};P_2(\bm{\theta})}(\rho) = \sem{P_2(\bm{\theta})}(\sem{P_1(\bm{\theta})}(\rho))$;
        \item $\sem{\qif (\Box m\cdot M[\bar{q}] = m \to P_m(\bm{\theta})) \qfi}(\rho) = \sum_m \sem{P_m(\bm{\theta})}(\cE_m(\rho))$;
        \item $\sem{\qqwhile{\bar{q}}{P(\bm{\theta})}}(\rho) = \bigsqcup_{n=0}^{\infty}\sum_{k=0}^{n}\cE_0\circ(\sem{P(\bm{\theta})}\circ\cE_1)^k(\rho)$,
    \end{itemize}
    where $\{\ket{n}_q\}$ is an orthonormal basis of state space $\cH_q$ of variable $q$, $\cE_{m}: \rho \mapsto M_m\rho M_m^{\dagger}$ are defined for each measurement $M = \lbrace M_m\rbrace$ in $P(\bm{\theta})$ and $\bigsqcup$ stands for the least upper bound in the CPO of partial density operators with the L\"{o}wner order $\sqsubseteq$\footnote{The L\"owner order $\sqsubseteq$ is defined as $A \sqsubseteq B$ if $B-A$ is positive semidefinite.} (see ~\cite[Lemma 3.3.2]{ying2016foundations}).
\end{definition}

For a quantum program $P$, we define $var(P)$ to be the set of quantum variables $q\in \qvar$ appearing in a program $P$, and let
$\cH_{P} = \bigotimes_{q\in var(p)} \cH_q.$
When dealing with the semantics of a program $P(\bm{\theta})$, we only consider the states on $\cH_{P(\bm{\theta})}$, that is using $\rho\in \cH_{P(\bm{\theta})}$ to represent a product state $\rho\otimes \rho_0 \in \cDa$, where $ \rho_0\in \cD(\cH_{\qvar\setminus var(P(\bm{\theta}))})$.
When the dimension of $\cH_{P(\bm{\theta})}$ is finite, we have that $\cD(\cH_{P(\bm{\theta})})$ is a compact set, thus
\begin{equation}\label{semantics-series}
\begin{aligned}
    \sem{\qqwhile{\bar{q}}{P(\bm{\theta})}}(\rho)
    & ={} \sum_{k=0}^{\infty}\cE_0\circ(\sem{P(\bm{\theta})}\circ\cE_1)^k(\rho).
\end{aligned}
\end{equation}
Note further that a semantic mapping $\sem{\cdot}:\cDa\to\cDa$ defined on $\cDa$ can be used as a mapping on the set of linear operators $\cL(\cH_{all})$ by linear extension. 

\subsection{Expectation Functions and Differentiability}
\label{subsec:expectation}

The output of a quantum program is often regarded as the expectation of an observable obtained by measurements after its execution. 
We define the \emph{expectation functions} to capture the output of quantum programs, which is similar to the observable semantics introduced by~\cite{Zhu2020Differentiable}.

\begin{definition}
    [Expectation Function]
    For a parameterized quantum \textbf{while}-program $P(\bm{\theta})$ with parameter $\bm{\theta}\in\mathbb{R}^k$, an initial state $\rho\in \cDa$ and an observable $O$ on $\cH_{all}$, the expectation function $f: \mathbb{R}^k \to \mathbb{R}\cup \lbrace\pm \infty\rbrace$ that maps the parameter $\bm{\theta}$ to the output expectation is defined by
    \begin{equation}
        f(\bm{\theta}) = \tr(O\sem{P(\bm{\theta})}(\rho)). 
    \end{equation}
\end{definition}

For any $1\leq j\leq k$, the partial derivatives $\frac{\partial f}{\partial\theta_j}$ of expectation function $f$ with respect to parameter $\theta_j$ can be defined in the standard way.
Their existence in the finite-dimensional case is guaranteed by the following: 
\begin{theorem}[Differentiability] \label{thm:diff-exist}
    For a parameterized quantum \textbf{while}-program $P(\bm{\theta})$ with parameter $\bm{\theta}$ $= (\theta_1, \theta_2,$ $ \ldots, \theta_k)$ $\in\mathbb{R}^k$, an initial state $\rho\in \cDa$ and an observable $O$ on $\cH_{all}$, if
    $\cH_{P(\bm{\theta})}$ is finite-dimensional, then $\frac{\partial f}{\partial\theta_i}$ exists.
\end{theorem}
\begin{proof}
    This is a corollary of Theorem~\ref{thm:exact}.
\end{proof}

As shown in the following example, however, it is possible that the expectation function $f$ is non-differentiable when $\cH_{P(\bm{\theta})}$ is an infinite-dimensional space.

\begin{example}
    [Non-differentiable Infinite-Dimensional quantum program]\label{eg:non-diff}
    Let $q, r$ be two qubits with state space $\cH_q = \cH_r = \Span\{\ket{0},\ket{1}\}$,
    $t_1, t_2$ be quantum variables with state space $\cH_\infty = \lbrace \sum_{n=-\infty}^{\infty}\alpha_n\ket{n}: \alpha_n\in \mathbb{C}, \sum_n\abs{\alpha_n}^2<\infty\rbrace$, and $\theta \in \mathbb{R}$ be a parameter.
    Consider the following parameterized quantum program $P(\theta)$:{\small\setlength{\lineskip}{0pt}
    \begin{align*}
        P(\theta) \equiv{}&  \qinit{q}; \qinit{t_1}; & C(\theta) \equiv{}& q \coloneqq 0;t_1 \coloneqq 0;\\
        &\qwhile M[q] = 1 \qdo && \qwhile q \neq 1 \qdo \\
        &\quad \qut{R}{t_1}; &&\quad t_1\coloneqq t_1+1;\\
        &\quad
            \tikz[baseline]{\node[anchor=base,fill=black!5,rounded corners=4pt,pin={[overlay,pin edge={overlay,black!5,thick},draw=black!40,text=black!40,rounded corners=4pt,thick]left:\rotatebox{90}{$\qutp{2^{t_1}\theta}{\ketbra{+}{+}}{r}$}}]{$\begin{aligned}
                &\qinit{r};\qinit{t_2}; \\
                &\qwhile M[q] = 1 \qdo \\
                &\quad \qutp{\theta}{\ketbra{+}{+}}{r};\qut{R}{t_2};\\
                &\quad \qut{EX}{t_1,t_2,q} \\
                &\qod; \\
            \end{aligned}$}} &&\quad \tikz[baseline]{\node[anchor=base,fill=black!5,rounded corners=4pt,pin={[overlay,pin edge={overlay,black!5,thick},draw=black!40,text=black!40,rounded corners=4pt,thick]left:\rotatebox{90}{$r\coloneqq \sin(2^{t_1}\theta)+1$}}]{$\begin{aligned}
                &r\coloneqq 1;t_2\coloneqq 0; \\
                &\qwhile t_2 \neq t_1 \qdo \\
                &\quad r\coloneqq r*2; t_2\coloneqq t_2+1 \\
                &\qod; \\
                &r\coloneqq \sin(r * \theta)+1;
            \end{aligned}\vphantom{\begin{aligned}
                &\qinit{r};\qinit{t_1}; \\
                &\qwhile M[q] = 1 \qdo \\
                &\quad \qutp{\theta}{\braket{+}{+}}{r};\qut{R}{t_2};\\
                &\quad \qut{EX}{t_1,t_2,q} \\
                &\qod; \\
            \end{aligned}}$}}\\
        &\quad \qut{H}{q}; &&\quad q\coloneqq 0\oplus_{\frac{1}{2}}1;\\
        &\qod && \qod
    \end{align*}}
    where:
    \begin{itemize}
        \item $M = \lbrace M_0 = \ketbra{1}{1}, M_1 = \ketbra{0}{0}\rbrace$ is the measurement on qubit $q$ in the computational basis;
        \item $R = \sum_{j}\ketbra{j+1}{j}$ is the right-translation operator on $t_1$; and
        \item $EX = \sum_{2^j=k}\ketbra{jk}{jk}\otimes X + \sum_{2^j\neq k}\ketbra{jk}{jk}\otimes I$ is a unitary that performs $X$ operation on $q$ if $t_1$ and $t_2$ is in state $\ket{j}$ and $\ket{k}$, respectively, and $k = 2^j$ for any $j,k\in \mathbb{Z}$.
    \end{itemize}
    For an initial state $\rho = \ketbra[q]{0}{0}\otimes\ketbra[r]{0}{0}\otimes\ketbra[t_1]{0}{0}\otimes\ketbra[t_2]{0}{0}$ and an observable $O = 2\ketbra[r]{\psi}{\psi}$ with $\ket{\psi} =  (\ket{0} - i\ket{1})/\sqrt{2}$, a calculation using (\ref{semantics-series}) yields the expectation function of $P(\theta)$:
    \begin{align*}
        f(\theta) &= \sum_{k=1}^{\infty} \frac{1}{2^k}\tr\left(2\ketbra[r]{\psi}{\psi} e^{-i2^k \theta \ketbra[r]{+}{+}}\ketbra[r]{0}{0} e^{i2^k \theta \ketbra[r]{+}{+}}\right)
        = \sum_{k=1}^{\infty}\frac{1+\sin(2^k\theta)}{2^k} = 1+\sum_{k=1}^{\infty}\frac{1}{2^k}\sin(2^k\theta), 
    \end{align*}
    which is well-defined. 
    However, {\color{red}$f$ is non-differentiable everywhere} due to Weierstrass's non-differentiable function~\cite[Theorem 1.31]{hardy1916weierstrass}:
    the function
    $S(x) = \sum_{n=0}^{\infty}a^n\sin(b^nx)$
    converges uniformly on $\mathbb{R}$, which implies $S$ is continuous on $\mathbb{R}$, but nowhere differentiable for any $0<a<1,ab\geq 1$.

    The probabilistic program $C(\theta)$ is a counterpart of $P(\theta)$ for illustration, where $q\coloneqq 0\oplus_{\frac{1}{2}}1$ assigns $0$ to $q$ with probability $\frac{1}{2}$ and $1$ to $q$ otherwise.
    The boxed commands assigns $\sin(2^{t_1}\theta) + 1$ to $r$, thus we can see that
    the expectation of variable $r$ after runs $C(\theta)$ is also $f(\theta)$ and non-differentiable everywhere.
\end{example}

\section{Automatic Differentiation for Unbounded Quantum Loops}\label{sec:auto_diff}
In this section, we develop the AD technique for parameterized quantum \textbf{while}-programs to overcome the major difficulty of finding analytical derivatives of unbounded loops.

\subsection{Differentiation on a Single-Occurrence of Parameter}\label{subsec-commutator}
Our first contribution is a new DSOP technique, called the commutator-form rule.
\citet{PhysRevLett.118.150503} and \citet{mitarai2018quantum} first proposed a derivative formula for Pauli rotations, which is named by \citet{schuld2019evaluating} as the ``parameter-shift rule'' to handle the case of  $U(\theta) = e^{-i\theta H}$ with $H$ having at most two distinct eigenvalues.

Our commutator-form rule is designed to be applicable to $e^{-i\theta H}$ for general $H$.
Technically, it was inspired by a few existing works~\cite{mitarai2018quantum, beer2020training, Lloyd_2014} which leverage the commutator form for various purposes. 
We also note some recent independent developments~\cite{wierichs2021general, izmaylov2021analytic, kyriienko2021generalized} of variants of the parameter-shift rules to handle more general $e^{-i\theta H}$. 
However, our rule is based on a very different technique, which could be of independent interest by itself. 
Precisely, 

\begin{lemma}\label{lem:a4}
    Let $\cH_1, \cH_2, \cH_3$ be Hilbert spaces and $\cE_1:\cD(\cH_1)$ $\to \cD(\cH_2), \cE_2:\cD(\cH_2) \to \cD(\cH_3)$ be superoperators. For any Hermitian operator $H$ on $\cH_2$ and $\theta\in\mathbb{R}$, we define
    $\cE_{H,\theta}(\rho) = e^{-i\theta H}\rho e^{i\theta H}$ for all $\rho\in\cD(\cH_2)$. Then for any density operator $\rho$ on $\cH_1$: 
    \[ \frac{\mathrm{d}}{\mathrm{d} \theta}(\cE_2\circ\cE_{H, \theta}\circ\cE_1(\rho)) = \cE_2\circ \cE_{H,\theta}(-i[H, \cE_1(\rho)]),\]
    where commutator $[\cdot,\cdot]$ is defined as follows: $[A,B]=AB-BA$ for any operators $A$ and $B$.
\end{lemma}
\begin{proof}
    See Appendix~\ref{prf:a4}.
\end{proof}

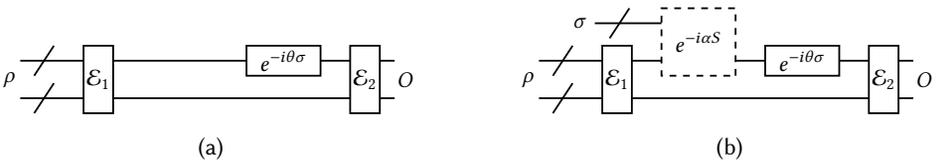
\begin{figure}[ht]
  \centering
  \vspace{-12pt}
  \begin{subfigure}[b]{.49\linewidth}
  \centering
      \scalebox{0.85}{
\begin{tikzpicture}[scale=1.10000,x=1pt,y=1pt,thick]
\filldraw[color=white] (0.000000, -7.500000) rectangle (152.000000, 37.500000);
\draw[color=black] (0.000000,15.000000) -- (152.000000,15.000000);
\draw[color=black] (0.000000,0.000000) -- (152.000000,0.000000);
\filldraw[color=white,fill=white] (0.000000,-3.750000) rectangle (-4.000000,18.750000);
\draw[color=black] (2.000000,7.500000) node[left] {$\rho$};
\draw[color=black] (150.000000,7.500000) node[right] {$O$};
\draw (6.000000, 9.000000) -- (14.000000, 21.000000);
\draw (6.000000, -6.000000) -- (14.000000, 6.000000);
\draw (32.000000,15.000000) -- (32.000000,0.000000);
\begin{scope}
\draw[fill=white] (32.000000, 7.500000) +(-45.000000:8.485281pt and 19.091883pt) -- +(45.000000:8.485281pt and 19.091883pt) -- +(135.000000:8.485281pt and 19.091883pt) -- +(225.000000:8.485281pt and 19.091883pt) -- cycle;
\clip (32.000000, 7.500000) +(-45.000000:8.485281pt and 19.091883pt) -- +(45.000000:8.485281pt and 19.091883pt) -- +(135.000000:8.485281pt and 19.091883pt) -- +(225.000000:8.485281pt and 19.091883pt) -- cycle;
\draw (32.000000, 7.500000) node {$\cE_1$};
\end{scope}
\begin{scope}[xshift=42]
\draw[fill=white] (65.000000, 15.000000) +(-45.000000:21.213203pt and 8.485281pt) -- +(45.000000:21.213203pt and 8.485281pt) -- +(135.000000:21.213203pt and 8.485281pt) -- +(225.000000:21.213203pt and 8.485281pt) -- cycle;
\clip (65.000000, 15.000000) +(-45.000000:21.213203pt and 8.485281pt) -- +(45.000000:21.213203pt and 8.485281pt) -- +(135.000000:21.213203pt and 8.485281pt) -- +(225.000000:21.213203pt and 8.485281pt) -- cycle;
\draw (65.000000, 15.000000) node {$e^{-i\theta \sigma}$};
\end{scope}
\draw (140.000000,15.000000) -- (140.000000,0.000000);
\begin{scope}
\draw[fill=white] (140.000000, 7.500000) +(-45.000000:8.485281pt and 19.091883pt) -- +(45.000000:8.485281pt and 19.091883pt) -- +(135.000000:8.485281pt and 19.091883pt) -- +(225.000000:8.485281pt and 19.091883pt) -- cycle;
\clip (140.000000, 7.500000) +(-45.000000:8.485281pt and 19.091883pt) -- +(45.000000:8.485281pt and 19.091883pt) -- +(135.000000:8.485281pt and 19.091883pt) -- +(225.000000:8.485281pt and 19.091883pt) -- cycle;
\draw (140.000000, 7.500000) node {$\cE_2$};
\end{scope}
\end{tikzpicture}}
      \caption{}\label{fig:ca}
  \end{subfigure}
  \begin{subfigure}[b]{.49\linewidth}
  \centering
      \scalebox{0.85}{
\begin{tikzpicture}[scale=1.10000,x=1pt,y=1pt,thick]
\filldraw[color=white] (0.000000, -7.500000) rectangle (152.000000, 37.500000);
\draw[color=black] (0.000000,15.000000) -- (152.000000,15.000000);
\draw[color=black] (0.000000,0.000000) -- (152.000000,0.000000);
\filldraw[color=white,fill=white] (0.000000,-3.750000) rectangle (-4.000000,18.750000);
\draw[color=black] (2.000000,7.500000) node[left] {$\rho$};
\draw[color=black] (150.000000,7.500000) node[right] {$O$};
\draw (6.000000, 9.000000) -- (14.000000, 21.000000);
\draw (6.000000, -6.000000) -- (14.000000, 6.000000);
\draw (32.000000,15.000000) -- (32.000000,0.000000);
\begin{scope}
\draw[fill=white] (32.000000, 7.500000) +(-45.000000:8.485281pt and 19.091883pt) -- +(45.000000:8.485281pt and 19.091883pt) -- +(135.000000:8.485281pt and 19.091883pt) -- +(225.000000:8.485281pt and 19.091883pt) -- cycle;
\clip (32.000000, 7.500000) +(-45.000000:8.485281pt and 19.091883pt) -- +(45.000000:8.485281pt and 19.091883pt) -- +(135.000000:8.485281pt and 19.091883pt) -- +(225.000000:8.485281pt and 19.091883pt) -- cycle;
\draw (32.000000, 7.500000) node {$\cE_1$};
\end{scope}
\begin{scope}[xshift=42]
\draw[fill=white] (65.000000, 15.000000) +(-45.000000:21.213203pt and 8.485281pt) -- +(45.000000:21.213203pt and 8.485281pt) -- +(135.000000:21.213203pt and 8.485281pt) -- +(225.000000:21.213203pt and 8.485281pt) -- cycle;
\clip (65.000000, 15.000000) +(-45.000000:21.213203pt and 8.485281pt) -- +(45.000000:21.213203pt and 8.485281pt) -- +(135.000000:21.213203pt and 8.485281pt) -- +(225.000000:21.213203pt and 8.485281pt) -- cycle;
\draw (65.000000, 15.000000) node {$e^{-i\theta \sigma}$};
\end{scope}
\begin{scope}[xshift=-42]
\draw[color=black] (65,30.000000) -- (107.000000,30.000000);
\draw[color=black] (65.000000,30.000000) node[left] {$\sigma$};
\draw (107.000000,30.000000) -- (107.000000,15.000000);
\draw (71.000000, 24.000000) -- (79.000000, 36.000000);
\draw[fill=white,dashed] (107.000000, 22.500000) +(-45.000000:21.213203pt and 19.091883pt) -- +(45.000000:21.213203pt and 19.091883pt) -- +(135.000000:21.213203pt and 19.091883pt) -- +(225.000000:21.213203pt and 19.091883pt) -- cycle;
\clip (107.000000, 22.500000) +(-45.000000:21.213203pt and 19.091883pt) -- +(45.000000:21.213203pt and 19.091883pt) -- +(135.000000:21.213203pt and 19.091883pt) -- +(225.000000:21.213203pt and 19.091883pt) -- cycle;
\draw (107.000000, 22.500000) node {$e^{-i\alpha S}$};
\end{scope}
\draw (140.000000,15.000000) -- (140.000000,0.000000);
\begin{scope}
\draw[fill=white] (140.000000, 7.500000) +(-45.000000:8.485281pt and 19.091883pt) -- +(45.000000:8.485281pt and 19.091883pt) -- +(135.000000:8.485281pt and 19.091883pt) -- +(225.000000:8.485281pt and 19.091883pt) -- cycle;
\clip (140.000000, 7.500000) +(-45.000000:8.485281pt and 19.091883pt) -- +(45.000000:8.485281pt and 19.091883pt) -- +(135.000000:8.485281pt and 19.091883pt) -- +(225.000000:8.485281pt and 19.091883pt) -- cycle;
\draw (140.000000, 7.500000) node {$\cE_2$};
\end{scope}
\end{tikzpicture}}
      \caption{}\label{fig:cb}
  \end{subfigure}
  \vspace{-5pt}
  \caption{Introduce the commutator with a similar circuit.}\label{fig:commutator}
  \vspace{-12pt}
  \Description{commutator}
\end{figure}

\subsubsection{Commutator-Form Rule.}
The way of introducing commutators is visualized in \figref{fig:commutator}. We define 
\begin{equation}
    f(\theta) = \tr(O \cE_2(e^{-i\theta\sigma}\cE_1(\rho) e^{i\theta\sigma}))
\end{equation}
as the expectation function in \figref{fig:ca} and 
\begin{equation}\label{eq:swap_g}
    g(\theta;\alpha) = \tr(O\cE_2(e^{-i\theta\sigma}e^{-i\alpha S}\cE_1(\rho)\otimes \sigma e^{i\alpha S}e^{i\theta\sigma}))
\end{equation}
as the expectation function in \figref{fig:cb}, where $S$ is the SWAP operator\footnote{The SWAP operator $S$ on a space $\cH\otimes\cH$ is defined as $S(\ket{a}\otimes\ket{b}) = \ket{b}\otimes\ket{a}$ for any $\ket{a},\ket{b}\in\cH$ that swaps the states of two systems.}. With Lemma~\ref{lem:a4}, we have
\[\frac{\mathrm{d}}{\mathrm{d} \theta}f(\theta) = \tr\left(O \cE_2\left(e^{-i\theta\sigma}(-i[\sigma\otimes I, \cE_1(\rho)]) e^{i\theta\sigma}\right)\right).\]
Inspired by the trick of applying unitary transformation $e^{-i\theta\rho}$ of any density operator $\rho$ in quantum principal component analysis~\cite{Lloyd_2014}, we find that for any $\alpha\in(0,\frac{\pi}{2})$:
\begin{equation}\label{eq:com}
    \begin{aligned}
        \llap{(\text{commutator-form rule})\qquad} \frac{\mathrm{d}}{\mathrm{d} \theta}f(\theta) = \frac{1}{\sin(2\alpha)}\left(g(\theta;\alpha) - g(\theta;-\alpha)\right).
    \end{aligned}
\end{equation}

\subsection{Code Transformation for Unbounded Loops}\label{sec:unbounded}

Our AD scheme (Fig.~\ref{fig:running_example_programs}) could leverage any DSOP technique (both the commutator-form rule and the parameter-shift rule). 
Due to the space limit, we illustrate the code-transformation based on the commutator-form rule and leave the details based on the parameter-shift rule in Appendix~\ref{subsec-phase}. 
As mentioned in the introduction, we construct the code transformation with respect to parameter $\theta$ (denoted $T_{\theta}$ in \figref{fig:trans_rules})  in the pure quantum fashion so that differential programs can be written in the same syntax as the original one. 
However, it is not hard to see the correctness proof directly carries over to the hybrid quantum-classical case.  

\begin{definition}[Code Transformation]\label{dfn:derivative}
    For a parameterized quantum \textbf{while}-program $P(\bm{\theta})$ with parameter $\bm{\theta}\in\mathbb{R}^k$, its differential program with respect to $\theta$ is defined as a parameterized quantum \textbf{while}-program $\qd{\theta}(P(\bm{\theta}))$:
        \[ \qd{\theta}(P(\bm{\theta})) = \dinit;T_{\theta}(P(\bm{\theta})), \]
    with $\dinit$ defined as follows and $T_{\theta}, C, GP$ given in \figref{fig:trans_rules}, 
    \begin{align*}
        \dinit \equiv{}& \qinit{q_1};\qinit{q_2};\qinit{q_c};\qut{C}{q_c};\qut{GP}{q_c,q_2}.
    \end{align*}
\end{definition}

\begin{figure}[t]
    \centering
    \small
    \vspace{-4pt}
    \begin{align*}
        T_{\theta}(\qskip) \equiv{} \qskip \qquad\qquad
        T_{\theta}(\qinit{q}) &\equiv{} \qinit{q} \qquad\qquad
        T_{\theta}(\qut{U}{\bar{q}}) \equiv{} \qut{U}{\bar{q}} \\
        T_{\theta}(\qutp{\theta'}{\sigma}{\bar{q}}) &\equiv{} \qutp{\theta'}{\sigma}{\bar{q}} \quad (\theta' \neq \theta) \\
        T_{\theta}(P_1(\bm{\theta});P_2(\bm{\theta})) &\equiv{} T_{\theta}(P_1(\bm{\theta}));T_{\theta}(P_2(\bm{\theta})) \\
        T_{\theta}(\qif (\Box m\cdot M[\bar{q}] = m \to P_m(\bm{\theta})) \qfi) &\equiv{} \qif (\Box m\cdot M[\bar{q}] = m \to T_{\theta}(P_m(\bm{\theta}))) \qfi \\
        T_{\theta}(\qwhile M[\bar{q}] = 1 \qdo P(\bm{\theta}) \qod) &\equiv{} \qwhile M[\bar{q}] = 1 \qdo T_{\theta}(P(\bm{\theta})) \qod
    \end{align*}
    \vspace{-13pt}
    \begin{align*}
        \begin{aligned}
        (*)\  T_{\theta}(\qutp{\theta}{\sigma}{\bar{q}}) \equiv{} \vphantom{\qif (M_{q_1,q_2}[q_1, q_2] = 0 \to{} \qut{C}{q_c};\qut{GP}{q_c,q_2}} \\
        \vphantom{\Box{} = 1 \to{} \qut{X}{q_1}; \qinitd{\textstyle\frac{I}{2}}{q_2}; \qinitd{\sigma}{\bar{q}'};} \\
        \vphantom{\phantom{\,\, = 1 \to{}} \qut{AS}{q_2,\bar{q},\bar{q}'}} \\
        \vphantom{\Box{} = 2 \to{} \qskip ) \qfi;{\color{black}\qutp{\theta}{\sigma}{\bar{q}}}}
        \end{aligned}&{\color{blue}\begin{aligned} \vphantom{T_{\theta}(\qutp{\theta}{\sigma}{\bar{q}}) \equiv{}}
            \qif (M_{q_1,q_2}[q_1, q_2] &= 0 \to{} \qut{C}{q_c};\qut{GP}{q_c,q_2} \\
            \Box{} &= 1 \to{} \qut{X}{q_1}; \qinitd{\textstyle\frac{I}{2}}{q_2}; \qinitd{\sigma}{\bar{q}'};\\
            & \phantom{\,\, = 1 \to{}} \qut{AS}{q_2,\bar{q},\bar{q}'}\\
            \Box{} &= 2 \to{} \qskip ) \qfi;{\color{black}\qutp{\theta}{\sigma}{\bar{q}}}
            \end{aligned}}
    \end{align*}
    \vspace{-10pt}
    \caption{Code transformation rules with respect to parameter $\theta$. The blue part of $(*)$ refers to the EUL part of \figref{fig:running_example_programs}, where $q_1, q_2$ are two qubit variables, $q_c$ is a quantum variable with state space $\cH_c = \Span\lbrace \ket{n}:n \in \mathbb{Z}\rbrace$, $M_{q_1,q_2} = \lbrace M_0 = \ketbra{00}{00}, M_1 = \ketbra{01}{01},$ $ M_2=\ketbra{10}{10}+\ketbra{11}{11}\rbrace$,
    $C = \sum_{j=0}^{\infty}\ketbra{j+1}{j}$ is the right-translation operator,
    $GP = \sum_{j=1}^{\infty}\ketbra{j}{j}\otimes R_y(2\arcsin(\sqrt{b_j}))$ and $b_j= \mu(j)/(1-\sum_{k=1}^{j-1}\mu(k))$, $AS = \ketbra{0}{0}\otimes e^{-i\frac{\pi}{4}S_{\bar{q}, \bar{q}'}} + \ketbra{1}{1}\otimes e^{i\frac{\pi}{4}S_{\bar{q}, \bar{q}'}}$ and  $S_{\bar{q}, \bar{q}'}$ is the SWAP operator between $\cH_{\bar{q}}$ and $\cH_{\bar{q}'}$.}
    \label{fig:trans_rules}
    \Description{Code transformation rules}
\end{figure}

The connection to the hybrid quantum-classical case can
be seen as follows:
the blue part of rule $(*)$ in~\figref{fig:trans_rules} corresponds to the EUL commands of our AD scheme in~\figref{fig:running_example_programs}.
Our construction $T_{\theta}$ guarantees that no entanglement will be created between $q_1, q_2, q_c$ and other quantum variables in the differential program, which means that there is only classical correlation rather than quantum correlation (see~\cite[Section VI. Bipartite Entanglement]{RevModPhys.81.865}) between $q_1, q_2, q_c$ and other quantum variables.
Moreover, $q_c$ will be always in its basis states $\{\ket{n}:n\in \mathbb{Z}\}$ and $q_1,q_2$ are two qubit variables.
Therefore, $q_1, q_2, q_c$ can be separated from the differential program and simulated efficiently by a classical computer.

To establish the correctness of our code transformation, we develop the following lemma about finite-dimensional quantum programs in light of Example~\ref{eg:non-diff}.

\begin{lemma}\label{lem:a9}
  Consider a quantum loop $P \equiv \qqwhile{\bar{q}}{Q}$.  
  Assume that the state space $\cH_P$ is finite-dimensional. We define superoperators $\cE_i: \cD(\cH_P) \to \cD(\cH_P)$ by $\cE_i(\rho) = M_i \rho M_i^{\dagger}$, $i=0,1$ and  $\cE: \cD(\cH_P) \to \cD(\cH_P)$ by $\cE(\rho)= \sem{Q}(\rho)$. Then for any $\epsilon \in (0, 1)$, there exists $N = N_{\epsilon} > 0$ such that $\forall n \in \mathbb{N}, \forall \rho \in \cD(\cH_P),$
  \[ \tr(\cE_0\circ(\cE\circ\cE_1)^n(\rho)) \leq \epsilon^{\lfloor \frac{n}{N}\rfloor}\tr(\rho).\]
\end{lemma}
\begin{proof}
  See Appendix~\ref{prf:a9}.
\end{proof}
The above lemma ensures the probability that the finite-dimensional program runs out of the loop has an exponential decay on the number of loop iterations. 
This observation leads to an exponential decay of partial derivatives for corresponding occurrences of the parameter, which in turn guarantees the existence of the derivative  and the validity of exchanging the order between the infinite summation and the derivation. 

\begin{theorem}
    [Correctness of Code Transformation]\label{thm:exact}
    Given a parameterized quantum \textbf{while}-program $P(\bm{\theta})$ with parameter $\bm{\theta}\in\mathbb{R}^k$ and finite-dimensional state space $\cH_{P(\bm{\theta})}$, an observable $O$ and an input state $\rho$. Let $f(\bm{\theta})$
    be the expectation function of $P(\bm{\theta})$ with respect to $\rho$ and $O$. Then the partial derivative of $f$ with respect to $\theta$ is
    \begin{equation}\label{eq:main}
         \frac{\partial }{\partial \theta}f(\bm{\theta}) = \tr\left((O_d\otimes O)\sem{\qd{\theta}(P(\bm{\theta}))}(\rho)\right),
    \end{equation}
the expectation function of $\qd{\theta}(P(\bm{\theta}))$ with respect to $\bm{\theta}$, observable $O_d\otimes O$ and input state $\rho$ with 
$O_d = \sum_{j=1}^{\infty}\frac{2}{\mu(j)}\ketbra{j}{j}\otimes \ketbra{1}{1}\otimes Z$ is an observable on $\cH_{q_c}\otimes\cH_{q_1}\otimes \cH_{q_2}$.
\end{theorem}

\begin{proof}
    [Outline of the proof] 
    We can take Fig.~\ref{fig:running_example_programs} as an example to briefly illustrate the outline of the proof, while the full details are deferred to Appendix~\ref{prf:exact}.
\begin{enumerate}[leftmargin=3mm]
    \item Since our AD is performed by inserting commands, the execution branches of $P(\theta)$'s differential program in Fig.~\ref{fig:running_example:c} are the same as $P(\theta)$ in Fig.~\ref{fig:running_example:a}. Thus, we consider each execution path of $P(\theta)$.
    \item For a fixed execution path of $P(\theta)$, its derivative has the form $\partial_{\theta}f_k(\theta) \equiv \frac{1}{2^{k+1}}\sum_{j=1}^k(g_{\theta})^{j-1}\circ \left(\partial_\theta g_{\theta}\right)\circ(g_{\theta})^{k-j}$. 
    The $\partial_{\theta}f_k(\theta)$ corresponds to perform $k$ times differentiation operations in different occurrences of $\theta$. For the same branch of the fixed execution path, $P(\theta)$'s differential program can also perform the same $k$ times differentiation operations with probability $\mu(1),\ldots,\mu(k)$.
    Then $P(\theta)$'s differential program can produce $\partial_{\theta}f_k(\theta)$ by estimation of expectation.
    \item Finally, one adds up all the $\partial_{\theta}f_k(\theta)$ that are produced by $P(\theta)$'s differential program with respect to $P(\theta)$'s execution paths: $\sum_{k=1}^{\infty}\partial_{\theta}f_k(\theta)$, and prove it is uniformly convergent, the main challenging of the proof that relies on the finite-dimensional condition, and equal to $f(\theta)$'s derivative.
\end{enumerate}
\end{proof}

\section{Implementation and Sample Complexity}\label{sec:implement}
In this section, we discuss the implementation of our AD scheme and analyze its efficiency in terms of sample complexity, the number of required samples to estimate gradients.

\subsection{Implementation in A Hybrid Style}\label{sec:hybrid}
There are a few candidates of high-level quantum programming languages which support hybrid quantum-classical programming with classical control flow, e.g.~Microsoft's Q\#~\cite{Svore_2018} and ETH Z\"urich's Silq~\cite{10.1145/3385412.3386007}. 
Since Q\# provides a Python package \code{qsharp} that enables simulation of Q\# programs from regular Python programs, we choose Python and Q\# to implement a parser that transforms parameterized quantum \textbf{while}-programs (with a restricted set of unitaries) to Q\# and implement AD to generate Q\# codes for evaluating gradients.

Our current implementation supports parameterized Pauli rotations and controlled Pauli rotations as follows: 
\begin{align*}
    \left\lbrace \begin{aligned} 
    & R_{\Delta}(\theta) = e^{-i\frac{\theta}{2} \Delta}, e^{-i\frac{\theta}{2}\ketbra{1}{1}\otimes \Delta}, R_{\Delta\otimes\Delta}(\theta) = e^{-i\frac{\theta}{2}\Delta\otimes \Delta} \end{aligned}  : \Delta = X, Y, Z; \theta\in \bm{\theta}\right\rbrace.
\end{align*}
The Pauli rotations and controlled Pauli rotations are internally replaced with their corresponding density operator form, e.g., unitary $e^{-i\frac{\theta}{2}X\otimes X}$ is replaced by $ e^{-i(2\theta)\frac{X\otimes X+I}{4}}$ with $(X\otimes X+I)/4$, a density operator, then we can apply our technique of AD to it and get the derivative with a scale $2$.
 
\subsection{Variance and Sample Complexity} \label{sec:sample}
Our main theorem (Theorem~\ref{thm:exact}) asserts that the desired partial derivative can be expressed by the expectation of observable $O_d\otimes O$ with respect to state $\sem{\qd{\theta}(P(\bm{\theta}))}(\rho)$, which we denote 
$\langle O_d\otimes O\rangle$ for simplicity. 
We denote the \emph{sample complexity} as the number of repetitions to estimate $\langle O_d\otimes O \rangle$ to a given precision $\delta$.
To estimate the sample complexity, we consider the \emph{variance} of observable $O_d\otimes O$: 
$
    \mathrm{Var}(O_d\otimes O) = \left\langle \left(O_d\otimes O - \langle O_d\otimes O\rangle\right)^2\right\rangle = \left\langle O_d^2\otimes O^2\right\rangle - \left\langle O_d\otimes O\right\rangle^2.
$

Inspired by the  ``Occurrence Count for $\theta$'' in~\cite{Zhu2020Differentiable}, 
we introduce two technical notions, i.e., the ``Running Count for $\theta$'' in program $P(\bm{\theta})$, denoted $RC_{\theta}(P(\bm{\theta}))$, as the number of occurrences of $\theta$ in $P(\bm{\theta})$ and the ``Loop Count'' in $P(\bm{\theta})$, denoted $LC(P(\bm{\theta}))$, as the number of while-loop statements in $P(\bm{\theta})$, for upper bounding $\left\langle O_d^2\otimes O^2\right\rangle$.
For formal definitions of $RC_{\theta}(P(\bm{\theta}))$ and $LC(P(\bm{\theta}))$, please refer to Appendix~\ref{appendix:variance}.
We also need a terminating condition of parameterized programs in order to upper bound $\left\langle O_d^2\otimes O^2\right\rangle$.

\begin{definition}
    [Almost Sure Termination~\cite{ying2016foundations}]
    A parameterized quantum \textbf{while}-program $P(\bm{\theta})$ terminates almost surely at $\bm{\theta}$ if $\tr(\sem{P(\bm{\theta})}(\rho)) = \tr(\rho)$ for any $\rho \in \cD(\cH_{P(\bm{\theta})})$.  
\end{definition}

\begin{theorem}\label{thm:bound2}
    In the same setting as in Theorem~\ref{thm:exact} and distribution $\mu:\mathbb{Z}_+\to [0,1]$ satisfies \ref{condition}, if all the \textbf{while}-statements (subprograms) in $P(\bm{\theta})$ terminate almost surely, then $\left\langle O_d^2\otimes O^2\right\rangle$
    is bounded.
    Additionally, if the distribution $\mu:\mathbb{Z}_+\to [0,1]$ satisfies
    \begin{equation}\label{eq:distribution}
        \mu(n) \propto \frac{1}{n\ln^{1+s}(n+e)}  \text{ with constant } s \in (0,1],
    \end{equation}
    we have
    \[\left\langle O_d^2\otimes O^2\right\rangle \in \mathcal{O}\left(M_1^2\ln^{1+s}(M_1+e)+M_1^{3+s}M_2^{2+s}C(M_2)\right)\]
    with $M_1 = RC_{\theta}(P(\bm{\theta}))$, $M_2 = LC(P(\bm{\theta}))$ and $C(M_2)$ is a non-zero function of $M_2$.
\end{theorem}
\begin{proof}
See Appendix~\ref{appendix:variance}.
\end{proof}

\subsubsection{Comparison with \citet{Zhu2020Differentiable}.}
In the case of no unbounded loops in~\cite{Zhu2020Differentiable}, we have  $M_2 = 0$ and the bound given in Theorem~\ref{thm:bound2} becomes
$\mathcal{O}\left(M_1^2\ln^{1+s}(M_1+e)\right)$,
which implies the sample complexity  $\mathcal{O}(M_1^{2}\ln^{1+s}(M_1+e)/\delta^2)$ by Chebyshev's Inequality. This is comparable to the sample complexity  $\mathcal{O}(m^2/\delta^2)$
estimated in~\cite{Zhu2020Differentiable},  where $m \approx RC_{\theta}(P(\bm{\theta}))$ $ = M_1$, as all of the loops considered there are bounded and thus can be unfolded to nested conditional statements.

\subsubsection{Empirical Estimation of the Sample bound.}
The bound in Theorem~\ref{thm:bound2} could, however, be loose in practice, which would cost unnecessary samples. 
To resolve this issue, we develop an \emph{empirical} estimation of the sample bound, which usually leads to tighter bounds in our case studies.  

Our key idea is that one can empirically estimate $\langle O_d^2 \otimes O^2\rangle$ by sampling as we did for $\langle O_d \otimes O\rangle$ so as to get a better empirical bound than analytical ones. 
To that end, one can apply a similar technique in Theorem~\ref{thm:bound2} to bound $\langle O_d^4 \otimes O^4\rangle$ and hence the number of samples required to estimate $\langle O_d^2 \otimes O^2\rangle$.
However, at this time, we can tolerate a much larger additive error $\delta$ since $\langle O_d^2 \otimes O^2\rangle$ could be large itself, which makes $1/\delta^2$ in Chebyshev's Inequality scale nicely. 

\section{Case Studies} \label{sec:case_study}
In this section, we present the case studies to demonstrate the feasibility of our framework, including \emph{parameterized amplitude amplification}, \emph{quantum walk based search algorithm} and \emph{repeat-until-success unitary implementation}.
{
The chosen case studies, all of which contain unbounded quantum loops, are non-trivial and realistic examples from quantum literature.
We don't choose typical variational algorithms, e.g., QAOA~\cite{farhi2014quantum}, VQE~\cite{peruzzo2014variational}, or some variants studied in the previous work of differentiable quantum programming~\cite{Zhu2020Differentiable} since they don't contain unbounded loops.
Similarly, because there is no realistic example yet of nested loops as existing quantum algorithms are far less than classical, we don't artificially construct experiments for nested loops.
On the other hand, our proposed commutator-form rule provides a more concise form than the parameter-shift rule for general Hamiltonian (e.g., Hamiltonian in QAOA~\cite{hadfield2019quantum}) and our inductively defined code transformation can handle nested loops.

\subsubsection{Experiment Workflow.} 
For experiments, our framework provides a unified principled way to identify suitable parameters of parameterized quantum \textbf{while}-programs automatically as follows:\\
\textbf{Given:} A parameterized quantum \textbf{while}-program $P(\bm{\theta}), \bm{\theta}\in \mathbb{R}^k, k\geq 1$, an quantum state $\rho$ as program's input and an observable $O$ defined on $\cH_{P(\bm{\theta})}$.\\
\textbf{Workflow:}\begin{enumerate}
        \item Use the implemented parser in Sect.~\ref{sec:hybrid} to convert the program $P(\bm{\theta})$ to \emph{Q\# functions} that can sample the value and the partial derivatives of expectation function $f(\bm{\theta}) = \tr(O\sem{P(\bm{\theta})}(\rho))$, which is the objective function to optimize. 
        \item Use the empirical estimation of the sample bound developed in Sect.~\ref{sec:sample} to estimate the number of samples needed for sampling the partial derivatives of $f(\bm{\theta})$.
        \item Use a gradient-based optimizer (in our experiments, we choose Adam optimizer~\cite{DBLP:journals/corr/KingmaB14}) to \emph{maximize/minimize} $f(\bm{\theta})$, where the initial value of parameters $\bm{\theta}_0$ is usually randomly given and the gradient of $f(\bm{\theta})$ is estimated by the Q\# functions (all run on the simulator provided by Q\#) in (1) with the number of samples estimated in (2).
    \end{enumerate}
}

In all experiments, the distribution $\mu$ in code transformation for AD is chosen as the distribution in Eq.~\eqref{eq:distribution} with $s = 0.25$.
Our experiments are performed on a desktop computer with Intel(R) Core(TM) i7-9700 CPU @ 3.00GHz $\times$ 8 Processor and 16GB RAM.

\subsection{Parameterized Amplitude Amplification}\label{sec:case_paa}
As the first case study, we show that our framework can be used to obtain a better parameter in the example of parameterized AA as in~\figref{fig:paa} than those given in the existing literature~\cite{andres2020quantum, mizel2009critically} analytically by hand.

\subsubsection{Parameterized AA Program.}
Consider the parameterized AA in a single-qubit system.
Given $p\in (0,1)$, suppose we have a single qubit unitary $A$ such that
$A\ket{0} = \sqrt{1-p}\ket{0} +\sqrt{p}\ket{1}$,
and its inverse $A^{\dagger}$.
State $\ket{1}$ is our target state.
The details of parameterized AA program are listed in~\figref{fig:paa_experiment}, where we put the overview of parameterized AA and its instance $P_1(\theta)$ used in this experiment together.

\begin{figure}[ht]
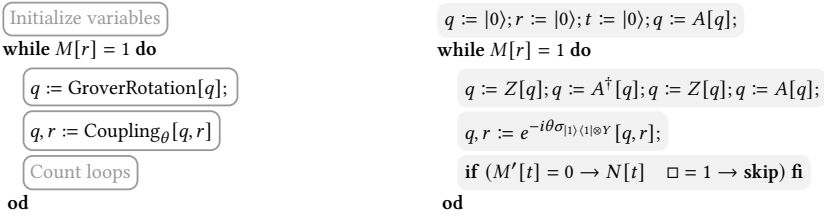

    \vspace{-4pt}
    \centering
    \small
    \hfill
    \begin{subfigure}[b]{0.4\linewidth}
        \centering
        \scalebox{0.85}{$\begin{aligned}
            & \tikz[baseline]{\node[anchor=base,draw=black!40,text=black!40,rounded corners=4pt,thick] (init_code) {Initialize variables\vphantom{$\qinit{q}; \qinit{r}; \qinit{t};\qut{A}{q};$}};}\\[-2pt]
            & \qwhile M[r] = 1 \qdo \\[-2pt]
            & \quad \tikz[baseline]{\node[anchor=base,draw=black!40,rounded corners=4pt,thick] (grover) {$q: = \text{GroverRotation}[q]; \vphantom{\qut{Z}{q}; \qut{A^{\dagger}}{q};\qut{Z}{q};\qut{A}{q};}$};} \\[-2pt]
            & \quad \tikz[baseline]{\node[anchor=base,draw=black!40,rounded corners=4pt,thick] (paau) {$q, r: = \text{Coupling}_{\theta}[q, r]\vphantom{\qutp{\theta}{\sigma_{\ketbra{1}{1}\otimes Y}}{q,r};}$};}  \\[-2pt]
            & \quad \tikz[baseline]{\node[anchor=base,draw=black!40,text=black!40,rounded corners=4pt,thick] (count) {Count loops\vphantom{$\qif (M'[t] = 0 \to N[t]\quad \Box{}= 1 \to \qskip )\qfi$}};}\\[-2pt]
            & \qod
        \end{aligned}$}
        \caption{Overview of parameterized AA in~\figref{fig:paa}.}
    \end{subfigure}
    \hfill
    \begin{subfigure}[b]{0.54\linewidth}
        \centering
        \scalebox{0.85}{$\begin{aligned}
            & \tikz[baseline]{\node[fill=black!5,draw=black!5,thick,text=black,rounded corners=4pt,] (init_code1) {$\qinit{q}; \qinit{r}; \qinit{t};\qut{A}{q};$};} \\[-2pt]
            & \qwhile M[r] = 1 \qdo \\[-2pt]
            & \quad \tikz[baseline]{\node[fill=black!5,draw=black!5,thick,text=black,rounded corners=4pt,] (grover1) {$\qut{Z}{q}; \qut{A^{\dagger}}{q};\qut{Z}{q};\qut{A}{q};$};} \\[-2pt]
            & \quad \tikz[baseline]{\node[fill=black!5,draw=black!5,thick,text=black,rounded corners=4pt,] (paau1) {$\qutp{\theta}{\sigma_{\ketbra{1}{1}\otimes Y}}{q,r};$};}  \\[-2pt]
            & \quad \tikz[baseline]{\node[fill=black!5,draw=black!5,thick,text=black,rounded corners=4pt,] (count1) {$\qif (M'[t] = 0 \to N[t]\quad \Box{}= 1 \to \qskip )\qfi $};} \\[-2pt]
            & \qod
        \end{aligned}$
        }
        \caption{Specific instance of parameterized AA: $P_1(\theta)$.}
        \label{fig:paa_specific}
    \end{subfigure}
    \hfill
    \vspace{-4pt}
    \caption{Parameterized AA used in the experiment.}
    \label{fig:paa_experiment}
\end{figure}

In~\figref{fig:paa_specific}, $q,r$ are qubit variables, measurement $M = \lbrace M_0 = \ketbra{1}{1}, M_1 = \ketbra{0}{0} \rbrace$ and $\sigma_{\ketbra{1}{1}\otimes Y} = (\ketbra{1}{1}\otimes Y + I \otimes I)/4$.
The variable $r$ together with the unitary $e^{-i\theta \sigma_{\ketbra{1}{1}\otimes Y}}$ and measurement $M$ forms a $\frac{\theta}{4}$-measurement in~\cite{andres2020quantum}.
To count the calls of $A$ and $A^{\dagger}$ (the running number of loops), we introduce a block of ``count loops'' that doesn't affect the behavior of parameterized AA, where $t$ is a quantum variable in the space
$\cH_{p} = \Span\left\lbrace \ket{0},\ldots,\ket{4\lfloor 1/\sqrt{p}\rfloor}\right\rbrace$, unitary $N = \sum_{n=0}^{4\lfloor 1/\sqrt{p}\rfloor-1}\ketbra[t]{n+1}{n} + \ketbra[t]{0}{4\lfloor 1/\sqrt{p}\rfloor }$ and measurement 
$M' = \left\lbrace M_0' = \sum_{n=0}^{4\lfloor 1/\sqrt{p}\rfloor-1}\ketbra[t]{n}{n}, M_1' = I - M_0'\right\rbrace$.
With the conditional statement of measurement $M'$, the variable $t$ will remain unchanged once it reaches the state $\ket{4\lfloor 1/\sqrt{p}\rfloor }$.

The program's input can be arbitrary since there are variables' initialization in~\figref{fig:paa_specific}.
The observable we choose is $O_1 = \frac{\sqrt{p}}{4}\sum_{n=1}^{4\lfloor 1/\sqrt{p}\rfloor} n\ketbra[t]{n}{n}$, which expresses the running number of loop iterations and represents the total running time.%
\footnote{This is only an approximation of the total running time since the state of $t$ will always be $\ket{4\lfloor 1/\sqrt{p}\rfloor }$ after $4\lfloor 1/\sqrt{p}\rfloor$ loop iterations.
But this does not matter, because the running number of loop iterations is concentrated below $1/\sqrt{p}$.
}
The scale $\sqrt{p}/4$ is used to normalize the output of $O_1$.
Our \textbf{target} is to \emph{identify $\theta$ so that the expectation of the running time ($O_1$) of parameterized AA in~\figref{fig:paa_specific} is minimized}.

We summarize below the needed configuration in the experiment workflow.
\begin{description}
    \item[\textbf{Given:}] Parameterized AA program $P_1(\theta)$, arbitrary input state $\rho$ and observable $O_1$.
    \item[\textbf{Workflow:}] In (2), samples' number: $5/\sqrt{p} \times 10^3$. In (3), Adam's setting: $\beta_1 = 0.9, \beta_2 = 0.999, \alpha = 0.1$; initial parameter: $\theta = 4 \arccos((1-2\sqrt{p(1-p)})/(1+2\sqrt{p(1-p)}))$ (analytical but sub-optimal value from~\cite{mizel2009critically});
    Goal: \textbf{minimize} the expectation function of observable $O_1$.
\end{description}

For the number of samples, a numerical calculation based on a finer version of Theorem~\ref{thm:bound2} provides $799.72$ as the bound of $\langle O_d^2\otimes O_1^2\rangle$  with $p = 1/100$.
However, applying our empirical estimation, the actual value of $\langle O_d^2\otimes O_1^2\rangle$ would be bounded by 44.26 when $p=1/100$, which leads to the current $5/\sqrt{p}\times 10^3$ bound ($= 5\times 10^4$ when $p =1/100$) with additive error $\delta=0.1$ by Chebyshev's Inequality. 
Please refer to details in Appendix~\ref{sec:details}.

\subsubsection{Results.} We choose $p = 1/10^2,1/15^2,\ldots,1/30^2$ to run this experiment. In Table~\ref{tab:2}, we list the value of $\langle 4O_1\rangle = 4\langle O_1\rangle$, which expresses the (approximate) ratio of the number of loops to $1/\sqrt{p}$, that we find in this experiment (See the column ``Ours''), as well as those in previous works~\cite{andres2020quantum} (see the column ``B'') and \cite{mizel2009critically} (see the column ``C'') for the probability $p$ specified in each row.
For each $p$, a better result (both smaller $\langle 4O_1\rangle$ and smaller variance $\text{Var}(4O_1)$ that implies less fluctuation around the expectation) is found by our experiment. 
Recall that both B and C results are based on analytical forms developed by domain experts. 

Since our goal is to minimize the expectation of $O_1$, we find that the experimental results confirm the feasibility of our framework and also validate the experiment workflow for automatically getting suitable parameters.

\begin{table}[ht]
    \vspace{-5pt}
    \centering
    \scalebox{0.85}{
    \begin{tabular}{cp{30pt}<{\centering}p{36pt}<{\centering}p{36pt}<{\centering}p{36pt}<{\centering}p{36pt}<{\centering}p{36pt}<{\centering}}
        \toprule
        & \multicolumn{2}{c}{Ours} & \multicolumn{2}{c}{B} & \multicolumn{2}{c}{C}\\
        \cline{2-7}$p$ & $\langle 4O_1\rangle$ & $\text{Var}(4O_1)$ & $\langle 4O_1\rangle$ & $\text{Var}(4O_1)$ & $\langle 4O_1\rangle$ & $\text{Var}(4O_1)$\\
        \hline
        $1/10^2$ & $\bm{0.6499}$ & $\bm{0.1184}$ & $1.6884$ & $1.6480$ & $0.6773$ & $0.1872$\\
        \hline
        $1/15^2$ & $\bm{0.6733}$ & $\bm{0.1232}$ & $1.7404$ & $1.6336$ & $0.7018$ & $0.1872$\\
        \hline
        $1/20^2$ & $\bm{0.6885}$ & $\bm{0.1216}$ & $1.7541$ & $1.6752$ & $0.7182$ & $0.1904$ \\
        \hline
        $1/25^2$ & $\bm{0.6926}$ & $\bm{0.1200}$ & $1.7622$ & $1.7504$ & $0.7245$ & $0.1936$ \\
        \hline
        $1/30^2$ & $\bm{0.6934}$ & $\bm{0.1152}$ & $1.7742$ & $1.6624$ & $0.7253$ & $0.2112$ \\
        \hline
    \end{tabular}}
    \caption{Experiment results on parameterized AA.  The smaller $\langle 4O_1 \rangle$, the better query complexity.}\label{tab:2}
    \vspace{-15pt}
\end{table}

\subsection{Quantum Walk with Parameterized Shift Operator}
Quantum walk (QW) algorithms~\cite{ambainis2001one, szegedy2004quantum,childs2009quantum,wong2017equivalence},   which shares some similarities with Grover's algorithm, 
are vibrant in the area of quantum algorithms.
In the context of the grid search,  \citet{benioff2000space} observed that the standard Grover's search algorithm needs $\Omega(N)$ steps to find a marked vertex in an $\sqrt{N}\times \sqrt{N}$ grid.  
Quantum walk with a natural (`moving') shift operator $S_m$, which keeps the direction (also called the coin) after every move, also takes at least $\Omega(N)$ steps to find a marked vertex in this grid~\cite{ambainis2005coins}.
\begin{align*}
    S_f: \ket{\Leftarrow, x,y} &\to \ket{\Rightarrow, x-1,y} & S_m: \ket{\Leftarrow, x,y} &\to \ket{\Leftarrow, x-1,y} \\
    \ket{\Rightarrow, x,y} &\to \ket{\Leftarrow, x+1, y} & \ket{\Rightarrow, x,y} &\to \ket{\Rightarrow, x+1, y} \\
    \ket{\Uparrow, x, y} &\to \ket{\Downarrow, x,y+1} & \ket{\Uparrow, x, y} &\to \ket{\Uparrow, x,y+1} \\
    \ket{\Downarrow, x, y} &\to \ket{\Uparrow, x,y-1} & \ket{\Downarrow, x, y} &\to \ket{\Downarrow, x,y-1}
\end{align*}
They resolved this issue by introducing another shift operator $S_f$, which can be interpreted as changing direction after every move, and the quantum walk associated with $S_f$ takes $O(\sqrt{N}\log N)$ steps to find a marked vertex in this grid.

We can see that designing a shift operator, the direction (coin) transformation, is important for the performance of the quantum walk.
It motivates us to parameterize the shift operator and use our framework to determine a good shift operator.

\subsubsection{Parameterized QW Program.}
The quantum walk search algorithm in~\cite{ambainis2005coins} first initializes the coin variable and position variable in the uniform superposition and applies the marked quantum walk operator for several times (here we apply it twice), then measure the position variable to check if the measured vertex is the marked one.
We parameterize the shift operator and write it as follows:
\begin{equation}\label{prog:pqw}
    \begin{aligned}
        P_2(\theta_1, \theta_2)
        \equiv{}& \qinit{t};\\
        & \qwhile M[q_x,q_y] = 1 \qdo \\
        & \quad \qinit{c_{x}};\qinit{c_{y}};\qinit{q_x};\qinit{q_y};\\
        & \quad \qut{H}{c_x};\qut{H}{c_y};\qut{\tilde{H}}{q_x,q_y}; \\
        & \quad \qut{C}{c_x,c_y,q_x,q_y}; \qut{S(\theta_1,\theta_2)}{c_x,c_y,q_x,q_y}; \\
        & \quad \qut{C}{c_x,c_y,q_x,q_y}; \qut{S(\theta_1,\theta_2)}{c_x,c_y,q_x,q_y}; \\
        & \quad \qif (M'[t] = 0 \to A[t] \quad \Box{}= 1 \to \qskip)\qfi \qod
    \end{aligned}
\end{equation}
where $c_x$ and $c_y$ are two qubit variables for coin, indicating the direction $\Leftarrow, \Rightarrow$ and $\Uparrow, \Downarrow$, respectively.
$q_x, q_y$ are two variables with space $\cH_{\sqrt{N}} = \lbrace\ket{0}, \ldots, \ket{\sqrt{N}-1}\rbrace$, indicating the position.
The variable $t$ is a variable with space $\cH_{\sqrt{N}}$, for counting the running times of loops.
$\tilde{H}$ is a Hadamard-like unitary to create uniform superposition on $q_x, q_y$, which is composited by local operations that only allow transition on adjacent position, e.g., $\ket{x,y}$ and $\ket{x-1\!\!\mod \sqrt{N},y}$.
$C$ is the marking coin operator in~\cite{ambainis2005coins} and $S(\theta_1,\theta_2) = e^{-i\theta_2\ketbra[c_y]{+}{+}} e^{-i\theta_1\ketbra[c_x]{+}{+}}S_m $ is the parameterized shift operator, which can be implemented by a subprogram as follows:
\[\qut{S_m}{c_x,c_y,q_x,q_y};\qut{e^{-i\theta_1\ketbra[c_x]{+}{+}}}{c_x};\qut{e^{-i\theta_2\ketbra[c_y]{+}{+}}}{c_y}\]
In particular, we have $S(0,0)=S_m$ and $S(\pi, \pi)=S_f$. 
$A = \sum_{n=0}^{\sqrt{N}-1}\ketbra{n+1}{n}+\ketbra{0}{\sqrt{N}}$ adds $t$ by $1$ in every loop and measurement $M'$ checks the value of $t$ by 
$M' = \left\{ M_0' = \sum_{n=0}^{\sqrt{N}-1}\ketbra[t]{n}{n}, M_1' = I - M_0'\right\}.$

In this experiment, we choose $N = 16$ and the grid is $\{(i,j):0\leq i,j\leq 3\}$.
The loop measurement $M$ is $\lbrace M_0 = \ketbra[q_x]{3}{3}\otimes\ketbra[q_y]{3}{3}, M_1 = I- M_0\rbrace$ with $(3,3)$ being the marked vertex for convenience.
The input state can be arbitrary since all variables in $P_2(\theta_1,\theta_2)$ will be initialized.
The observable we choose is $O_2 = \frac{1}{\sqrt{N}}\sum_{n=1}^{\sqrt{N}} n\ketbra[t]{n}{n}$, which expresses the running number of loop
iterations as the $O_1$ in Sect.~\ref{sec:case_paa}.
Our \textbf{target} is to \emph{identify shift operator $S(\theta_1^*, \theta^*_2)$ that minimizes the expectation function of observable $O_2$ with $P_2(\theta_1, \theta_2)$}.

We summarize below the needed configuration in the experiment workflow.
\begin{description}
    \item[\textbf{Given:}] Parameterized QW program $P_2(\theta_1,\theta_2)$, arbitrary input state $\rho$ and observable $O_2$.
    \item[\textbf{Workflow:}] In (2), samples' number: $2\times 10^4$ empirically chosen with details in Appendix~\ref{sec:details}. In (3), Adam's setting: $\beta_1 = 0.9, \beta_2 = 0.999, \alpha = 0.1$; initial parameter: $\theta_1, \theta_2$ are randomly initialized in $2\pi\times [0.1, 0.9]$ to avoid certain extreme cases when $(\theta_1,\theta_2)$ is close to $(0,0)$; Goal: \textbf{minimize} the expectation function of observable $O_2$.
\end{description}

\subsubsection{Results.} During the optimizing process, we recorded the Mean-Squared-Error (MSE) distance\footnote{MSE distance between $(\theta_1, \theta_2)$ and $(\pi,\pi)$ is  $\frac{1}{2}((\theta_1-\pi)^2+ (\theta_2-\pi)^2)$.} of 
parameters $(\theta_1,\theta_2)$ from $(\pi,\pi)$, which is shown in \figref{fig:mse}.
Each colored line in \figref{fig:mse} represents an independent optimization with different initial parameters.
In particular, the initial parameters of the red line in \figref{fig:mse} are manually set to $(0.2\pi,0.2\pi)$ to be far away from $(\pi,\pi)$. 
All independent training optimizing threads converge to the shift operator $S(\pi, \pi)=S_f$ after $60$ steps, which recovers the operator $S_f$ by human design~\cite{ambainis2005coins}, automatically in our experiment.

\begin{SCfigure}
\scalebox{0.65}{
    \includegraphics{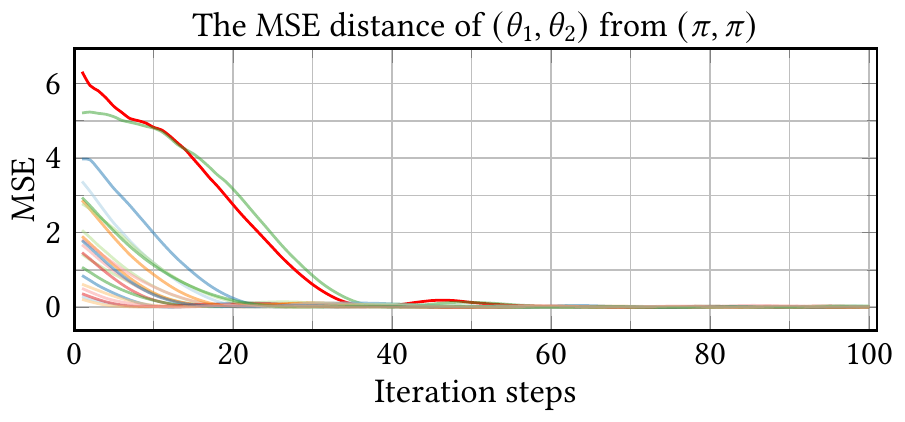}}
    \caption{MSE distance with respect to the iteration steps in optimizing $P_2(\theta_1,\theta_2)$. Differently colored lines represent $21$ experiments with randomly initialized parameters $(\theta_1,\theta_2)$.}\label{fig:mse}
\end{SCfigure}

\subsection{Repeat-Until-Success Unitary Implementation}
In this subsection, we 
demonstrate that our framework can learn realizable instances of repeat-until-success (RUS) circuits. RUS depicts a design pattern, repeating an operation until getting the desired result, which has been widely used in quantum circuit design~\cite{10.5555/2685179.2685181,10.5555/3179320.3179329,PhysRevLett.95.030505,PhysRevLett.114.080502}.
A general layout of RUS circuits~\cite{PhysRevLett.114.080502} is shown in Fig.~\ref{fig:rus_circuit}, where the dashed part is always applied if the measurement outcome is undesirable.
Notice that $W_j$ in Fig.~\ref{fig:rus_circuit} is designed to  restore the state of the system to $\ket{0}\ket{\psi}$ based on the measurement outcome, as only one copy of $\ket{\psi}$ is provided.
The RUS circuits have been shown to achieve a better (expected) depth over ancilla-free techniques for single-qubit unitary decomposition~\cite{PhysRevLett.114.080502, 10.5555/2685179.2685181}.
\begin{figure}[ht]
    \centering
    \scalebox{0.9}{
    \includegraphics{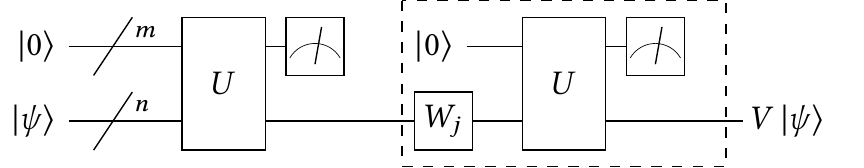}
    }
    \caption{RUS design circuit to implement unitary $V$ \cite{PhysRevLett.114.080502}.}\label{fig:rus_circuit}
\end{figure}

\subsubsection{Parameterized RUS Program.} Consider the  program:
\begin{equation}\label{prog:prus}
    \begin{aligned}
        P_3(\theta_1, \theta_2, \theta_3)
        \equiv{}& \qinit{r}; \qut{U}{q,r}; \\
        &\qwhile M[r] = 1 \qdo \qut{W(\theta_1,\theta_2,\theta_3)}{q}; \qinit{r}; \qut{U}{q,r}; \qod,
    \end{aligned}
\end{equation}
where $q$ is a qubit variable and $r$ is an ancilla qubit variable, measurement $M = \lbrace M_0=\ketbra[r]{0}{0}, M_1=\ketbra[r]{1}{1}\rbrace$.
We are provided with a unitary $U$ that will induce the  desired operation on $q$ if the outcome of performing measurement $M$ after execution of $\qinit{r}; \qut{U}{q,r}$ is $0$,
otherwise, we need a recovery operation $W(\theta_1,\theta_2,\theta_3) = e^{-i\theta_1\ketbra[q]{0}{0}}e^{-i\theta_2\ketbra[q]{+}{+}}e^{-i\theta_3\ketbra[q]{0}{0}}$, a fully parameterized single-qubit unitary ($Z$-$X$ decomposition in~\cite{nielsen2010quantum}) that can be implemented by a subprogram as follows,
\[\qutp{\theta_3}{\ketbra[q]{0}{0}}{q};\qutp{\theta_2}{\ketbra[q]{+}{+}}{q};\qutp{\theta_1}{\ketbra[q]{0}{0}}{q}\]
to restore the state of $q$ and repeat the whole process until obtaining the outcome $0$.

In this experiment, the unitary $U$ in program $P_3(\theta_1,\theta_2,\theta_3)$ is chosen as $ (\ketbra[r]{0}{0}\otimes V_1 + \ketbra[r]{1}{1}\otimes V_2)(H\otimes I)$ with randomly generated single-qubit unitaries $V_1$ and $V_2$.
Our \textbf{target} is to \emph{identify suitable parameters $\theta_1^*,\theta_2^*,\theta_3^*$ such that program $P_3(\theta_1^*,\theta_2^*,\theta_3^*)$ acts as the same as the unitary $V_1$ for an information-complete basis
\[\{\ket{\psi_1} = \ket{0}, \ket{\psi_2}=\ket{1}, \ket{\psi_3} = \ket{+} = (\ket{0}+\ket{1})/\sqrt{2}, \ket{\psi_4} = \ket{Y_+}=(\ket{0}+i\ket{1})/\sqrt{2}\}\]
of variable $q$}. That is, for any $1\leq j\leq 4$,
$\sem{P_3(\theta_1^*,\theta_2^*,\theta_3^*)}(\ketbra[q]{\psi_j}{\psi_j}) = V_1\ketbra[q]{\psi_j}{\psi_j}V_1^{\dagger}$, which is equivalent to
$\tr\left(V_1\ketbra[q]{\psi_j}{\psi_j}V_1^{\dagger} \sem{P_3(\theta_1^*,\theta_2^*,\theta_3^*)}(\ketbra[q]{\psi_j}{\psi_j})\right) = 1$.
Therefore, we choose four pairs of input states and observable
$(\rho_j = \ketbra[q]{\psi_j}{\psi_j}, O_{3,j} = V_1\ketbra[q]{\psi_j}{\psi_j}V_1^{\dagger}), 1\leq j \leq 4$ and denote the expectation function of program $P_3(\theta_1,\theta_2,\theta_3)$ with respect to input state $\rho_j$ and observable $O_{3,j}$ as $f_j(\theta_1,\theta_2,\theta_3)$ for $1\leq j \leq 4$.
In order to optimize the four functions $f_j,1\leq j\leq 4$ simultaneously close to $1$, we introduce a MSE loss function $l(\theta_1,\theta_2,\theta_3)= \frac{1}{4}\sum_{j=1}^4 f_j(\theta_1,\theta_2,\theta_3)- 1)^2$ to be minimized in the experiment workflow. 

We summarize below the needed configuration in the experiment workflow.

\begin{description}
    \item[\textbf{Given:}] Parameterized RUS program $P_3(\theta_1,\theta_2,\theta_3)$ and four pairs of input state and observable $(\rho_j,O_{3,j}), 1\leq j\leq 4$.
    \item[\textbf{Workflow:}] In (2), samples' number: $4.7\times 10^4$ empirically chosen with details in Appendix~\ref{sec:details}. In (3), Adam's setting: $\beta_1 = 0.9, \beta_2 = 0.999, \alpha = 0.2$; initial parameter: $\theta_1, \theta_2$ are randomly initialized; goal: \textbf{minimize} the MSE loss function $l(\theta_1,\theta_2,\theta_3)$.
\end{description}

\subsubsection{Results.} We did $10$ independent optimizing. In each optimization, $V_1$ and $V_2$ are randomly generated.
With the iteration steps less than $60$, the MSE loss $l$ can be reduced to $0.0001$, which implies $f_j$ is greater than $0.98$ for all $j$.
Because $\lbrace \ketbra[q]{\psi_j}{\psi_j}\rbrace_{j=1}^4$ forms a complete basis of $\cD(\cH_{q})$, we can conclude that the program $P_3$ produces an approximate operation as desired $V_1$ in each optimizing.
Therefore, the experimental result confirms the feasibility of our framework and also validates the experiment workflow for automatically getting suitable parameters.

\begin{acks}                            
  We thank anonymous reviewers for constructive suggestions that improve the presentation of the paper. 
  W. Fang and M. Ying were partially supported by the 
  \grantsponsor{}{National Key R\&D Program of China} under Grant No.~\grantnum{}{2018YFA0306701} and the
  \grantsponsor{}{National Natural Science Foundation of China}{} under Grant No.~\grantnum{}{61832015}. 
  X.W. was partially funded by the U.S. Department of Energy, Office of Science, Office of
Advanced Scientific Computing Research, Quantum Testbed Pathfinder Program under Award
Number DE-SC0019040 and by the U.S. National Science Foundation
grant CCF-1942837 (CAREER).
 
\end{acks}

\section*{Code availability}
\addcontentsline{toc}{section}{Code availability}
The code for implementation (Sect.~\ref{sec:implement}) and experiments (Sect.~\ref{sec:case_study}) are available at \url{https://github.com/njuwfang/DifferentiableQPL}.

\bibliography{ref}

\clearpage

\appendix

\section{Quantum Preliminaries} \label{sec:prelim}
In this section, we recall some basic knowledge of quantum computing. The reader can consult the
standard textbook~\cite[Chapter 2, 4]{nielsen2010quantum} for more details.

\subsection{States and Hilbert Spaces}
The state space of an isolated quantum system is represented by a complex Hilbert space.
We use the Dirac notation $\ket{\psi}$ to denote a vector in a Hilbert space.
The (vector dual) Hermitian conjugate of $\ket{\psi}$ is denoted by $\bra{\psi}$.
The inner product of $\ket{\psi}$ and $\ket{\phi}$ is denoted by $\braket{\phi}{\psi}$, considered as a shorthand for $\bra{\phi}(\ket{\psi})$.
The norm of a vector $\ket{\psi}$ is defined as $\norm{\ket{\psi}} = \sqrt{\braket{\psi}{\psi}}$. A unit vector is referred to as a \emph{pure state}.

For example, a quantum bit (qubit) has a two-dimensional state space $\cH_2$ with $\ket{0} = (1,0)^{\dagger}$ and $\ket{1} = (0,1)^{\dagger}$ form an orthonormal basis of $\cH_2$.
A pure state $ \ket{\psi}\in \cH_2$ can be represented by $\alpha\ket{0} + \beta\ket{1}$ with $\abs{\alpha}^2 + \abs{\beta}^2 = 1$.
There are also two states of a qubit often appear:
$\ket{+} = \frac{1}{\sqrt{2}}(\ket{0}+\ket{1}), \enspace \ket{-} = \frac{1}{\sqrt{2}}(\ket{0} - \ket{1}).$
They also form a basis of $\cH_2$.

A (linear) operator is a linear mapping between Hilbert spaces and the set of all operators from $\cH$ to $\cH'$ is denoted by $\cL(\cH,\cH')$.
Specifically, an operator $A\in\cL(\cH,\cH)$ is said an operator on $\cH$ and write $A\in \cL(\cH)$. 
We often write $I_{\cH}$ for the identity operator on $\cH$. 

The Hermitian conjugate (adjoint) of an operator $A$ is denoted by $A^{\dagger}$.
An operator $A$ on $\cH$ is \emph{Hermitian} if $A^{\dagger} = A$.
An operator $A$ on $\cH$ is \emph{positive semidefinite} if for all vectors $\ket{\psi} \in \cH$, $\bra{\psi}A\ket{\psi} \geq 0$.
The L\"{o}wner order $\sqsubseteq$ is defined as $ A \sqsubseteq B$ if $B - A$ is positive semidefinite.
The \emph{trace} of an operator $A$ on $\cH$ is defined as $\tr(A) = \sum_{j}\bra{\psi_j}A\ket{\psi_j}$, with $\lbrace \ket{\psi_j}\rbrace$ an orthonormal basis of $\cH$.

When the state of a quantum system is not completely known, people may think it as a \emph{mixed state} (ensemble of pure state) $\lbrace (p_j, \ket{\psi_j}) \rbrace$ meaning that it is at $\ket{\psi_j}$ with probability $p_j$.
A \emph{density operator} for this system is defined as $ \rho = \sum_{j} p_j\ketbra{\psi_j}{\psi_j}.$
Formally, a density operator $\rho$ on a Hilbert space $\cH$ is a positive semidefinite operator with $\tr(\rho) = 1$.
Moreover, a \emph{partial density operator} $\rho$ on $\cH$ is defined as a positive semidefinite operator with $\tr(\rho) \leq 1$.
We use $\cD(\cH)$ to denote the set of partial density operators on $\cH$.

\subsection{Quantum Operations}
\subsubsection{Unitary transformations.} An operator $U$ on a Hilbert space $\cH$ is a \emph{unitary transformation} if $U^{\dagger}U = UU^{\dagger}=I_{\cH}$.
A unitary transformation $U$ describes the evolution from any pure state $\ket{\psi}$ to $U\ket{\psi}$.
For mixed states, this evolution is reformulated as from any mixed state $\rho$ to $U\rho U^{\dagger}$.

There are some important unitaries: the \emph{Hadamard} operator $H$, which has the action $H\colon \ket{0} \mapsto \ket{+}$, $\ket{1}\mapsto \ket{-}$;
the \emph{Pauli} operators: $X,Y,Z$ with $X\colon \ket{0} \mapsto \ket{1}$, $\ket{1} \mapsto \ket{0}$, $Y\colon \ket{0} \mapsto i\ket{1}, \ket{1}\mapsto -i\ket{1}$, $Z\colon \ket{0}\mapsto \ket{0}, \ket{1}\mapsto -\ket{1}$.

\subsubsection{Measurements and observables.}
A measurement on a system with a state space $\cH$ is described by a collection $\lbrace M_m \rbrace$ of \emph{measurement operators} on $\cH$ with the \emph{completeness equation}:
$ \sum_m M_m^{\dagger}M_m = I_{\cH}.$
When performing a measurement $\lbrace M_m \rbrace$ on a pure state $\ket{\psi}$ and a mixed state $\rho$, the outcome of index $m$ occurs with probability
$p(m) = \bra{\psi}M_m^{\dagger}M_m\ket{\psi}$ and $p(m) = \tr(M_m \rho M_m^{\dagger})$, the corresponding state of the system after the measurement is $\ket{\psi_m} =  {M_m\ket{\psi}}/{\sqrt{p(m)}}$ and $\rho_m = {M_m\rho M_m^{\dagger}}/{p(m)}$, respectively.
In the context of mixed states, if we don't know the outcome of the measurement, the state of the system after the measurement can be described by 
$\sum_{m} p(m)\rho_m = \sum_m M_m\rho M_m^{\dagger}.$

A projective measurement is often described by an \emph{observable}, $M$, a Hermitian operator on $\cH$.
The  spectral decomposition of $M= \sum_m m P_m$, corresponds to a quantum measurement $\lbrace P_m \rbrace$ with measurement outcome $m$ for each $P_m$. The average value of this measurement performed on a state $\ket{\psi}$ is
$\bra{\psi}M\ket{\psi}.$
The value $\bra{\psi}M\ket{\psi}$ is often written as $\langle M \rangle$ and called the \emph{expectation} of $M$.
For a mixed state $\rho$, the expectation of $M$ is $\tr(M \rho)$.

\subsubsection{General quantum operations.}
A \emph{superoperator} is a linear mapping between $\cL(\cH)$ and $\cL(\cH')$.
For mixed states, unitary transformations and measurements can be described by a general form of completely-positive and trace-non-increasing superoperators, which has the \emph{Kraus representation}: $\sum_j E_j (\cdot) E_j^{\dagger}$ with $E_j \in \cL(\cH, \cH')$ and $\sum_j E_j^{\dagger}E_j \sqsubseteq I_{\cH}$~\cite{wolf2012quantum}.
The Schr\"{o}dinger-Heisenberg dual of a superoperator $\cE$ with the Kraus representation $\cE(\cdot) = \sum_j E_j(\cdot)E_j^{\dagger}$ is $\cE^*(\cdot) = \sum_j E_j^{\dagger} (\cdot) E_j$.

\subsection{Composite systems and Tensor Products}
The tensor product of two vectors $\ket{\psi_1}$ and $\ket{\psi_2}$ is denoted by $\ket{\psi_1}\otimes\ket{\psi_2}$, which is sometimes written as $\ket{\psi_1}\ket{\psi_2}$ or even $\ket{\psi_1\psi_2}$ for short.
The tensor product of two Hilbert spaces $\cH_1$ and $\cH_2$ is denoted by $\cH_1\otimes \cH_2$.
For any linear operator $A_1 \in \cL(\cH_1)$ and $A_2 \in \cL(\cH_2)$, their tensor product operator $A_1\otimes A_2 \in \cL(\cH_1\otimes \cH_2)$ is defined by linear extensions of $A_1\otimes A_2 (\ket{\psi_1}\otimes \ket{\psi_2}) = (A_1\ket{\psi_1})\otimes (A_2\ket{\psi_2})$ for any $\ket{\psi_1}\in\cH_1, \ket{\psi_2}\in \cH_2$.

The state space of a composite quantum system is the tensor product of its components' state spaces, e.g., if a system with two components in state $\ket{\psi_1} \in \cH_1$ and state $\ket{\psi_2} \in \cH_2$ respectively, then the joint state of the composite system is $\ket{\psi_1}\otimes\ket{\psi_2} \in \cH_1\otimes \cH_2$.
For mixed states, if a system with two components in the state $\rho_1 \in \cD(\cH_1)$ and the state $\rho_2 \in \cD(\cH_2)$ respectively, then the joint state of the composite system is $\rho_1\otimes\rho_2 \in \cD(\cH_1\otimes \cH_2)$.

For two Hilbert spaces $\cH_1$, $\cH_2$ and any operator $A \in \cL(\cH_1\otimes \cH_2)$, the \emph{partial trace} over space $\cH_2$ of $A$ is
$ \tr_{\cH_2}(A) = \sum_{j}(I_{\cH_1}\otimes\bra{\psi_j})A(I_{\cH_1}\otimes\ket{\psi_j}) \in \cL(\cH_1)$,
where $\lbrace \ket{\psi_j}\rbrace$ is an orthonormal basis of $\cH_2$ and we often write $\tr_2$ for $\tr_{\cH_2}$ if there is no ambiguity. 

\section{Practical Variance Bound for Differential Programs}\label{appendix:variance}
In this section, we give fine bounds of $\langle O_d^2\otimes O^2\rangle$ and $\langle O_d^4\otimes O^4\rangle$ that are used in our case studies for estimating the number of samples. Before that, we give the formal definitions of the two previous notions, $RC_{\theta}(P(\bm{\theta}))$ and $LC(P(\bm{\theta}))$.

\begin{definition}
    The ``Running Count for $\theta$'' in $P(\bm{\theta})$, denoted $RC_{\theta}(P(\bm{\theta}))$, is defined inductively on the program structure:
    \begin{itemize}
        \item   $RC_{\theta}(P(\bm{\theta})) = 0$ for  $P(\bm{\theta}) \equiv \qskip$, $ \qinit{q}$ or $ \qut{U}{\bar{q}}$.
        \item If $P(\bm{\theta}) \equiv \qutp{\theta'}{\sigma}{\bar{q}}$, then  $RC_{\theta}(P(\bm{\theta})) = 1$ when  $\theta'$ is $\theta$; otherwise, $RC_{\theta}(P(\bm{\theta})) = 0$.
        \item If $P(\bm{\theta}) \equiv P_1(\bm{\theta});P_2(\bm{\theta})$, then  $RC_{\theta}(P(\bm{\theta})) = RC_{\theta}(P_1(\bm{\theta})) + RC_{\theta}(P_2(\bm{\theta}))$.
        \item If $P(\bm{\theta}) \equiv \qif (\Box m\cdot M[\bar{q}] = m \to P_{m}(\bm{\theta})) \qfi$, then  $RC_{\theta}(P(\bm{\theta})) = \max_m RC_{\theta}(P_m(\bm{\theta}))$.
        \item If  $P(\bm{\theta}) \equiv \qqwhile{\bar{q}}{Q(\bm{\theta})}$, then  $RC_{\theta}(P(\bm{\theta}))= RC_{\theta}(Q(\bm{\theta}))$.
    \end{itemize}
\end{definition}

\begin{definition}\label{dfn:lc}
    The ``Loop Count'' in $P(\bm{\theta})$, denoted $LC(P(\bm{\theta}))$, is defined by induction on the program structure as follows:
    \begin{itemize}
        \item  $LC(P(\bm{\theta})) = 0$ for $P(\bm{\theta}) \equiv \qskip$, 
     $ \qinit{q}$,  $ \qut{U}{\bar{q}}$ or $ \qutp{\theta'}{\sigma}{\bar{q}}$.     \item If $P(\bm{\theta}) \equiv P_1(\bm{\theta});P_2(\bm{\theta})$, then  $LC(P(\bm{\theta})) = LC(P_1(\bm{\theta})) + LC(P_2(\bm{\theta}))$.
        \item If $P(\bm{\theta}) \equiv \qif (\Box m\cdot M[\bar{q}] = m \to P_{m}(\bm{\theta})) \qfi$, then  $LC(P(\bm{\theta})) = \sum_m LC(P_m(\bm{\theta}))$.
        \item If  $P(\bm{\theta}) \equiv \qqwhile{\bar{q}}{Q(\bm{\theta})}$, then  $LC(P(\bm{\theta}))$ $= LC(Q(\bm{\theta})) + 1$.
    \end{itemize}
\end{definition}

\begin{theorem}\label{thm:bound}
    In the same setting as in Theorem~\ref{thm:exact}, for a fixed $\bm{\theta}$, if all the \textbf{while}-statements (subprograms) in $P(\bm{\theta})$ terminate almost surely, then the expectation of $O_d^2\otimes O^2$: 
    \[\left\langle O_d^2\otimes O^2\right\rangle = \tr\left(O_d^2\otimes O^2\sem{\qd{\theta}(P(\bm{\theta}))}(\rho)\right)\]
    is upper-bounded by
    \begin{equation}\label{eq:var}
        \begin{aligned}
        & M^2\bigggl(4S(M_1) + \sum_{k=1}^{\infty}\biggl((M_2+(k-1)(k^{M_2-1}-1)_+) S\left((k+1)^{M_2}M_1-1\right)\left(2\epsilon^{\lfloor \frac{k-1}{N_{\epsilon}}\rfloor} + 2\epsilon^{\lfloor \frac{k-1}{N_{\epsilon}}\rfloor-1}\right)\biggr)\bigggr),
    \end{aligned}
    \end{equation}
    where \begin{itemize}
        \item $M$ is the largest eigenvalue of $\abs{O}$;
        \item $M_1 = RC_{\theta}(P(\bm{\theta}))$, $M_2=LC(P(\bm{\theta}))$;
        \item $\mu$ is the distribution we adopted in code transformation rules and satisfies \ref{condition};
        \item $S(n) \equiv \sum_{j=1}^{n}1/\mu(j)$ for every $n\geq 1$; $(x)_+\equiv \max\lbrace 0, x\rbrace$;
        \item $\epsilon \in (0,1)$ and $N_{\epsilon}$ is the largest number of $N$ in Lemma~\ref{lem:a9} that is applied to all $M_2$ loop statements in $P(\bm{\theta})$.
    \end{itemize}
\end{theorem}
\begin{proof}
See Appendix~\ref{prf:bound}.
\end{proof}

The \ref{condition} ensures \[\lim_{n\to\infty}\sqrt[n]{S(n)} = 1.\]
Thus, the infinite summation terms in Equation~(\ref{eq:var}) have exponential damping factor $\epsilon^{\frac{1}{N_{\epsilon}}} < 1$, then the summation is convergent. To give a clear sense of this bound, we can derive a corollary when we use the distribution mentioned in Equation~(\ref{eq:distribution}).

\begin{cor}\label{cor:1}
In the same setting as in Theorem~\ref{thm:bound}, let the distribution $\mu$ be 
\[ \mu(j) = \frac{1}{c(s) j\ln^{1+s}(j+e)},\]
where $c(s) = \sum_{j=1}^{\infty} 1/(j)/\ln^{1+s}(j+e) < \infty, s\in(0,1]$. We have the expectation of $O_d^2\otimes O^2$ is upper-bounded by
\begin{equation}
    \begin{aligned}
    & M^2\bigggl(4T(M_1) + \sum_{k=1}^{\infty}\biggl((M_2+(k-1)(k^{M_2-1}-1)_+) T\left((k+1)^{M_2}M_1-1\right)\left(2\epsilon^{\lfloor \frac{k-1}{N_{\epsilon}}\rfloor} + 2\epsilon^{\lfloor \frac{k-1}{N_{\epsilon}}\rfloor-1}\right)\biggr)\bigggr),
\end{aligned}
\end{equation}
where $T(n) = \frac{c(s)}{2}x^2\ln^{1+s}(x+e)$ is an upper bound of $S(n)$ by integral. 
\end{cor}

With $T(n)$ substituted, Corollary~\ref{cor:1} implies 
\begin{align*}
    \langle O_d^2\otimes O^2\rangle \leq{}& 2M^2c(s)\Bigl(M_1^2\ln^{1+s}(M_1+e) + M_1^{3+s}M_2^{2+s}\sum_{k=1}^{\infty}(k+1)^{3M_2+1}\epsilon^{\frac{k-1}{N_{\epsilon}}-2}\Bigr).
\end{align*}
The infinite summation $\sum_{k=1}^{\infty}(k+1)^{3M_2+1}\epsilon^{\frac{k-1}{N_{\epsilon}}-2}$ is related to Eulerian polynomials~\cite{barry2011eulerian} and can be easily bounded by
\[\frac{(3M_2+1)!}{\epsilon^{2+\frac{1}{N_{\epsilon}}}(1-\epsilon^{\frac{1}{N_{\epsilon}}})^{3M_2+2}}.\]
Thus we get the bound for $\langle O_d^2\otimes O^2\rangle$ that implies Theorem~\ref{thm:bound2}:
\[2M^2c(s)\left(M_1^2\ln^{1+s}(M_1+e) + \frac{M_1^{3+s}M_2^{2+s}((3M_2+1)!)}{\epsilon^{2+\frac{1}{N_{\epsilon}}}(1-\epsilon^{\frac{1}{N_{\epsilon}}})^{3M_2+2}}\right).\]

For bound of $\langle O_d^4\otimes O^4\rangle$, we have a theorem similar to Theorem~\ref{thm:bound}.

\begin{theorem}\label{thm:bound4}
    In the same setting as in Theorem~\ref{thm:exact}, for a fixed $\bm{\theta}$, if all the \textbf{while}-statements (subprograms) in $P(\bm{\theta})$ terminate almost surely, then the expectation of $O_d^4\otimes O^4$: 
    \[\left\langle O_d^4\otimes O^4\right\rangle = \tr\left(O_d^4\otimes O^4\sem{\qd{\theta}(P(\bm{\theta}))}(\rho)\right)\]
    is upper-bounded by
    \begin{equation}
        \begin{aligned}
        & 4M^2\bigggl(4S'(M_1) + \sum_{k=1}^{\infty}\biggl((M_2+(k-1)(k^{M_2-1}-1)_+) S'\left((k+1)^{M_2}M_1-1\right)\left(2\epsilon^{\lfloor \frac{k-1}{N_{\epsilon}}\rfloor} + 2\epsilon^{\lfloor \frac{k-1}{N_{\epsilon}}\rfloor-1}\right)\biggr)\bigggr),
    \end{aligned}
    \end{equation}
    where \begin{itemize}
        \item $M$ is the largest eigenvalue of $\abs{O}$;
        \item $M_1 = RC_{\theta}(P(\bm{\theta}))$, $M_2=LC(P(\bm{\theta}))$;
        \item $\mu$ is the distribution we adopted in code transformation rules and satisfies \ref{condition};
        \item $S'(n) \equiv \sum_{j=1}^{n}1/\mu^3(j)$ for every $n\geq 1$; $(x)_+\equiv \max\lbrace 0, x\rbrace$;
        \item $\epsilon \in (0,1)$ and $N_{\epsilon}$ is the largest number of $N$ in Lemma~\ref{lem:a9} that is applied to all $M_2$ loop statements in $P(\bm{\theta})$.
    \end{itemize}
\end{theorem}
\begin{proof}
    The Proof is similar to Theorem~\ref{thm:bound}.
\end{proof}

When unbounded loops exist (i.e., $M_2>0$), the bound given in Theorem~\ref{thm:bound} will depend on the number $N$ of iterations in Lemma~\ref{lem:a9}. We provide the following variant of Lemma~\ref{lem:a9} for a practically more approachable calculation, which, in particular, has been used in determining the number of samples in our case studies.

\begin{lemma}\label{lem:a91}
    Consider a quantum loop $P \equiv \qqwhile{\bar{q}}{Q}$ with fixed parameters (omitted) and a finite dimension of $\cH_P$; we define superoperators $\cE_i: \cD(\cH_P) \to \cD(\cH_P)$ by $\cE_i(\rho)=  M_i \rho M_i^{\dagger}$, $i=0,1$ and  $\cE: \cD(\cH_P) \to \cD(\cH_P)$ by $\cE(\rho)=  \sem{Q}(\rho)$. Assume that the Kraus representation of $\cE\circ\cE_1$ is $\sum_j E_j(\cdot)E_j^{\dagger}$.
    We write $E = \sum_jE_j\otimes E_j^*$. If $E$ can be diagonalized,  and there exists $\epsilon < 1$ such that the module of $E$'s eigenvalues is either equal to $1$ or less than $\epsilon$, then for $\forall n\in \mathbb{N}, \forall \rho \in \cD(\cH_P)$,
    \[ \tr(\cE_0\circ(\cE\circ\cE_1)^n(\rho)) \leq \epsilon^n \tr(\rho).\]
\end{lemma}
\begin{proof} This lemma can be proved using the techniques developed in \cite[Section 5]{YING20131679}, where it was proved that the module of all $E$'s eigenvalues is less than or equal to $1$. We omit the details here.
\end{proof}

\section{Automatic Differentiation Based on Parameter-shift Rule}\label{subsec-phase}
\subsubsection{{Parameter-shift Rule}} A way for differentiation of Hamiltonian simulation $e^{-i\theta H}$ with $H$ having at most two distinct eigenvalues, called the parameter-shift rule,   was given in~\cite{schuld2019evaluating, mitarai2018quantum}. For those Hamiltonians with more than two  distinct eigenvalues, the differentiation can be obtained via LCU (Linear Combination of Unitaries)~\cite{childs2012hamiltonian}. 
We also note some recent independent developments~\cite{wierichs2021general, izmaylov2021analytic, kyriienko2021generalized} of variants of the parameter-shift rules to handle more general $e^{-i\theta H}$. But here, for the sake of convenience, we still only use the previous parameter-shift rule, which is adopted in~\citet{Zhu2020Differentiable}'s work on differentiable quantum programming.

Let us consider a simple example: the expectation function
$f(\theta) = \tr(O e^{-i\theta X} \rho e^{i\theta X})$. We can check that
\[\frac {\mathrm{d}} {\mathrm{d} \theta}f(\theta) = f(\theta+\frac{\pi}{4}) - f(\theta-\frac{\pi}{4}).\]
More generally, if the Hamiltonian $H$ has only two eigenvalues $\pm r$, $r > 0 $ and 
\[f(\theta) =\tr\left(O e^{-i\theta H} \rho e^{i\theta H}\right),\] then
\[\frac {\mathrm{d}} {\mathrm{d} \theta}f(\theta) = r\left(f\left(\theta +\frac{\pi}{4r}\right) - f\left(\theta-\frac{\pi}{4r}\right)\right).\]
Although this form looks like a finite difference, it does express the exact derivative of $f$ rather than an approximate value. Therefore, the derivative can be obtained by shifting a single gate parameter. It is worth mentioning that the same differentiation was effectively achieved in   \cite{Zhu2020Differentiable} using one extra ancilla as the control qubit to create a superposition of two quantum circuits. 

The parameter-shift rule can be used as an alternative to the commutator form rule in the {DSOP} part of Fig.~\ref{fig:running_example_programs}. Furthermore, we can construct code transformation rules for AD based on the  parameter-shift rule, as we did based on the commutator form rule in the main body.

\subsubsection{{AD by Code Transformations}} We only considered  the parameterized forms used by~\cite{Zhu2020Differentiable}. That is, our parameterized unitary $U(\bm{\theta})$ is chosen from Pauli rotations
 \begin{equation}\label{eq:pauli}
     \left\lbrace R_{\sigma}(\theta) = e^{-i\frac{\theta}{2} \sigma}, R_{\sigma\otimes\sigma}(\theta) = e^{-i\frac{\theta}{2}\sigma\otimes \sigma} : \sigma = X, Y, Z; \theta\in \bm{\theta}\right\rbrace.
 \end{equation}
 We construct a code transformation operation $T_{\theta}$ in Fig.~\ref{fig:trans_rule2},
where $$C_-U(\theta) = \ketbra[A]{0}{0}\otimes U(0) + \ketbra[A]{1}{1}\otimes U(\pi)$$  (see~\cite{Zhu2020Differentiable} for more detail of the definition of $C_-U$), with $U(\theta)$ being  in the form of Pauli rotations and $A$ an ancilla quantum variable.
\begin{figure}[ht]
    \centering
    \begin{align*}
        \qtrans{}{\theta}(\qskip) &\equiv \qskip \\
        \qtrans{}{\theta}(\qinit{q}) &\equiv{} \qinit{q} \\
        \qtrans{}{\theta}(\qut{U}{\bar{q}}) &\equiv{} \qut{U}{\bar{q}} \\
        \qtrans{}{\theta}(\qut{U(\theta')}{\bar{q}}) &\equiv{} \qut{U(\theta')}{\bar{q}} \quad (\theta' \neq \theta) \\
        \qtrans{}{\theta}(P_1(\bm{\theta});P_2(\bm{\theta})) &\equiv{} \qtrans{}{\theta}(P_1(\bm{\theta}));\qtrans{}{\theta}(P_2(\bm{\theta})) \\
        \qtrans{}{\theta}(\qif (\Box m\cdot M[\bar{q}] = m \to P_m(\bm{\theta})) \qfi) &\equiv{} \qif (\Box m\cdot M[\bar{q}] = m \to \qtrans{}{\theta}(P_m(\bm{\theta}))) \qfi \\
        \qtrans{}{\theta}(\qwhile M[\bar{q}] = 1 \qdo P(\bm{\theta}) \qod) &\equiv{} \qwhile M[\bar{q}] = 1 \qdo \qtrans{}{\theta}(P(\bm{\theta})) \qod
    \end{align*}
    \begin{align*}
        \begin{aligned}
            (*)\  \qtrans{}{\theta}(\qut{U(\theta)}{\bar{q}}) \equiv{} \vphantom{\qif (M_{q_1,q_2}[q_1, q_2] = 0 \to{} \qut{C}{q_c}; \qut{GP}{q_c,q_2}} \\
            \vphantom{\Box{} = 1 \to{} \qut{X}{q_1};\qut{H}{A};} \\
            \vphantom{ \qut{C_-U(\theta)}{A, \bar{q}}; \qut{H}{A}} \\
            \vphantom{\Box{} = 2 \to{} \qskip ) \qfi;{\color{black}\qutp{\theta}{\sigma}{\bar{q}}}}
        \end{aligned}
        & 
        {\color{blue}\begin{aligned}
            \qif (M_{q_1,q_2}[q_1, q_2] = 0 \to{}& \qut{C}{q_c}; \qut{GP}{q_c,q_2} \\
            \Box{} = 1 \to{}& \qut{X}{q_1}; \qut{H}{A}; \\
            & \qut{C_-U(\theta)}{A, \bar{q}}; \qut{H}{A} \\
            \Box{} = 2 \to{}& \qskip ) \qfi;{\color{black}\qutp{\theta}{\sigma}{\bar{q}}.}
            \end{aligned}}
    \end{align*}
    \caption{Code transformation rules for  $T_{\theta}$,  where $M_{q_1,q_2} = \lbrace M_0 = \ketbra{00}{00}, M_1 = \ketbra{01}{01},$ $ M_2=\ketbra{10}{10}+\ketbra{11}{11}\rbrace$, $C = \sum_{j=0}^{\infty}\ketbra{j+1}{j}$ is the right-translation operator, $GP = \sum_{j=1}^{\infty}\ketbra{j}{j}\otimes R_y(2\arcsin(\sqrt{b_j}))$ and $b_j = \mu(j)/(1-\sum_{k=1}^{j-1}\mu(j))$.}
    \label{fig:trans_rule2}
\end{figure}
Then we have the following theorem:

\begin{theorem}\label{thm:exact2}
    Given a quantum program $P(\bm{\theta})$ that is parameterized by Pauli rotations in Eq.~(\ref{eq:pauli}), an initial state $\rho$, an observable $O$ on $\cH_{P(\bm{\theta})}$ of a finite dimension. Let
    \begin{align*}
        \frac{\partial}{\partial \theta}(P(\bm{\theta})) \equiv{}& \qinit{q_1};\qinit{q_2};\qinit{q_c}; \qinit{A}; \qut{C}{q_c};\qut{GP}{q_c,q_2};T_{\theta}(P(\bm{\theta}))
    \end{align*}
    Then
    \[\frac{\partial}{\partial \theta}\left(\tr(O\sem{P(\bm{\theta})}(\rho))\right) = \tr\left(Z_A\otimes O_c \otimes O\sem{\frac{\partial}{\partial \theta}(P(\bm{\theta}))}(\rho)\right),\]
    where $Z_A = \ketbra[A]{0}{0}-\ketbra[A]{1}{1}$, and \[O_c = \sum_{j=1}^{\infty}\frac{1}{\mu(j)}\ketbra{j}{j}\otimes\ketbra{1}{1}\] 
    is an observable on $\cH_{q_c}\otimes\cH_{q_1}$.
\end{theorem}
\begin{proof}
    The proof is similar to that of  Theorem~\ref{thm:exact}. The only difference is that in this proof, for every computation path $\pi$ and the subset $A_{\pi}$ we mentioned in the proof of Theorem~\ref{thm:exact}, let $\cE_{\eta}$ denote the superoperator of $\eta$ for any path $\eta$; then we need the following result  
    \[ \frac{\partial}{\partial\theta}(\tr(O\cE_{\pi}(\rho))) = \tr\left(Z_A\otimes O_c \otimes O\sum_{\eta\in A_{\pi}}\cE_{\eta}(\rho)\right),\]
    which is guaranteed by  the soundness theorem  in~\cite{Zhu2020Differentiable} (Theorem 6.2 therein).
\end{proof}

Given that all parameterized quantum programs with bounded loops considered in~\cite{Zhu2020Differentiable} are defined in the setting of Pauli rotations,  the above theorem
 (together with code transformation $T_{\theta}$) strictly improves the corresponding result of~\cite{Zhu2020Differentiable} with unbound loops.

\section{Number of Samples in Case studies}\label{sec:details}
In this section, we elaborate on how to determine the number of samples for estimating the expectation function of differential programs in case studies.
Our analysis mainly relies on the bound of $\langle O_d^2\otimes O^2\rangle$ and $\langle O_d^4\otimes O^4\rangle$ in Appendix~\ref{appendix:variance}. Although the variance bound we proved in Theorem~\ref{thm:bound} matches the one of \citet{Zhu2020Differentiable} when there is no unbounded loop, it is still not scalable in practical applications.
In the following analysis of case studies, the actual value of $\langle O_d^2\otimes O^2\rangle$ is much less than the bound we prove.

\subsection{Parameterized Amplitude Amplification}
It is a bit troublesome to directly estimate the variance bound of $P_1(\theta)$.
We need an auxiliary program:
\begin{align*}
    Q(\theta) \equiv{}&
    \qinit{q};\qinit{r};\qut{A}{q}; \\
    & \qwhile M[r] = 1 \qdo \\
    & \quad \qut{Z}{q};\qut{A^{\dagger}}{q};\qut{Z}{q};\qut{A}{q}; \\
    & \quad \qutp{\theta}{\sigma_{\ketbra{1}{1}\otimes Y}}{q,r} \\
    & \qod
\end{align*}
This $Q(\theta)$ doesn't contain the quantum variable $t$ in the program $P_1(\theta)$, but its behavior is similar to $P_1(\theta)$.
Its differential program $\qd{\theta}(Q(\theta))$ is also similar to $\qd{\theta}(P_1(\theta))$.
Consider the observable $\tilde{O}_1 = I$ for program $Q(\theta)$, we can conclude that the expectation $\langle O_d^2\otimes O_1^2\rangle$ of $\qd{\theta}(P_1(\theta))$ is less than the expectation $\langle O_d^2\otimes \tilde{O}_1^2\rangle$ of $\qd{\theta}(Q(\theta))$:
\begin{enumerate}
    \item Observable $O_1$ for $P_1(\theta)$ yields the result that is equal to or less than $1$ and observable $\tilde{O}_1$ for $Q(\theta)$ leads to the result $1$, which indicates that the output of $Q(\theta)$ is always greater then $P_1(\theta)$.
    \item Since the differential program keeps the same structure as the original program, the above result also holds for program $\qd{\theta}(Q(\theta))$ with observable $O_d^2\otimes \tilde{O}_1^2$ and program $\qd{\theta}(P_1(\theta))$ with observable $O_d^2\otimes O_1^2$ if they have executed the same branches.
\end{enumerate}
Therefore, the expectation $\langle O_d^2\otimes O_1^2\rangle$ of $\qd{\theta}(P_1(\theta))$ is less than the expectation $\langle O_d^2\otimes \tilde{O}_1^2\rangle$ of $\qd{\theta}(Q(\theta))$.

The Theorem~\ref{thm:bound} with the fact that $Q(\theta)$ meets the conditions of Lemma~\ref{lem:a91} can give us an upper bound of $\langle O_d^2\otimes \tilde{O}_1^2\rangle$.
When $p = 1/100$ and $\theta = 4\arccos((1-2\sqrt{p(1-p)})/(1+2\sqrt{p(1-p)})) = 3.3568$, we numerically calculate the $\epsilon$ for $Q(\theta)$ in Lemma~\ref{lem:a91} as $0.6681$.
With $M = 1, M_1=1, M_2=1$ in Theorem~\ref{thm:bound}, we obtain $799.72$ as an upper bound of $\langle O_d^2\otimes \tilde{O}_1^2\rangle$.
On the other hand, the bound for $\langle O_d^4\otimes \tilde{O}_1^4\rangle$ in Theorem~\ref{thm:bound4} can also be numerically calculated as $1.013\times 10^7$.
By Chebyshev's Inequality, we use $1.013\times 10^7/30^2/0.1 \approx 1\times 10^5$ samples to sample $\langle O_d^2\otimes O_1^2\rangle \approx 14.26$ in an error of $30$ with failure probability less than $10\%$. Thus we can use $30+14.26 = 44.26$ as the actual value of $\langle O_d^2\otimes O_1^2\rangle$. This $44.26$ is much less than $799.72$.
By Chebyshev's Inequality, to estimate $\langle O_d\otimes O_1\rangle$ in precision $\delta = 0.1$ with failure probability less than $c = 10\%$, the number of samples we need is less than $\text{Var}(O_d\otimes O_1)/(\delta^2c) \leq 4.426\times 10^4$.
In this experiment, we use $5/\sqrt{p}\times 10^3$ ($= 5\times 10^4$ when $p =1/100$) samples for each $p$.

\subsection{Quantum Walk with Parameterized Shift
Operator}
To estimate the number of samples for $P_2(\theta_1,\theta_2)$, we need another similar program as follows:
    \begin{align*}
        Q(\theta_1, \theta_2)
        \equiv{}& \qinit{t}; \qwhile M[q_x,q_y] = 1 \qdo \\
        & \quad \qinit{c_{x}};\qinit{c_{y}};\qinit{q_x};\qinit{q_y};\\
        & \quad \qut{H}{c_x};\qut{H}{c_y};\qut{\tilde{H}}{q_x,q_y}; \\
        & \quad \qut{C}{c_x,c_y,q_x,q_y}; \qut{S(\theta_1,\theta_2)}{c_x,c_y,q_x,q_y}; \\
        & \quad \qut{C}{c_x,c_y,q_x,q_y}; \qut{S_m}{c_x,c_y,q_x,q_y}; \\
        & \quad \qif (M'[t] = 0 \to A[t] \Box{}= 1 \to \qskip)\qfi \qod
    \end{align*}
The second shift operator in $Q(\theta_1,\theta_2)$ is not parameterized. But its behavior is the same as $P(\theta_1,\theta_2)$. Thus, we only need to estimate the number of samples for $Q(\theta_1,\theta_2)$.
When $(\theta_1,\theta_2) = (\pi,\pi)$, we can numerically calculate the $\epsilon$ for $Q(\theta_1,\theta_2)$ in Lemma~\ref{lem:a91} is less than $0.76$. With $M = 1, M_1 = 1, M_2 = 1$ in Theorem~\ref{thm:bound}, we obtain $2230.86$ as an upper bound of variance.
The bound for $\langle O_d^4\otimes O^4\rangle$ in Theorem~\ref{thm:bound4} is $8.14\times 10^7$. Then we use $8.14\times 10^7/20^2/0.1\approx 2\times 10^5$ samples to sample $\langle O_d^2\otimes O^2\rangle \approx 104.23$ in an error of $20$ with failure probability less than $10\%$.
Thus we can use $104.23 + 20 = 124.23$ as the actual value of $\langle O_d^2\otimes O_2^2 \rangle$.
By Chebyshev’s Inequality, to estimate $\langle O_d\otimes O_2\rangle$ in precision $\delta = 0.1$ with failure probability less than $c = 10\%$, the number of samples we need is less than $1.24\times 10^5$.

However, this number of samples is large for us as the simulation of $P_2(\theta_1,\theta_2)$ takes a lot of time in Q\#.
In this experiment, we choose $2\times 10^4$ as the number of samples.
Our experiment shows that this number of samples is already good for training.
This phenomenon has been studied in the optimization of PQCs (VQCs):
\citet{Sweke2020stochasticgradient} found that even using single measurement outcomes for estimation of expectation values is sufficient in optimization algorithms, which results in a form of \emph{stochastic} gradient descent optimization~\cite{ruder2017overview}.

\subsection{Repeat-Until-Success Unitary Implementation}
It is easy to see that $P_3(\theta_1,\theta_2,\theta_3)$ satisfies the conditions of Lemma~\ref{lem:a91} with $\epsilon = 0.5$. With $M = 1, M_1=1, M_2=1$ for each parameter, the variance bound in Theorem~\ref{thm:bound} is $243.19$. While the estimated value of $\langle O_d^2\otimes O^2\rangle$ is $15.298$.

The partial derivative of $l(\theta_1,\theta_2,\theta_3)$ with respect to $\theta_1$ is
\[\frac{\partial l}{\partial \theta_1} = \frac{1}{2}\sum_{j=1}^4 (E_j-1)\frac{\partial E_j}{\partial \theta_1}.\]
Suppose $E_j$ and ${\partial E_j}/{\partial \theta_1}$ are estimated in precision $\delta_1$ and $\delta_2$, respectively.
Then ${\partial l}/{\partial \theta_1}$ is in precision
\begin{align*}
    \frac{1}{2}\sum_{j=1}^4\left(\abs{E_j-1}\delta_2 + \abs{\frac{\partial E_j}{\partial \theta_1}}\delta_1 +\delta_1\delta_2\right)
    ={}&\left(\frac{1}{2}\sum_{j=1}^4\abs{E_j-1}\right)\delta_2 + \left(\frac{1}{2}\sum_{j=1}^4\abs{\frac{\partial E_j}{\partial \theta_1}}\right)\delta_1 + 2\delta_1\delta_2 \\
    \equiv{}& A \delta_1 + B \delta_2 + 2\delta_1\delta_2
\end{align*}
To limit it in $0.1$, we can choose $\delta_1 \sim 0.01/B, \delta_2\sim 0.09/A$.
Then assume that during most of the training process, $A \leq \frac{1}{2}\times 4 \times 0.3 = 0.6$, we have $\delta_2 \leq 0.15$.
By Chebyshev’s Inequality, to estimate ${\partial E_j}/{\partial \theta_1}$ in precision $\delta = 0.15$ with failure probability less than $c = 1\%$, the number of samples we need is less than $4.72\times 10^4$. With this number of samples, the probability of all ${\partial E_j}/{\partial \theta_1}, j=1,2,3,4$ are estimated in $0.15$ is greater than $0.99^4 = 92.2\%$.

\section{Detailed Proofs}\label{appendix:a}

\subsection{Proof of Lemma~\ref{lem:a9}}\label{prf:a9}
Before giving proof details of Lemma~\ref{lem:a9}, we need some lemmas and definitions.
For those who want a more detailed understanding of this subsection, you can refer to quantum graph theory and quantum Markov chains~\cite{ying2013Reachability,ying2016foundations, GUAN201855}.

With the same notations of Lemma~\ref{lem:a9}, we first list other needed notations:
\begin{itemize}
    \item $\cG = \cE\circ\cE_1$.
    \item $\sigma = \lim_{n\to\infty}\frac{1}{n}\sum_{k=0}^n \cG^n(I_{\cH_P})$, where $I_{\cH_P}$ is the identify operator on $\cH_P$. The existence of $\sigma$ is easily obtained, moreover $\lim_{n\to\infty}\frac{1}{n}\sum_{k=0}^n \cG^n$ is also a superoperator~\cite{wolf2012quantum}.
    \item $\cY = \left\lbrace \ket{\psi}\in \cH_P \mid \bra{\psi}\sigma\ket{\psi} = 0\right\rbrace$, $\cX = \supp(\sigma) \equiv \cY^{\perp}$.
    \item $P_X$ denotes the projector onto a space $X$, $I_{\cH_P}$ denotes the identify operator on $\cH_P$.
    \item The notation $\ket{\psi} \in \cH$ for any Hilbert space $\cH$ assumes $\norm{\ket{\psi}} = \braket{\psi}{\psi} = 1$, if there is no special remark.
\end{itemize}

\begin{lemma}
    [Modified from {\cite[Theorem 1]{ying2013Reachability}}]\label{lem:a7}~
    \begin{itemize}
        \item $\forall \ket{\psi}\neq 0\in \cX. \forall n \in \mathbb{N}. \tr(P_{\cX}\cG^n(\ketbra{\psi}{\psi})) = \braket{\psi}{\psi} = 1$.
        \item $\forall n\in \mathbb{N}. (\cG^*)^n(P_{\cX}) \sqsupseteq P_{\cX}$.
    \end{itemize}
\end{lemma}
\begin{proof} 
    Both of these propositions are trivial if $\cX$ is Zero space. Thus, in the following, we assume $\cX$ is not Zero space, which means $\sigma \neq 0$.
    \begin{itemize}
        \item According to~\cite[P. 105]{nielsen2010quantum}, for any $\ket{\psi}\in\cX = \supp(\sigma)$, there exist $\lambda > 0$ and $\mu \in \cD(\cX)$ such that $\sigma  = \lambda \ketbra{\psi}{\psi} + \mu$, then 
        \begin{align*}
            \tr(\sigma) &= \tr(P_{\cX}\sigma) = \tr(P_{\cX}\cG(\sigma)) = \cdots =\tr(P_{\cX}\cG^n(\sigma)) \\
            &= \lambda\tr(P_{\cX}\cG^n(\ketbra{\psi}{\psi})) + \tr(P_{\cX}\cG^n(\mu)) \\
            &\leq \lambda\tr(\ketbra{\psi}{\psi}) + \tr(\mu) = \tr(\sigma).
        \end{align*}
        Because $\lambda > 0$, we conclude that \[\tr(P_{\cX}\cG^n(\ketbra{\psi}{\psi})) = \braket{\psi}{\psi} = 1.\]
        \item The above statement tells that for any $\ket{\psi} \in\cX$ and any $n\in \mathbb{N}$
        \[ 1 = \tr(P_{\cX}\cG^n(\ketbra{\psi}{\psi})) = \bra{\psi}(\cG^*)^n(P_{\cX})\ket{\psi}\]
        together with $I_{\cH_P} \sqsupseteq (\cG^*)^n(I_{\cH_P}) \sqsupseteq (\cG^*)^n(P_{\cX})$, which means $\norm{(\cG^*)^n(P_{\cX})\ket{\psi}} \leq 1$, we have 
        \[\forall \ket{\psi}\in\cX. \forall n\in \mathbb{N}. \ket{\psi} = (\cG^*)^n(P_{\cX})\ket{\psi}.\]
        Now, for any $\ket{\alpha} = x\ket{\psi} +y \ket{\varphi} \in \cH_P$, where $\ket{\psi} \in \cX, \ket{\varphi}\in \cY, \abs{x}^2+\abs{y}^2 = 1$, we have
        \begin{align*}
            \bra{\alpha}((\cG^*)^n(P_{\cX}) - P_{\cX})\ket{\alpha}
            \geq{}& \bar{x}y\bra{\psi}(\cG^*)^n(P_{\cX})\ket{\varphi} + x\bar{y}\bra{\varphi}(\cG^*)^n(P_{\cX})\ket{\psi} \\
            ={}& \bar{x}y\braket{\psi}{\varphi} + x\bar{y}\braket{\varphi}{\psi} = 0.
        \end{align*}
        Thus $(\cG^*)^n(P_{\cX}) \sqsupseteq P_{\cX}$ for all $n \in \mathbb{N}$.
    \end{itemize}
\end{proof}

\begin{lemma}\label{lem:a8}~
    \begin{itemize}
        \item $\forall \epsilon \in (0, 1), \exists N > 0, \forall \ket{\psi}\in \cH_P, \forall n > N,$ \[\tr(P_{\cY}\cG^n(\ketbra{\psi}{\psi})) < \epsilon.\]
        \item $\forall \epsilon \in (0, 1). \exists N > 0.  (\cG^*)^{N}(P_{\cY}) \sqsubseteq \epsilon P_{\cY}$.
    \end{itemize}
\end{lemma}
\begin{proof}~
    \begin{itemize}
        \item For any $\ket{\psi}\in \cH_P$, we first prove that 
        \[\lim_{n\to\infty}\tr(P_{\cY}\cG^n(\ketbra{\psi}{\psi})) =0.\]
        With $\cG^*(P_{\cX}) \sqsupseteq p_{\cX}$ in Lemma~\ref{lem:a7} and $\cG^*(P_{\cX} + P_{\cY}) = \cG^*(I_{\cH_P}) \sqsubseteq I_{\cH_P} =  P_{\cX} + P_{\cY}$ by the definition of $\cG$, we have $\cG^*(P_{\cY}) \sqsubseteq P_{\cY}$, which means $\tr(P_{\cY}\cG^n(\ketbra{\psi}{\psi})) \geq 0$ is non-increasing with $n \to \infty$, then there exists $a \geq 0$ such that
        \[ \lim_{n\to\infty}\tr(P_{\cY}\cG^n(\ketbra{\psi}{\psi})) =a.\]
        Therefore
        \begin{align*}
            0 = \tr(P_{\cY}\sigma)
            &= \tr(P_{\cY}\lim_{n\to\infty}\frac{1}{n}\sum_{k=0}^n\cG^k(I_{\cH_P})) \\
            &= \lim_{n\to\infty}\frac{1}{n}\sum_{k=0}^n\tr(P_{\cY}\cG^k(I_{\cH_P} - \ketbra{\psi}{\psi})) + \lim_{n\to\infty}\frac{1}{n}\sum_{k=0}^n\tr(P_{\cY}\cG^n(\ketbra{\psi}{\psi})) \\
            &\geq \lim_{n\to\infty}\frac{1}{n}\sum_{k=0}^n\tr(P_{\cY}\cG^n(\ketbra{\psi}{\psi})) \\
            &= \lim_{n\to\infty}\tr(P_{\cY}\cG^n(\ketbra{\psi}{\psi})) \\
            &=a \geq 0,
        \end{align*}
        which results $a = 0$. We define a series of continuous functions $f_n, n\in \mathbb{N}$ on $A = \lbrace \ket{\psi} \mid \ket{\psi}\in \cH_P\rbrace$,
        \[f_n(\ket{\psi}) = \tr(P_{\cY}\cG^n(\ketbra{\psi}{\psi})).\]
        We have that $f_n$ is monotonically decreasing and convergent to $0$. Besides, $\cH_P$ is finite-dimensional, then $A$ is a compact set (unit sphere). By Dini's Theorem~\cite[Theorem 7.13]{rudin1976principles}, $f_n$ is uniform convergent to $0$, which is 
        \[ \forall \epsilon>0. \exists N > 0. \forall \ket{\psi}\in A. \forall n > N. \abs{f_n(\ket{\psi})} < \epsilon.\]
        Therefore, $\forall \epsilon \in (0,1), \exists N > 0, \forall \ket{\psi}\in \cH_P, \forall n > N,$ \[\tr(P_{\cY}\cG^n(\ketbra{\psi}{\psi})) < \epsilon.\]
        \item According to the above, for any $\epsilon \in (0,1)$, there exists $N_0 >0$ such that $\forall \ket{\varphi} \in \cY \subseteq \cH_P$ and $N = N_0 +1$, we have
        \[ \tr(P_{\cY}\cG^N(\ketbra{\varphi}{\varphi})) \leq \epsilon.\]
        which is 
        \[ \bra{\varphi}(\cG^*)^N(P_{\cY})\ket{\varphi} \leq \epsilon.\]
        Consider any $\ket{\psi} \in \cX$, in the proof of Lemma~\ref{lem:a7} we already know that
        \begin{align*}
            1 &= \tr(P_{\cX}\cG^n(\ketbra{\psi}{\psi})) \leq \tr((P_{\cX} + P_{\cY})\cG^n(\ketbra{\psi}{\psi})) \leq 1.
        \end{align*}
        As the same in the proof of Lemma~\ref{lem:a7}, we also have 
        \[ \ket{\psi} = (\cG^*)^N(P_{\cX})\ket{\psi} = (\cG^*)^N(P_{\cX} +P_{\cY})\ket{\psi}.\]
        Thus for any $\ket{\psi} \in \cX$, $(\cG^*)^N(P_{\cY})\ket{\psi} = 0\ket{\psi}$ (zero vector).

        Now, consider any $\ket{\alpha} = a\ket{\psi} + b \ket{\varphi} \in \cH_P$, where $\ket{\psi} \in \cX, \ket{\varphi} \in \cY, \abs{a}^2 + \abs{b}^2 = 1$, we have
        \begin{align*}
            \bra{\alpha}((\cG^*)^N(P_{\cY}) - \epsilon P_{\cY})\ket{\alpha}
            ={}&
            a\bar{a}\bra{\psi}(\cG^*)^N(P_{\cY})\ket{\psi} + b\bar{b}\bra{\varphi}(\cG^*)^N(P_{\cY})\ket{\varphi} - \epsilon b\bar{b} \\
            & \quad + \bar{a}b\bra{\psi}(\cG^*)^N(P_{\cY})\ket{\varphi} + a\bar{b}\bra{\varphi}(\cG^*)^N(P_{\cY})\ket{\psi} \\
            ={}& b\bar{b}\bra{\varphi}(\cG^*)^N(P_{\cY})\ket{\varphi} - \epsilon b\bar{b} \\
            \leq{}& \epsilon b\bar{b} - \epsilon b\bar{b} = 0.
        \end{align*} 
        Then, $(\cG^*)^N(P_{\cY}) \sqsubseteq \epsilon P_{\cY}$.  We finally conclude that 
        \[\forall \epsilon \in (0, 1). \exists N > 0.  (\cG^*)^{N}(P_{\cY}) \sqsubseteq \epsilon P_{\cY}.\]
    \end{itemize}
\end{proof}

\begin{proof}
    [Proof of Lemma~\ref{lem:a9}]
    From Lemma~\ref{lem:a7}, we have $(\cG^*)(P_{\cX}) \sqsupseteq P_{\cX}$, which is 
    \[ \cE_1^*\circ\cE^*(P_{\cX}) \sqsupseteq P_{\cX}\]
    with $I_{\cH_P} \sqsupseteq \cE^*(I_{\cH_P}) \sqsupseteq \cE^*(P_{\cX})$, it results $\cE^*_1(I_{\cH_P}) \sqsupseteq P_{\cX}$.
    Since 
    \[\cE^*_0(I_{\cH_P}) + \cE^*_1(I_{\cH_P}) = I_{\cH_P} = P_{\cX} + P_{\cY}.\]
    we have
    \[ \cE^*_0(I_{\cH_P}) \sqsubseteq P_{\cY}.\]
    By Lemma~\ref{lem:a8}, for any $\epsilon\in (0,1)$, there exists $N > 0$ such that $(\cG^*)^N(P_{\cY}) \sqsubseteq \epsilon P_{\cY}$. Therefore, for any $n\in \mathbb{N}$, any $\rho \in \cD(\cH_P)$, we have
    \begin{align*}
        \tr(\cE_0\circ(\cE\circ\cE_1)^n(\rho))
        ={}& \tr((\cG^*)^n(\cE_0^*(I_{\cH_P})) \rho) \\
        \leq{}& \tr((\cG^*)^n(P_{\cY})\rho) \\
        ={}& \tr((\cG^*)^{\lfloor \frac{n}{N}\rfloor * N + (n - \lfloor \frac{n}{N}\rfloor * N)}(P_{\cY})\rho) \\
        \leq{}& \epsilon^{\lfloor \frac{n}{N}\rfloor}\tr(P_{\cY}\rho) \\
        \leq{}& \epsilon^{\lfloor \frac{n}{N}\rfloor} \tr(\rho).
    \end{align*}
\end{proof}

\subsection{Proof of Lemma~\ref{lem:a4}}
\label{prf:a4}
\begin{proof}
    \begin{align*}
        \frac{\mathrm{d}}{\mathrm{d}\theta}(\cE_2\circ\cE_{H,\theta}\circ\cE_1(\rho)) 
        ={}& \frac{\mathrm{d}}{\mathrm{d}\theta}\left(\cE_2\left(e^{-i\theta H}\cE_1(\rho)e^{i\theta H}\right)\right) \\
        ={}& \cE_2\left(-iH e^{-i\theta H}\cE_1(\rho)e^{i\theta H} + e^{-i\theta H}\cE_1(\rho)e^{i\theta H}(i H)\right) \\
        ={}& \cE_2\left( e^{-i\theta H}(-i H\cE_1(\rho))e^{i\theta H} + e^{-i\theta H}(\cE_1(\rho)(i H))e^{i\theta H}\right) \\
        ={}& \cE_2\circ\cE_{H,\theta}(-iH\cE_1(\rho) + i\cE_1(\rho)H) \\
        ={}& \cE_2\circ\cE_{H,\theta}(-i[H, \cE_1(\rho)]).
    \end{align*}
\end{proof}

\subsection{Proof of Commutator-form Rule}
The following lemma modified from qPCA (quantum principal component analysis)~\cite{Lloyd_2014} helps incorporate commutators into the semantics of parameterized quantum programs. 

\begin{lemma}[Modified from~\cite{Lloyd_2014}]\label{lem:a3}
    Let $\cH_1,\, \cH_2,\, \cH_3$ be Hilbert spaces with $\dim(\cH_2) = \dim(\cH_3)$, $S$ a SWAP operator on $\cH_2\otimes\cH_3$,  $\rho \in \cD(\cH_1\otimes\cH_2)$,  $\sigma\in\cD(\cH_3)$, and parameter $\alpha\in\mathbb{R}$. Then 
    \begin{equation}\label{eq:commutator}
        \begin{aligned}
        \tr_3(e^{-i\alpha S}\rho\otimes\sigma e^{i\alpha S})
        ={}& \cos(\alpha)^2\rho + \sin(\alpha)^2\tr_{2}(\rho)\otimes\sigma  -i\cos(\alpha)\sin(\alpha)[[\sigma]_2, \rho]
        \end{aligned}
    \end{equation}
    where $[\sigma]_2$ denotes the operator $I_1\otimes \sigma$ on $\cH_1\otimes \cH_2$ with $I_1$ be the identity on $\cH_1$, and $\tr_{j}$ denotes the partial trace on $\cH_j$.
\end{lemma}
\begin{proof}
    \begin{align*}
        & \tr_3(e^{-i\alpha S}\rho\otimes\sigma e^{i\alpha S}) \\
        =& \tr_3((\cos(\alpha)I -i\sin(\alpha)S)\rho\otimes\sigma(\cos(\alpha)I+i\sin(\alpha)S)) \\
        =& \tr_3\Bigl(\cos(\alpha)^2\rho\otimes\sigma + \sin(\alpha)^2S\rho\otimes\sigma S + i\cos(\alpha)\sin(\alpha)\rho\otimes\sigma S - i\cos(\alpha)\sin(\alpha) S\rho\otimes\sigma\Bigr) \\
        =& \cos(\alpha)^2\rho + \sin(\alpha)^2\tr_{2}(\rho)\otimes\sigma + i\cos(\alpha)\sin(\alpha)\rho[\sigma]_{2} - i\cos(\alpha)\sin(\alpha)[\sigma]_2\rho \\
        =& \cos(\alpha)^2\rho + \sin(\alpha)^2\tr_2(\rho)\otimes\sigma -i\cos(\alpha)\sin(\alpha)[[\sigma]_2, \rho].
    \end{align*}
\end{proof}

With Lemma~\ref{lem:a3}, the $g(\theta;\alpha)$ in Eq.~(\ref{eq:swap_g}) can be rewritten as
\begin{align*}
g(\theta;\alpha) =& \cos(\alpha)^2\tr\left(O \cE_2\left(e^{-i\theta\sigma}\cE_1(\rho)
e^{i\theta\sigma}\right)\right) + \cos(\alpha)^2\tr\left(O \cE_2\left(e^{-i\theta\sigma}(\sigma \otimes \tr_1(\cE_1(\rho)))
e^{i\theta\sigma}\right)\right) \\
& \quad + \cos(\alpha)\sin(\alpha)\tr\left(O \cE_2\left(e^{-i\theta\sigma}(-i[\sigma\otimes I, \cE_1(\rho)]) e^{i\theta\sigma}\right)\right)
\end{align*}
where $\tr_1(\cE_1(\rho))$ is partial trace of $\cE_1(\rho)$ over the space of $e^{-i\theta\sigma}$ acts. Thus, 
\begin{align*}
    g(\theta;\alpha)-g(\theta;-\alpha) &= 2\cos(\alpha)\sin(\alpha)\tr\left(O \cE_2\left(e^{-i\theta\sigma}(-i[\sigma\otimes I, \cE_1(\rho)]) e^{i\theta\sigma}\right)\right) = \sin(2\alpha) \frac{\mathrm{d}}{\mathrm{d} \theta}f(\theta).
\end{align*}

\subsection{Proof of Theorem~\ref{thm:exact}}\label{prf:exact}
To prove Theorem~\ref{thm:exact}, we need the Super-operator-Valued Transition Systems~\cite{ying2017invariant}, which provides us with a convenient way for
modeling the control flow of quantum programs. In there, we use a modified version.
\begin{definition}
    [Modified Super-operator-Valued Transition Systems]\label{dfn:A.3}
    A modified super-operator-valued transition system (mSVTS for short) is a $5$-tuple $\cS = \langle \cH, L, l_0, \cT, \rho_0\rangle$, where:
    \begin{itemize}
        \item $\cH$ is a Hilbert space called the state space;
        \item $L$ is a finite set of locations;
        \item $l_0\in L$ is the initial location;
        \item $\cT$ is a set of transitions. Each transition $\tau \in \cT$ is a triple $\tau = \langle l, l', \cE \rangle$, often written as $\tau = l \overto{\cE} l'$, where $l,l' \in L$ are pre- and post-locations of $\tau$, respectively, and $\cE$ is a super-operator in $\cH$. It is required that
        \begin{equation}\label{eq:12}
            \sum \multiset{\cE^*(I_{\cH}): l\overto{\cE}l'\in \cT} \sqsubseteq I_{\cH}
        \end{equation}
        for each $l\in L$, where $I_{\cH}$ is the identify operator on $\cH$ and $\cE^*$ is the Schr\"{o}dinger-Heisenberg dual of $\cE$. 
        \item $\rho_0$ is an initial state at $l_0$.
    \end{itemize}
\end{definition}

For any path $\pi = l_1 \overto{\cE_1} l_2 \overto{\cE_2} \cdots \overto{\cE_{n-1}} l_n$ in the mSVTS graph, we write $l_1 \overset{\pi}{\Rightarrow} l_n$ and use $\cE_{\pi}$ to denote the composition of the super-operator along the path, i.e. $\cE_{\pi} = \cE_{n-1}\circ \cdots\circ\cE_2\circ\cE_1$.
If the transition $l \overto{\cE} l'$ in $\cT$ has superoperator $\cE$ simply defined by an operator $E$, i.e. $\cE(\rho) = E\rho E^{\dagger}$ for all density operator $\rho$ in $\cH$, we will write $l \overto{E} l'$ for $l \overto{\cE} l'$.

\begin{lemma}\label{lem:a11}
    Let $A$ be a set of paths in $\cS = \langle \cH, L, l_0, \cT, \rho_0\rangle$. All paths in $A$ have a same initial location and each path $\pi\in A$ is not a prefix of others in $A$, then for any $\rho \in \cD(\cH)$,
    \[\sum_{\pi\in A}\tr(\cE_{\pi}(\rho)) \leq \tr(\rho).\]
\end{lemma}
\begin{proof}
    We first assume that $A$ is finite and prove it by induction through the size of $\abs{A}$.
    \begin{itemize}
        \item $\abs{A} = 1$. $A$ has only one element $\pi$. We write $\pi = \pi = l_1 \overto{\cE_1} l_2 \overto{\cE_2} \cdots \overto{\cE_{n-1}} l_n$, then by Formula~(\ref{eq:12}), for any $\rho \in \cD(\cH)$,
        \begin{align*}
            \tr(\cE_{\pi}(\rho)) &= \tr(\cE_{n-1}\circ \cdots\circ\cE_2\circ\cE_1(\rho)) \\
            &= \tr(\cE_{n-1}^*(I_{\cH})\cdot \cE_{n-2}\circ\cdots\circ\cE_2\circ\cE_1(\rho)) \\
            &\leq \tr(\cE_{n-2}\circ\cdots\circ\cE_2\circ\cE_1(\rho)) \\
            &\cdots \\
            &\leq \tr(\rho).
        \end{align*}
        If $\abs{A} = 0$, then $\sum{\pi\in A}\tr(\cE_{\pi}(\rho)) = 0 \leq \tr(\rho)$.
        \item Suppose when $\abs{A} \leq n, n\geq 1$, we have that for any $\rho \in \cD(\cH)$,
        \[ \sum_{\pi\in A}\tr(\cE_{\pi}(\rho)) \leq \tr(\rho).\]
        Then consider $\abs{A} = n + 1$, we choose a path $\pi = l_1 \overto{\cE_1} l_2 \overto{\cE_2} \cdots \overto{\cE_{n-1}} l_n \in A$ and let $\pi_j = l_{j}\overto{\cE_{j}}l_{j+1}, 1\leq j\leq n-1$, then $\pi = \pi_1\pi_2\cdots\pi_{n-1}$. For convenience, we use $\pi_0$ to denote an empty path. Then for this $\pi$, we define
        \[ B = \lbrace j : 0\leq j \leq n-1, \forall \pi \in A. \exists \pi'. s.t. \pi = \pi_0\pi_1\pi_2\cdots\pi_j\pi' \rbrace.\]
        $B$ must contains $0$, thus $B$ is not empty. As each path in $A$ is not a prefix of others in $A$, we have that $n-1 \not\in B$. Let $j_0 = \max B$, then $j_0 < n-1$. Consider all the transitions in $\cT$ with $l_{j_0}$ as pre-location: 
        \[\tau_1 = l_{j_0}\overto{\cG_1} l_{1}', \tau_2 = l_{j_0}\overto{\cG_2} l_{2}', \ldots, \tau_{n'} = l_{j_0}\overto{\cG_{n'}} l_{n'}'.\]
        It is followed that $\forall \pi' \in A$, $\pi'$ must have a prefix $\pi_0\pi_1\cdots\pi_{j_0}\tau_{k}$ with $1\leq k\leq n'$, otherwise $\pi' = \pi_0\pi_1\cdots\pi_{j_0}$ (If $j_0 >0$, then $\pi' \neq \pi_0$. If $j_0 = 0$, then this $\pi'$ doesn't exist), it is a prefix for all paths in $A$, which is a contradiction. Therefore
        \[ A = \bigcup_{k=1}^{n'} C_{k}\]
        with
        \[ C_k \equiv \lbrace \pi' \in A: \exists \pi''. s.t. \pi' = \pi_0\pi_1\cdots\pi_{j_0}\tau_{k}\pi'' \rbrace, \]
        for $1\leq k\leq n'$, and each path in $C_k$ is not a prefix of of others in $C_k$. We claim that $\abs{C_k} \leq n$, otherwise there exists $ C_{k_0} = A$, then $ \pi_0\pi_1\cdots\pi_{j_0}\tau_{k_0}$ is a prefix for all paths in $A$, this is contradict to the definition of $j_0$, because in that case $j_0 + 1 \in B$. Then define
        \[ D_k \equiv \lbrace \pi'' : \pi_0\pi_1\cdots\pi_{j_0}\tau_{k}\pi'' \in C_k\rbrace, \enspace 1\leq k\leq n'\]
        we have that for any $1\leq k\leq n'$, all paths in $D_k$ have same initial location and each path $\pi''\in D_k$ is not a prefix of others in $D_k$ and $\abs{D_k} = \abs{C_k} \leq n$. By inductive hypothesis, for any $\rho\in \cD(\cH)$,
        \begin{equation}\label{eq:13}
            \sum_{\pi''\in D_k}\tr(\cE_{\pi''}(\rho)) \leq \tr(\rho), \enspace 1\leq k\leq n'.
        \end{equation}

        Finally, for any $\rho \in \cD(\cH)$,
        \begin{align*}
            \sum_{\pi'\in A} \tr(\cE_{\pi}(\rho)) &= \sum_{k=1}^{n'}\sum_{\pi'\in C_k}\tr\left(\cE_{\pi}(\rho)\right) \\
            &= \sum_{k=1}^{n'}\sum_{\pi''\in D_k}\tr\left(\cE_{\pi_0\pi_1\cdots\pi_{j_0}\tau_{k}\pi''}(\rho)\right) \\
            &= \sum_{k=1}^{n'}\sum_{\pi''\in D_k}\tr\left(\cE_{\pi''}\left(\cE_{\pi_0\pi_1\cdots\pi_{j_0}\tau_{k}}(\rho)\right)\right) \\
            &\leq \sum_{k=1}^{n'} \tr\left(\cE_{\pi_0\pi_1\cdots\pi_{j_0}\tau_{k}}(\rho)\right) \tag{by Inequality~(\ref{eq:13})}\\
            &= \sum_{k=1}^{n'}\tr\left(\cE_{\tau_k}\left(\cE_{\pi_0\pi_1\cdots\pi_{j_0}}(\rho)\right)\right) \\
            &= \sum_{k=1}^{n'}\tr\left(\cG_k\left(\cE_{\pi_0\pi_1\cdots\pi_{j_0}}(\rho)\right)\right) \\
            &\leq \tr\left(\cE_{\pi_0\pi_1\cdots\pi_{j_0}}(\rho)\right) \tag{by Formula~(\ref{eq:12})} \\
            &\leq \tr(\rho). \tag{as the same proof of $\abs{A} = 1$}
        \end{align*}
    \end{itemize}
    Thus, when $\abs{A}$ is finite, we prove this proposition. Because of the order-preserving property of limitation, when $\abs{A} = \infty$, we also have the same result.
\end{proof}

\begin{definition}
    [Computation path of mSVTS]\label{dfn:A.5}
    A path $\pi = l_0 \overto{\cE_1} m_1 \overto{\cE_2} \cdots \overto{\cE_{n}} m_n $ in $\cS = \langle \cH, L, l_0, \cT, \rho_0\rangle$ is a computation path if for any $\tau$ in $\cT$, $m_n$ is not a pre-location of $\tau$. And we write $\Pi_{\cS}$ be the set of computation paths, which is
    \begin{align*}
        \Pi_{\cS} &= \lbrace \pi \text{ is a path in } \cS \mid l_0 \overset{\pi}{\Rightarrow} l \land \forall \tau \in \cT, l \text{ is not a pre-location of } \tau\rbrace.
    \end{align*}
    Moreover, we use $\Pi_{\cS}^{(n)}$ to denote the set of length $\leq n$ (transits $\leq n$ steps) paths in $\Pi_{\cS}$.
\end{definition}

As similar in~\cite{ying2017invariant}, the control flow graph of a quantum program can be represented by an mSVTS. For every parameterized quantum $\qwhile$-program $P(\bm{\theta})$, we define an mSVTS $\cS_{P(\bm{\theta})} = \langle \cH_{P(\bm{\theta})}, L, l_{in}^{P(\bm{\theta})}, \cT, \rho\rangle$ in the state Hilbert space $\cH_{P(\bm{\theta})}$ of $P(\bm{\theta})$ by induction on the program structure of $P(\bm{\theta})$, where $\rho$ is an input state of $P(\bm{\theta})$.
This transition system has two designated locations $l^P_{in}, l^P_{out}$, with the former being the initial location and the latter being the exit location.
We only need to consider definitions of $L$ and $\cT$.

\begin{itemize}
    \item $P(\bm{\theta}) \equiv \qskip$. $\cS_{P(\bm{\theta})}$ has only two locations $l_{in}^{P(\bm{\theta})}, l_{out}^{P(\bm{\theta})}$ and a single transition $l_{in}^{P(\bm{\theta})} \overto{I} l_{out}^{P(\bm{\theta})}$.
    \item $P(\bm{\theta}) \equiv \qinit{q}$. Let $\lbrace \ket{n}_q\rbrace$ be the basis of $\cH_{q}$, then $\cS_{P_{\bm{\theta}}}$ has two locations $l_{in}^{P(\bm{\theta})}, l_{out}^{P(\bm{\theta})}$ and a single transition $l_{in}^{P(\bm{\theta})} \overto{\cE_{q}} l_{out}^{P(\bm{\theta})}$, where $\cE_{q}(\rho) = \sum_n\ketbra[q]{0}{n}\rho\ketbra[q]{n}{0}$.
    \item $P(\bm{\theta}) \equiv \qinitd{\sigma}{\bar{q}}$.\footnote{Although this statement is not contained in the syntax, we use it for convenience as we have said in Remark~\ref{rmk:1}} Write $\sigma$ in spectral decomposition, $\sigma = \sum_{n}\lambda_n\ketbra[\bar{q}]{\psi_n}{\psi_n}, \lambda_n \geq 0$ (there we include eigenvalues of $0$), then $\cS_{P_{\bm{\theta}}}$ has two locations $l_{in}^{P(\bm{\theta})}, l_{out}^{P(\bm{\theta})}$ and a single transition $l_{in}^{P(\bm{\theta})} \overto{\cE_{\bar{q},\sigma}} l_{out}^{P(\bm{\theta})}$, where \[\cE_{\bar{q}, \sigma}(\rho) = \sum_{m,n}(\sqrt{\lambda_m}\ketbra[\bar{q}]{\psi_m}{\psi_n})\rho(\sqrt{\lambda_m}\ketbra[\bar{q}]{\psi_n}{\psi_m}).\]
    \item $P(\bm{\theta}) \equiv \qut{U}{\bar{q}}$. $\cS_{P(\bm{\theta})}$ has two locations $l_{in}^{P(\bm{\theta})}, l_{out}^{P(\bm{\theta})}$ and a single transition $l_{in}^{P(\bm{\theta})} \overto{[U]_{\bar{q}}} l_{out}^{P(\bm{\theta})}$.
    \item $P(\bm{\theta}) \equiv \qutp{\theta}{\sigma}{\bar{q}}$. $\cS_{P(\bm{\theta})}$ has two locations $l_{in}^{P(\bm{\theta})}, l_{out}^{P(\bm{\theta})}$ and a single transition $l_{in}^{P(\bm{\theta})} \overto{[e^{-i\theta\sigma}]_{\bar{q}}} l_{out}^{P(\bm{\theta})}$.
    \item $P(\bm{\theta}) \equiv P_1(\bm{\theta});P_2(\bm{\theta})$. Suppose $\cS_{P_1(\bm{\theta})}$, $\cS_{P_2(\bm{\theta})}$ are the control flow graphs of subprograms $P_1(\bm{\theta})$, $P_2(\bm{\theta})$, respectively. Then $\cS_{P(\bm{\theta})}$ in constructed as follows: we identify $l_{out}^{P_1(\bm{\theta})} = l_{in}^{P_2(\bm{\theta})}$ and concatenate $\cS_{P_1(\bm{\theta})}$, $\cS_{P_2(\bm{\theta})}$. We further set $l_{in}^{P_1(\bm{\theta})} = l_{in}^{P(\bm{\theta})}$, $l_{out}^{P_2(\bm{\theta})} = l_{out}^{P(\bm{\theta})}$.
    \item $P(\bm{\theta}) \equiv \qif (\Box m\cdot M[\bar{q}] = m \to P_{m}(\bm{\theta})) \qfi$. Suppose that $\cS_{P_m(\bm{\theta})}$ is the control flow graph of subprogram $P_m(\bm{\theta})$ for every $m$. Then $\cS_{P(\bm{\theta})}$ is constructed as follows: we put all $\cS_{P_m(\bm{\theta})}$'s together, and add a new location $l_{in}^{P(\bm{\theta})}$ and a transition $l_{in}^{P(\bm{\theta})} \overto{[M_m]_{\bar{q}}} l_{in}^{P_m(\bm{\theta})}$ for every $m$. We further set $l_{out}^{P_m(\bm{\theta})} = l_{out}^{P(\bm{\theta})}$ for all $m$.
    \item $P(\bm{\theta}) \equiv \qqwhile{\bar{q}}{Q(\bm{\theta})}$. We construct $\cS_{P(\bm{\theta})}$ from the control flow graph $\cS_{Q(\bm{\theta})}$ of subprogram $Q(\bm{\theta})$ as follows: we add two new locations $l_{in}^{P(\bm{\theta})}$, $l_{out}^{P(\bm{\theta})}$ and two transitions $l_{in}^{P(\bm{\theta})} \overto{[M_0]_{\bar[q]}} l_{out}^{P(\bm{\theta})}$, $l_{in}^{P(\bm{\theta})} \overto{[M_1]_{\bar{q}}} l_{in}^{Q(\bm{\theta})}$. We further identify $l_{out}^{Q(\bm{\theta})} = l_{in}^{P(\bm{\theta})}$.
\end{itemize}

There, we use the subscript $[U]_{\bar{q}}$ for unitary $U$ and quantum variables $\bar{q}$ to indicate that $U$ acts on the Hilbert space $\cH_{\bar{q}}$.

\begin{theorem}\label{thm:a1}
    For a parameterized quantum \textbf{while}-program $P(\bm{\theta})$ with an initial state $\rho$ and its mSVTS $ \cS_{P(\bm{\theta})}$, we have
    \[ \sem{P(\bm{\theta})}(\rho) = \sum_{\pi\in\Pi_{\cS_{P(\bm{\theta})}}}\cE_{\pi}(\rho) \equiv \bigsqcup_{n=1}^{\infty}\sum_{\pi\in\Pi_{\cS_{P(\bm{\theta})}^{(n)}}}\cE_{\pi}(\rho).\]
\end{theorem}
\begin{proof}
    We prove it by induction through the program structure.
    \begin{itemize}
        \item $P(\bm{\theta}) \equiv \qskip$. We have $\Pi_{\cS_{P(\bm{\theta})}} = \lbrace l_{in}^{P(\bm{\theta})} \overto{I} l_{out}^{P(\bm{\theta})} \rbrace$, then \[\sum_{\pi\in\Pi_{\cS_{P(\bm{\theta})}}}\cE_{\pi}(\rho) = I\rho I = \rho = \sem{\qskip}(\rho).\]
        \item $P(\bm{\theta}) \equiv \qinit{q}$. We have $\Pi_{\cS_{P(\bm{\theta})}} = \lbrace l_{in}^{P(\bm{\theta})} \overto{\cE_{q}} l_{out}^{P(\bm{\theta})} \rbrace$, then
        \begin{align*}
            \sum_{\pi\in\Pi_{\cS_{P(\bm{\theta})}}}\cE_{\pi}(\rho) &= \cE_q(\rho) = \sum_n\ketbra[q]{0}{n}\rho\ketbra[q]{n}{0} = \sem{\qinit{q}}(\rho).
        \end{align*}
        \item $P(\bm{\theta}) \equiv \qinitd{\sigma}{\bar{q}}$. We have $\Pi_{\cS_{P(\bm{\theta})}} = \lbrace l_{in}^{P(\bm{\theta})} \overto{\cE_{\bar{q},\sigma}} l_{out}^{P(\bm{\theta})} \rbrace$, then
        \begin{align*}
            \sum_{\pi\in\Pi_{\cS_{P(\bm{\theta})}}}\cE_{\pi}(\rho)
            = \cE_{\bar{q},\sigma}(\rho)
            ={} \sum_{m,n}(\sqrt{\lambda_m}\ketbra[\bar{q}]{\psi_m}{\psi_n})\rho(\sqrt{\lambda_m}\ketbra[\bar{q}]{\psi_n}{\psi_m}).
        \end{align*}
        In the Remark~\ref{rmk:1}, the statement $\qinitd{\sigma}{\bar{q}}$ aims to set the state $\rho$ to $\tr_{\bar{q}}(\rho)\otimes \sigma$, which should be $\sem{\qinitd{\sigma}{\bar{q}}}(\rho)$. We can check that for any $\rho$
        \[\sum_{m,n}(\sqrt{\lambda_m}\ketbra[\bar{q}]{\psi_m}{\psi_n})\rho(\sqrt{\lambda_m}\ketbra[\bar{q}]{\psi_n}{\psi_m}) = \tr_{\bar{q}}(\rho)\otimes \sigma.\]
        Thus, 
        \[\sum_{\pi\in\Pi_{\cS_{P(\bm{\theta})}}}\cE_{\pi}(\rho) = \sem{\qinitd{\sigma}{\bar{q}}}(\rho).\]
        \item $P(\bm{\theta}) \equiv \qut{U}{\bar{q}}$. We have $\Pi_{\cS_{P(\bm{\theta})}} = \lbrace l_{in}^{P(\bm{\theta})} \overto{[U]_{\bar{q}}} l_{out}^{P(\bm{\theta})} \rbrace$, then
        \[\sum_{\pi\in\Pi_{\cS_{P(\bm{\theta})}}}\cE_{\pi}(\rho) = [U]_{\bar{q}}\rho [U^{\dagger}]_{\bar{q}} = \sem{\qut{U}{\bar{q}}}(\rho).\]
        \item $P(\bm{\theta}) \equiv \qutp{\theta}{\sigma}{\bar{q}}$. We have $\Pi_{\cS_{P(\bm{\theta})}} = \lbrace l_{in}^{P(\bm{\theta})} \overto{[e^{-i\theta\sigma}]_{\bar{q}}} l_{out}^{P(\bm{\theta})} \rbrace$, then \[\sum_{\pi\in\Pi_{\cS_{P(\bm{\theta})}}}\cE_{\pi}(\rho) = [e^{-i\theta\sigma}]_{\bar{q}}\rho [e^{i\theta\sigma}]_{\bar{q}} = \sem{\qutp{\theta}{\sigma}{\bar{q}}}(\rho).\]
        \item $P(\bm{\theta}) \equiv P_1(\bm{\theta});P_2(\bm{\theta})$. For any $\pi_1\in\Pi_{\cS_{P_1(\bm{\theta})}}$, $\pi_2\in\Pi_{\cS_{P_2(\bm{\theta})}}$, we can write 
        \[\pi_1 = l_{in}^{P_1(\bm{\theta})} \overto{\cE_{1}^{P_1(\bm{\theta})}} l_1^{P_1(\bm{\theta})} \overto{\cE_2^{P_1(\bm{\theta})}} \cdots \overto{\cE_m^{P_1(\bm{\theta})}} l_{out}^{P_1(\bm{\theta})}\]
        \[\pi_2 = l_{in}^{P_2(\bm{\theta})} \overto{\cE_{1}^{P_2(\bm{\theta})}} l_1^{P_2(\bm{\theta})} \overto{\cE_2^{P_2(\bm{\theta})}} \cdots \overto{\cE_n^{P_2(\bm{\theta})}} l_{out}^{P_2(\bm{\theta})}\]
        then, from our construction
        \begin{align*}
            \pi ={}& l_{in}^{P(\bm{\theta})} \overto{\cE_{1}^{P_1(\bm{\theta})}} l_1^{P_1(\bm{\theta})} \overto{\cE_2^{P_1(\bm{\theta})}} \cdots \quad \overto{\cE_m^{P_1(\bm{\theta})}} l_{out}^{P_1(\bm{\theta})} \\
            & \quad \overto{\cE_{1}^{P_2(\bm{\theta})}} l_1^{P_2(\bm{\theta})} \overto{\cE_2^{P_2(\bm{\theta})}} \cdots \overto{\cE_n^{P_2(\bm{\theta})}} l_{out}^{P(\bm{\theta})} \in \Pi_{\cS_{P(\bm{\theta})}}.
        \end{align*}
        For convenience, we write $\pi = \pi_{1}\pi_{2}$, then
        \[ \lbrace \pi_1\pi_2 \mid \pi_1\in\Pi_{\cS_{P_1(\bm{\theta})}}, \pi_2\in\Pi_{\cS_{P_2(\bm{\theta})}} \rbrace \subseteq \Pi_{\cS_{P(\bm{\theta})}}.\]
        On the other hand, for any $\pi\in\Pi_{\cS_{P(\bm{\theta})}}$, write
        \[\pi = l_{in}^{P(\bm{\theta})} \overto{\cE_{1}^{P(\bm{\theta})}} l_1^{P(\bm{\theta})} \overto{\cE_2^{P(\bm{\theta})}} \cdots \overto{\cE_m^{P(\bm{\theta})}} l_{out}^{P(\bm{\theta})}.\]
        From the construction, we have that $l_{in}^{P(\bm{\theta})} = l_{in}^{P_1(\bm{\theta})}$, $l_{out}^{P(\bm{\theta})} = l_{out}^{P_2(\bm{\theta})}$. Then we can define $k$ to be the first index such that $l_k^{P(\bm{\theta})}$ is in $\cS_{P_1(\bm{\theta})}$ and $l_{k+1}^{P(\bm{\theta})}$ is in $\cS_{P_2(\bm{\theta})}$.
        Moreover, for any location $l\neq l_{out}^{P_1(\bm{\theta})}$ in $\cS_{P_1(\bm{\theta})}$, its post-location is still in $\cS_{P_1(\bm{\theta})}$, then we have $l_k^{P(\bm{\theta})} = l_{out}^{P_1(\bm{\theta})}$. By construction, $l_{in}^{P_2(\bm{\theta})} = l_{out}^{P_1(\bm{\theta})} = l_k^{P(\bm{\theta})}$ and for any location $l$ in in $\cS_{P_2(\bm{\theta})}$, its post-location is still in $\cS_{P_2(\bm{\theta})}$, thus for all $j\geq k+1$, $l_{j}^{P(\bm{\theta})}$ is in $\cS_{P_2(\bm{\theta})}$. Then
        \[\pi_1 = l_{in}^{P(\bm{\theta})} \overto{\cE_{1}^{P(\bm{\theta})}} l_1^{P(\bm{\theta})} \overto{\cE_2^{P(\bm{\theta})}} \cdots \overto{\cE_k^{P(\bm{\theta})}} l_{k}^{P(\bm{\theta})} \in \Pi_{\cS_{P_1(\bm{\theta})}}\]
        \[\pi_2 = l_{k}^{P(\bm{\theta})} \overto{\cE_{k+1}^{P(\bm{\theta})}} l_{k+1}^{P(\bm{\theta})} \overto{\cE_{k+1}^{P(\bm{\theta})}} \cdots \overto{\cE_k^{P(\bm{\theta})}} l_{m}^{P(\bm{\theta})} \in \Pi_{\cS_{P_2(\bm{\theta})}}\]
        therefore
        \[ \lbrace \pi_1\pi_2 \mid \pi_1\in\Pi_{\cS_{P_1(\bm{\theta})}}, \pi_2\in\Pi_{\cS_{P_2(\bm{\theta})}} \rbrace \supseteq \Pi_{\cS_{P(\bm{\theta})}}.\]
        Because all $\cE_{\pi}(\rho)$ are positive semidefinite, we consider the partial summation, for any $n, m\geq 1$:
        \begin{equation}\label{eq:1}
            \begin{aligned}
                \sum_{\pi \in\Pi_{\cS_{P(\bm{\theta})}}^{(n)}}\cE_{\pi}(\rho) \sqsubseteq{}& \sum_{\pi \in\Pi_{\cS_{P_2(\bm{\theta})}}^{(n)}}\sum_{\pi \in\Pi_{\cS_{P_1(\bm{\theta})}}^{(n)}} \cE_{\pi_1\pi_2}(\rho)
                ={} \sum_{\pi \in\Pi_{\cS_{P_2(\bm{\theta})}}^{(n)}} \cE_{\pi_2}\left(\sum_{\pi \in\Pi_{\cS_{P_1(\bm{\theta})}}^{(n)}} \cE_{\pi_1}(\rho)\right)
            \end{aligned}
        \end{equation}
        \begin{equation}\label{eq:2}
            \begin{aligned}
                \sum_{\pi \in\Pi_{\cS_{P(\bm{\theta})}}^{(m+n)}}\cE_{\pi}(\rho) \sqsupseteq{}& \sum_{\pi \in\Pi_{\cS_{P_2(\bm{\theta})}}^{(m)}}\sum_{\pi \in\Pi_{\cS_{P_1(\bm{\theta})}}^{(n)}} \cE_{\pi_1\pi_2}(\rho)
                ={} \sum_{\pi \in\Pi_{\cS_{P_2(\bm{\theta})}}^{(m)}} \cE_{\pi_2}\left(\sum_{\pi \in\Pi_{\cS_{P_1(\bm{\theta})}}^{(n)}} \cE_{\pi_1}(\rho)\right)
            \end{aligned}
        \end{equation}
        By the inductive hypothesis, for any $\sigma$: 
        \[ \sum_{\pi \in\Pi_{\cS_{P_1(\bm{\theta})}}} \cE_{\pi}(\sigma) = \sem{P_1(\bm{\theta})}(\sigma)\]
        \[ \sum_{\pi \in\Pi_{\cS_{P_2(\bm{\theta})}}} \cE_{\pi}(\sigma) = \sem{P_2(\bm{\theta})}(\sigma)\]
        Therefore, in equation~(\ref{eq:1}), we have
        \begin{align*}
            \sum_{\pi \in\Pi_{\cS_{P(\bm{\theta})}}^{(n)}}\cE_{\pi}(\rho) &\sqsubseteq \sum_{\pi \in\Pi_{\cS_{P_2(\bm{\theta})}}^{(n)}} \cE_{\pi_2}\left(\sum_{\pi \in\Pi_{\cS_{P_1(\bm{\theta})}}^{(n)}} \cE_{\pi_1}(\rho)\right) \\
            & \sqsubseteq \sum_{\pi \in\Pi_{\cS_{P_2(\bm{\theta})}}^{(n)}} \cE_{\pi_2}\left( \sem{P_1(\bm{\theta})}(\rho)\right) \\
            & \sqsubseteq \sem{P_2(\bm{\theta})}(\sem{P_1(\bm{\theta})}(\rho)) \\
            & = \sem{P_1(\bm{\theta});P_2(\bm{\theta})}(\rho)
        \end{align*}
        and, in equation~(\ref{eq:2}), let $n \to \infty$, we have for any $m\geq 1$:
        \begin{align*}
            \sum_{\pi \in\Pi_{\cS_{P(\bm{\theta})}}}\cE_{\pi}(\rho) & \sqsupseteq \sum_{\pi \in\Pi_{\cS_{P_2(\bm{\theta})}}^{(m)}} \cE_{\pi_2}\left(\sem{P_1(\bm{\theta})}(\rho)\right)
        \end{align*}
        then, let $m\to \infty$, we get
        \begin{align*}
            \sum_{\pi \in\Pi_{\cS_{P(\bm{\theta})}}}\cE_{\pi}(\rho) & \sqsupseteq \sem{P_2(\bm{\theta})}(\sem{P_1(\bm{\theta})}(\rho)) = \sem{P_1(\bm{\theta});P_2(\bm{\theta})}(\rho).
        \end{align*}
        Thus, $\sum_{\pi \in\Pi_{\cS_{P(\bm{\theta})}}}\cE_{\pi}(\rho) = \sem{P_1(\bm{\theta});P_2(\bm{\theta})}(\rho)$.
        \item $P(\bm{\theta}) \equiv \qif (\Box m\cdot M[\bar{q}] = m \to P_{m}(\bm{\theta})) \qfi$. Let $\pi_m = l_{in}^{P(\bm{\theta})} \overto{[M_m]_{\bar{q}}} l_{in}^{P_m(\bm{\theta})}$, then
        \[ \Pi_{\cS_{P(\bm{\theta})}} = \bigcup_{m} \lbrace \pi_m\pi \mid \pi \in\Pi_{\cS_{P_m(\bm{\theta})}}\rbrace\]
        therefore
        \begin{align*}
            \sum_{\pi\in\Pi_{\cS_{P(\bm{\theta})}}}\cE_{\pi}(\rho) &= \sum_m\sum_{\pi\in\Pi_{\cS_{P_m(\bm{\theta})}}}\cE_{\pi_m\pi}(\rho) \\
            &= \sum_m\sum_{\pi\in\Pi_{\cS_{P_m(\bm{\theta})}}} \cE_{\pi}(\cE_{\pi_m}(\rho)) \\
            &= \sum_m \sum_{\pi\in\Pi_{\cS_{P_m(\bm{\theta})}}} \cE_{\pi}([M_m]_{\bar{q}}\rho[M_m^{\dagger}]_{\bar{q}}) \\
            &= \sum_m \sem{P_m(\bm{\theta})}([M_m]_{\bar{q}}\rho[M_m^{\dagger}]_{\bar{q}}) \\
            &= \sem{\qif (\Box m\cdot M[\bar{q}] = m \to P_{m}(\bm{\theta})) \qfi}(\rho).
        \end{align*}
        In there, the second-to-last equality is provided by the inductive hypothesis.
        \item $P(\bm{\theta}) \equiv \qqwhile{\bar{q}}{Q(\bm{\theta})}$. Let $\pi_0 = l_{in}^{P(\bm{\theta})} \overto{[M_0]_{\bar[q]}} l_{out}^{P(\bm{\theta})}$, $\pi_1 = l_{in}^{P(\bm{\theta})} \overto{[M_1]_{\bar{q}}} l_{in}^{Q(\bm{\theta})}$. As same before, we can obtain that any path in $\Pi_{\cS_{P(\bm{\theta})}}$ has the form $\eta_0$ or $\pi_1\eta_1\pi_1\eta_2\cdots\pi_1\eta_n\pi_0$, $\eta_j\in\cS_{Q(\bm{\theta})}, j=1,2,\ldots, n$, which is
        \begin{align*}
            \Pi_{\cS_{P(\bm{\theta})}}
            =& \lbrace \pi_0\rbrace \cup \bigl\lbrace \pi_1\eta_1\pi_1\eta_2\cdots\pi_1\eta_n\pi_0 \mid (n\in \mathbb{N}_{+}) \land(\forall 1\leq j\leq n. \eta_j\in\Pi_{\cS_{Q(\bm{\theta})}})\bigr\rbrace,
        \end{align*}
        then, for any $m, n, n_j\geq 1, j=1,2,\ldots,m$:
        \begin{align*}
            \Pi_{\cS_{P(\bm{\theta})}}^{(mn)}
            \subseteq{}& \lbrace \pi_0\rbrace \cup\bigl\lbrace \pi_1\eta_1\pi_1\eta_2\cdots\pi_1\eta_k\pi_0 \mid (k\in \mathbb{N}^{+}, k\leq mn) \land(\forall 1\leq l\leq k. \eta_l\in\Pi_{\cS_{Q(\bm{\theta})}}^{(mn)})\bigr\rbrace
        \end{align*}
        \begin{align*}
            \Pi_{\cS_{P(\bm{\theta})}}^{((\sum_{j=1}^mn_j)+m+1)}
            \supseteq{}& \lbrace \pi_0\rbrace \cup\bigl\lbrace \pi_1\eta_1\pi_1\eta_2\cdots\pi_1\eta_k\pi_0 \mid (k\in \mathbb{N}^{+}, k\leq m) \land(\forall 1\leq l\leq k. \eta_l\in\Pi_{\cS_{Q(\bm{\theta})}}^{(n_l)})\bigr\rbrace
        \end{align*}
        Thus
        \begin{align}
            & \sum_{\pi\in\Pi_{\cS_{P(\bm{\theta})}}^{(mn)}}\cE_{\pi}(\rho) \notag\\
            \sqsubseteq{} & \cE_{\pi_0}(\rho) + \sum_{k=1}^{mn}\sum_{\eta_1,\eta_2,\ldots\eta_k\in\Pi_{\cS_{Q(\bm{\theta})}}^{(mn)}}\cE_{\pi_1\eta_1\pi_1\eta_2\cdots\pi_1\eta_k\pi_0}(\rho) \notag\\
            ={}& \cE_{\pi_0}(\rho) + \sum_{k=1}^{mn}\cE_{\pi_0}\Biggggl(\sum_{\eta_k\in\Pi_{\cS_{Q(\bm{\theta})}}^{(mn)}}\cE_{\eta_k}\circ\cE_{\pi_1}\Biggggl( \cdots \left(\sum_{\eta_1\in\Pi_{\cS_{Q(\bm{\theta})}}^{(mn)}} \cE_{\eta_1}\circ\cE_{\pi_1}(\rho)\right)\cdots\Biggggr)\Biggggr) \notag\\
            \sqsubseteq{} & \cE_{\pi_0}(\rho) + \sum_{k=1}^{mn}\cE_{\pi_0}\Biggggl(\sum_{\eta_k\in\Pi_{\cS_{Q(\bm{\theta})}}}\cE_{\eta_k}\circ\cE_{\pi_1}\Biggggl( \cdots \left(\sum_{\eta_1\in\Pi_{\cS_{Q(\bm{\theta})}}} \cE_{\eta_1}\circ\cE_{\pi_1}(\rho)\right)\cdots\Biggggr)\Biggggr) \notag\\
            ={}& \cE_{\pi_0}(\rho) + \sum_{k=1}^{mn}\cE_{\pi_0}\left(\sem{Q(\bm{\theta})}\circ\cE_{\pi_1}\right)^k(\rho) \tag{\text{by inductive hypothesis}} \\
            \sqsubseteq{} & \sem{\qqwhile{\bar{q}}{Q(\bm{\theta})}}(\rho) \label{eq:a1}
        \end{align}
        and
        \begin{align*}
            & \sum_{\pi\in\Pi_{\cS_{P(\bm{\theta})}}^{((\sum_{j=1}^mn_j)+m+1)}}\cE_{\pi}(\rho) \\
            \sqsupseteq{}& \cE_{\pi_0}(\rho) + \sum_{k=1}^{m}\sum_{\eta_k\in\Pi_{\cS_{Q(\bm{\theta})}}^{(n_k)}}\cdots\sum_{\eta_1\in\Pi_{\cS_{Q(\bm{\theta})}}^{(n_1)}}\cE_{\pi_1\eta_1\pi_1\eta_2\cdots\pi_1\eta_k\pi_0}(\rho) \\
            ={}& \cE_{\pi_0}(\rho) + \sum_{k=1}^{m}\cE_{\pi_0}\Biggggl(\sum_{\eta_k\in\Pi_{\cS_{Q(\bm{\theta})}}^{(n_k)}}\cE_{\eta_k}\circ\cE_{\pi_1}\Biggggl( \cdots \left(\sum_{\eta_1\in\Pi_{\cS_{Q(\bm{\theta})}}^{(n_1)}} \cE_{\eta_1}\circ\cE_{\pi_1}(\rho)\right)\cdots\Biggggr)\Biggggr).
        \end{align*}
        Use the inductive hypothesis and let $n_1\to \infty$, we have
        \begin{align*}
            & \sum_{\pi\in\Pi_{\cS_{P(\bm{\theta})}}} \cE_{\pi}(\rho) \\
            \sqsupseteq{} & \cE_{\pi_0}(\rho) + \sum_{k=1}^{m}\cE_{\pi_0}\Biggggl(\sum_{\eta_k\in\Pi_{\cS_{Q(\bm{\theta})}}^{(n_k)}}\cE_{\eta_k}\circ\cE_{\pi_1}\Biggggl( \cdots \left(\sum_{\eta_2\in\Pi_{\cS_{Q(\bm{\theta})}}^{(n_2)}} \cE_{\eta_1}\circ\cE_{\pi_1}(\sem{Q(\bm{\theta})}\circ\cE_{\pi_1}(\rho))\right)\cdots\Biggggr)\Biggggr)
        \end{align*}
        as the same, let $n_2\to \infty, \ldots, n_m\to \infty$ in order, we get
        \[ \sum_{\pi\in\Pi_{\cS_{P(\bm{\theta})}}} \cE_{\pi}(\rho)\sqsupseteq \cE_{\pi_0}(\rho) + \sum_{k=1}^{m}\cE_{\pi_0}\left(\sem{Q(\bm{\theta})}\circ\cE_{\pi_1}\right)^k(\rho).\]
        then let $m\to \infty$,
        \[ \sum_{\pi\in\Pi_{\cS_{P(\bm{\theta})}}} \cE_{\pi}(\rho)\sqsupseteq \sem{\qqwhile{\bar{q}}{Q(\bm{\theta})}}(\rho).\]
        Back to the Formula~(\ref{eq:a1}), let $m\to \infty$, we have
        \[ \sum_{\pi\in\Pi_{\cS_{P(\bm{\theta})}}} \cE_{\pi}(\rho)\sqsubseteq \sem{\qqwhile{\bar{q}}{Q(\bm{\theta})}}(\rho).\]
        Therefore, 
        \[\sum_{\pi\in\Pi_{\cS_{P(\bm{\theta})}}} \cE_{\pi}(\rho)= \sem{\qqwhile{\bar{q}}{Q(\bm{\theta})}}(\rho).\]
    \end{itemize}
\end{proof}

\begin{figure*}[ht]
    \centering
    \begin{subfigure}[b]{0.4\linewidth}
        \centering
        \begin{tikzpicture}[thick]
    \filldraw[white] (-2.5,-2) rectangle (2.5,2);
    \node[location, label={$l_0$}] (l0) at (-1, 0) {};
    \node[location, label={$l_1$}] (l1) at (1, 0) {};
    \path[draw, ->] (l0) -- node[transition] {$e^{-i \theta \sigma}$} (l1);
    \node at (-1.5, -0.3) (li) {};
    \path[draw, ->] (li) edge[bend left] node[pos=-0.4] {$\rho_0$} (l0);
\end{tikzpicture}
        \caption{mSVTS for $\qutp{\theta}{\sigma}{\bar{q}}$}
        \label{fig:msvtsa}
    \end{subfigure}
    \begin{subfigure}[b]{0.54\linewidth}
    \centering
        \begin{tikzpicture}[thick]
    \filldraw[white] (-3.8,-1.8) rectangle (2.7,2.2);
    \node[location, label=93:$l_0$] (l0) at (-2, 0) {};
    \node at (-2.7, -0.2) (li) {};
    \path[draw, ->] (li) edge[bend left] node[pos=-0.4] {$\rho_0\otimes\sigma_{1,2,c}$} (l0);
    \node[location, label=-85:$l_3$] (l3) at (1.5, 0) {};
    \node[location, label=$l_2$] (l2) at (-0.2, 0) {};
    \draw[->] (l0) -- node[transition] {$[M_2]_{q_1,q_2}$} (l2);
    \draw[->] (l2) -- node[transition] {$I$} (l3);
    \node[location, label=left:$l_4$] (l4) at (-1.9, 1.2) {};
    \draw[->] (l0) -- node[transition] {$[M_1]_{q_1,q_2}$} (l4);
    \node[location, label=$l_5$] (l5) at (-1.0, 1.7) {};
    \draw[->] (l4) -- node[transition] {$[X]_{q_1}$} (l5);
    \node[location, label=$l_6$] (l6) at (0.3, 1.7) {};
    
    \node[location, label=$l_8$] (l8) at (-1, -1.4) {};
    \draw[->] (l0) -- node[transition] {$[M_0]_{q_1,q_2}$} (l8);
    \node[location, label=$l_9$] (l9) at (0.4, -1.4) {};
    \draw[->] (l8) -- node[transition] {$[C]_{q_c}$} (l9);
    
    \draw[->] (l9) -- node[transition] {$[GP]_{q_c,q_2}$} (l3);
    \node[location, label=60:$l_7$] (l7) at (1.2,1.2) {};
    \draw[->] (l5) -- node[transition] {$\cE_{\bar{q}', \sigma}$} (l6);
    \draw[->] (l7) -- node[transition] {$[AS]_{q_2,\bar{q}, \bar{q}'}$} (l3);
    \draw[->] (l6) -- node[transition] {$\cE_{q_2, \frac{I}{2}}$} (l7);
    \node[location, label=$l_1$] (l1) at (3, 0) {};
    \draw[->] (l3) -- node[transition] {$[e^{-i\theta\sigma}]_{\bar{q}}$} (l1);
\end{tikzpicture}
        \caption{mSVTS for $T_{\theta}(\qutp{\theta}{\sigma}{\bar{q}})$}
        \label{fig:msvtsb}
    \end{subfigure}
    \caption{Example of mSVTS for parameterized quantum $\qwhile\!$-program.}
    \label{fig:msvts}
    \Description{mSVTS}
\end{figure*}
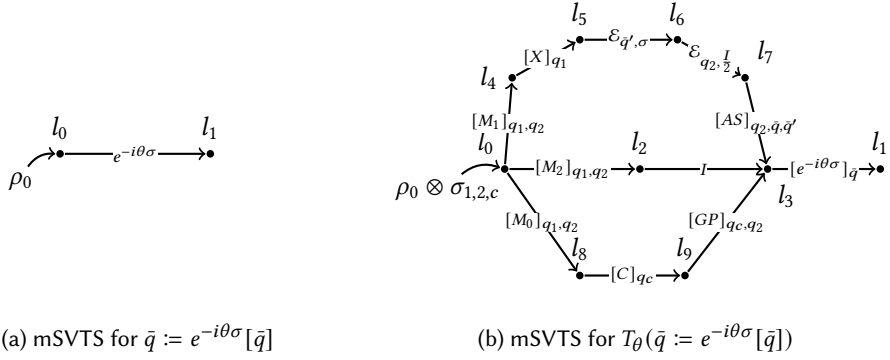

\figref{fig:msvts} shows the mSVTS for $\qutp{\theta}{\sigma}{\bar{q}}$ and the mSVTS for $T_{\theta}(\qutp{\theta}{\sigma}{\bar{q}})$. With this, we can construct an mSVTS for $\qd{\theta}(P(\bm{\theta}))$ by only modifying the mSVTS of $P(\bm{\theta})$.

\begin{lemma}\label{lem:a2}
    Let $\cS_{P(\bm{\theta})}$ be the control-flow graph of a parameterized quantum $\qwhile\!$-program $P(\bm{\theta})$. We can modify $\cS_{P(\bm{\theta})}$ to be the control-flow graph of $Q( \bm{\theta}) \equiv \qd{\theta}(P(\bm{\theta})) = \dinit;T_{\theta}(P(\bm{\theta}))$ as follows:
    \begin{enumerate}[label=(\arabic*).]
        \item For $\dinit$: We add $4$ locations $l_{in}^{Q}, l_{1}^Q, l_{2}^Q, l_{3}^Q, l_{4}^Q$ and five transitions $l_{in}^{Q}\overto{\cE_{q_1}} l_{1}^Q$, $l_{1}^{Q}\overto{\cE_{q_2}} l_{2}^Q$, $l_{2}^{Q}\overto{\cE_{q_c}} l_{3}^Q$, $l_{3}^Q \overto{[C]_{q_c}} l_{4}^Q$, $l_{4}^{Q}\overto{[GP]_{q_c,q_2}} l_{in}^{P(\bm{\theta})}$.
        \item For $T_{\theta}(P(\bm{\theta}))$: As shown in Figure~\ref{fig:msvts}, for every transition $a \overto{[e^{-i\theta\sigma}]_{\bar{q}}} b$ with a parameter symbol $\theta$ in $\cS_{P(\bm{\theta})}$, we add $8$ locations: $l_{a,b,j}, 2\leq j\leq 9$ and replace the transition $a \overto{e^{-i\theta\sigma}} b$ by $11$ transitions: $a \overto{[M_{1}]_{q_1,q_2}} l_{a,b,4}$, $l_{a,b,4} \overto{[X]_{q_1}} l_{a,b,5}$, $l_{a,b,5}\overto{\cE_{\bar{q}', \sigma}} l_{a,b,6}$, $l_{a,b,6}\overto{\cE_{q_2, \frac{I}{2}}} l_{a,b,7}$, $l_{a,b,7}\overto{[AS]_{q_2,\bar{q}, \bar{q}'}} l_{a,b,3}$, $a \overto{[M_2]_{q_1,q_2}} l_{a,b,2}$, $l_{a,b,2}\overto{I} l_{a,b,3}$, $a\overto{[M_0]_{q_1,q_2}} l_{a,b,8}$, $l_{a,b,8}\overto{[C]_{q_c}} l_{a,b,9}$, $l_{a,b,9} \overto{[GP]_{q_c,q_2}} l_{a,b,3}$, $l_{a,b,3} \overto{[e^{-i\theta\sigma}]_{\bar{q}}} b$.
    \end{enumerate}
    Then we get an mSVTS $\cS_{Q(\bm{\theta})}$, we have that $\cS_{Q(\bm{\theta})}$ represents the control-flow graph of $Q(\bm{\theta})$.
\end{lemma}
\begin{proof} 
    By definition, the control-flow graph of $Q(\bm{\theta}) = \dinit;$ $T_{\theta}(P(\bm{\theta}))$ can be constructed from $\cS_{\dinit}$ and $\cS_{T_{\theta}(P(\bm{\theta}))}$. In (1), it is obvious that (1) constructs the control-flow graph of the program $\dinit$. We next prove that in (2), it produces a control-flow graph for $T_{\theta}(P(\bm{\theta}))$. We prove this by induction through the program structure of $P(\bm{\theta})$.
    \begin{itemize}
        \item $P(\bm{\theta}) \equiv \qskip$, or $\qinit{q}$, or $\qut{U}{\bar{q}}$, or $\qutp{\theta'}{\sigma}{\bar{q}}$ (the symbol $\theta' \neq \theta$).
        By definition of $T_{\theta}$, $T_{\theta}(P(\bm{\theta})) = P(\bm{\theta})$ and we also see that $P(\bm{\theta})$ doesn't contain statement using $\theta$, then $\cS_{P(\bm{\theta})}$ has no transition that contains $\theta$.
        Thus, in (2), we don't change the $\cS_{P(\bm{\theta})}$. We have that it is still $\cS_{P(\bm{\theta})}$.
        \item $P(\bm{\theta}) \equiv \qutp{\theta}{\sigma}{\bar{q}}$. In (2), our construction comes from the \figref{fig:msvts}, we can check that the outcome represents the control-flow graph of $T_{\theta}(\qutp{\theta}{\sigma}{\bar{q}})$.
        \item $P(\bm{\theta}) \equiv P_1(\bm{\theta});P_2(\bm{\theta})$. Using
        \[T_{\theta}(P_1(\bm{\theta});P_2(\bm{\theta})) = T_{\theta}(P_1(\bm{\theta}));T_{\theta}(P_2(\bm{\theta}))\]
        and the inductive hypothesis on $P_1(\bm{\theta})$ and $P_2(\bm{\theta})$, the replacements in (2) are carried internally in $P_1(\bm{\theta})$ and $P_2(\bm{\theta})$, then concatenate them.
        This procedure is what we do in the definition of the control-flow graph with mSVTS.
        Thus, the outcome represents the control-flow graph of $T_{\theta}(P(\bm{\theta}))$.
        \item $P(\bm{\theta}) \equiv \qif (\Box m\cdot M[\bar{q}] = m \to P_{m}(\bm{\theta})) \qfi$.
        We have $T_{\theta}(\qif (\Box m\cdot M[\bar{q}] = m \to P_{m}(\bm{\theta})) \qfi) = \qif (\Box m\cdot M[\bar{q}] = m \to T_{\theta}(P_{m}(\bm{\theta}))) \qfi$ and the inductive hypothesis on $P_m(\bm{\theta})$.
        Then the rest is as same as above.
        \item $P(\bm{\theta}) \equiv \qqwhile{\bar{q}}{P'(\bm{\theta})}$.
        We have that \newline  $T_{\theta}(\qqwhile{\bar{q}}{P'(\bm{\theta})})$ $ =$ $\qwhile$ $M[\bar{q}]= 1 \qdo$ $T_{\theta}(P'(\bm{\theta}))$ and the inductive hypothesis on $P'(\bm{\theta})$.
        Then the rest is as same as above.
    \end{itemize}
    Therefore we get $\cS_{\dinit}$ and $\cS_{T_{\theta}(P(\bm{\theta}))}$.
    With $\cS_{\dinit}$ has the exit location $l_{in}^{P(\bm{\theta})}$ and $\cS_{T_{\theta}(P(\bm{\theta}))}$ has the same location $l_{in}^{P(\bm{\theta})}$ as the initial location, we have that $\cS_{Q(\bm{\theta})}$ represents the control-flow graph of $\qd{\theta}(P(\bm{\theta}))$.
\end{proof}

As the same settings in Lemma~\ref{lem:a2}, let $\eta_{\dinit} = l_{in}^{Q}\overto{\cE_{q_1}} l_{1}^{Q}\overto{\cE_{q_2}} l_{2}^{Q}\overto{\cE_{q_c}} l_{3}^{Q}\overto{[C]_{q_c}} l_{4}^Q \overto{[GP]_{q_c, q_2}} l_{in}^{P(\bm{\theta})}$ and for every transition $a \overto{[e^{-i\theta\sigma}]_{\bar{q}}} b$ with parameter symbol $\theta$ in $\cS_{P(\bm{\theta})}$, we write
\begin{align*}
    \eta_{a,b,0} &= a\overto{[M_0]_{q_1,q_2}} l_{a,b,8}\overto{[C]_{q_c}} l_{a,b,9} \overto{[GP]_{q_c,q_2}} l_{a,b,3} \overto{[e^{-i\theta\sigma}]_{\bar{q}}} b \\
    \eta_{a,b,1} &= a \overto{[M_{1}]_{q_1,q_2}} l_{a,b,4} \overto{[X]_{q_1}} l_{a,b,5} \overto{\cE_{\bar{q}', \sigma}} l_{a,b,6} \overto{\cE_{q_2,\frac{I}{2}}} l_{a,b,7} \overto{[AS]_{q_2,\bar{q}, \bar{q}'}} l_{a,b,3} \overto{[e^{-i\theta\sigma}]_{\bar{q}}} b \\
    \eta_{a,b,2} &= a \overto{[M_2]_{q_1,q_2}} l_{a,b,2}\overto{I} l_{a,b,3} \overto{[e^{-i\theta\sigma}]_{\bar{q}}} b
\end{align*}
For each $\pi\in\Pi_{\cS_{P(\bm{\theta})}}$ with $\pi = l_{in}^{P(\bm{\theta})} \overto{\cE_1} m_1 \overto{\cE_2} m_2 \overto{\cE_3} \cdots \overto{\cE_n} m_n$ and $1\leq i_1 < i_2< \ldots<i_k \leq n$ such that $\cE_{i_j} = [e^{-i\theta\sigma_{i_j}}]_{\bar{q_{i_j}}}$ for $j = 1, 2,\ldots, k$ (which is that the path $\pi$ has $k$ times occurrences of  the parameter symbol $\theta$), we write $\pi_j = m_{j-1} \overto{\cE_j} m_j$ for $1\leq j\leq n$ ($m_0 = l_{in}^{P(\bm{\theta})}$) and then define
\begin{equation}
    \begin{aligned}
    & A_{\pi} \\
    \equiv{}& \Bigl\lbrace\Bigr. \eta_{dinit}\pi_1\cdots\pi_{i_1-1}\mu_{i_1}\pi_{i_1+1}\cdots\pi_{i_k-1}\mu_{i_k}\pi_{i_k+1}\cdots\pi_n \mid  \mu_{i_j} \in \lbrace \eta_{m_{i_j-1}, m_{i_j}, 0}, \eta_{m_{i_j-1}, m_{i_j}, 1}, \eta_{m_{i_j-1}, m_{i_j}, 2}\rbrace \Bigl.\Bigr\rbrace
\end{aligned}
\end{equation}
In there, the set $A_{\pi}$ is obtained by replacing each $\pi_{i_j}$, which contains parameter symbol $\theta$, with one of $\eta_{m_{i_j-1}, m_{i_j}, 0}$, $\eta_{m_{i_j-1}, m_{i_j}, 1}$, $\eta_{m_{i_j-1}, m_{i_j}, 2}$ and adding $\eta_{dinit}$ to the front of $\pi$. We have the following lemma.

\begin{lemma}\label{lem:a6}
    As the same settings in Lemma~\ref{lem:a2}, and $A_{\pi}, \pi\in\Pi_{\cS_{P(\bm{\theta})}}$ defined above, we have 
    \begin{equation}\label{eq:3}
        \Pi_{\cS_{Q(\bm{\theta})}} = \bigcup_{\pi\in \Pi_{\cS_{P(\bm{\theta})}}}A_{\pi}.
    \end{equation}
\end{lemma}
\begin{proof}
    By constructions of $\cS_{Q(\bm{\theta})}$ and $A_{\pi}$, Formula~(\ref{eq:3}) is easy to see.
\end{proof}

\begin{lemma}\label{lem:a5}
    As the same settings in Lemma~\ref{lem:a2}, for any observable $O$ on $\cH_{P(\bm{\theta})}$, the same $O_d$ in Theorem~\ref{thm:exact} and $\pi\in\Pi_{\cS_{P(\bm{\theta})}}$ with $\pi = l_{in}^{P(\bm{\theta})} \overto{\cE_1} m_1 \overto{\cE_2} m_2 \overto{\cE_3} \cdots \overto{\cE_n} m_n$ and $1\leq i_1 < i_2< \ldots<i_k \leq n$ such that $\cE_{i_j} = [e^{-i\theta\sigma_{i_j}}]_{\bar{q_{i_j}}}$ for $j = 1, 2,\ldots, k$, we have
    \begin{equation}\label{eq:4}
        \begin{aligned}
            \tr\left(O_d\otimes O \sum_{\eta\in A_{\pi}}\cE_{\eta}(\rho)\right)
            =& \sum_{j=1}^k\Bigl( \tr\bigl(O\cE_{\pi_{i_j}\pi_{i_j+1}\cdots\pi_n}(-i[[\sigma_{i_j}]_{\bar{q}_{i_j}}, \cE_{\pi_1\cdots\pi_{i_j-1}}(\rho)])\bigr)\Bigr).
        \end{aligned}
    \end{equation}
    where $[\sigma_{i_j}]_{\bar{q}_{i_j}}$ denotes the operator $\sigma_{i_j}$ that acts on the Hilbert space $\cH_{\bar{q}_{i_j}}$.
\end{lemma}
\begin{proof}
    We only need to consider the state in $\cH_{q_c}\otimes\cH_{q_1}\otimes\cH_{q_2}\otimes\cH_{P(\bm{\theta})}$.
    \begin{align*}
        &\sum_{\eta\in A_{\pi}}\cE_{\eta}(\rho) \\
        =& \sum_{\genfrac{}{}{0pt}{4}{j=1,2,\ldots,k}{\mu_{i_j} \in \lbrace \eta_{m_{i_j-1}, m_{i_j}, 0}, \eta_{m_{i_j-1}, m_{i_j}, 1}, \eta_{m_{i_j-1}, m_{i_j}, 2}\rbrace}} \Bigl( \cE_{\eta_{dinit}\pi_1\cdots\pi_{i_1-1}\mu_{i_1}\pi_{i_1+1}\cdots\pi_{i_k-1}\mu_{i_k}\pi_{i_k+1}\cdots\pi_n}(\rho)\Bigr) \\
        =& \sum_{\genfrac{}{}{0pt}{4}{j=1,2,\ldots,k}{\mu_{i_j} \in \lbrace \eta_{m_{i_j-1}, m_{i_j}, 0}, \eta_{m_{i_j-1}, m_{i_j}, 1}, \eta_{m_{i_j-1}, m_{i_j}, 2}\rbrace}} \Bigl( \cE_{\pi_{i_1+1}\cdots\pi_{i_k-1}\mu_{i_k}\pi_{i_k+1}\cdots\pi_n}(\cE_{\eta_{dinit}\pi_1\cdots\pi_{i_1-1}\mu_{i_1}}(\rho))\Bigr) \\
        =& \sum_{\genfrac{}{}{0pt}{4}{j=2,\ldots,k}{\mu_{i_j} \in \lbrace \eta_{m_{i_j-1}, m_{i_j}, 0}, \eta_{m_{i_j-1}, m_{i_j}, 1}, \eta_{m_{i_j-1}, m_{i_j}, 2}\rbrace}} \Bigl( \cE_{\pi_{i_1+1}\cdots\pi_{i_k-1}\mu_{i_k}\pi_{i_k+1}\cdots\pi_n}\bigl( \sum_{h_1 = 0}^2\cE_{\eta_{m_{i_1-1}, m_{i_1}, h_1}}(\cE_{\eta_{dinit}\pi_1\cdots\pi_{i_1-1}}(\rho))\bigr)\Bigr).
    \end{align*}
    Let $\ket{\psi_j} = R_y(2\arcsin(\sqrt{b_j}))\ket{0}$, $b_j = \mu(j)/(1-\sum_{k=1}^{j-1}\mu(k))$, then with Lemma~\ref{lem:a3} and definition of $\cE_{\eta_{dinit}\pi_1\cdots\pi_{i_1-1}}$, we have
    \begin{align*}
        &\sum_{h_1 = 0}^2\cE_{\eta_{m_{i_1-1}, m_{i_1}, h_1}}(\cE_{\eta_{dinit}\pi_1\cdots\pi_{i_1-1}}(\rho)) \\
        =&\sum_{h_1 = 0}^2\cE_{\eta_{m_{i_1-1}, m_{i_1}, h_1}}\bigl( \ketbra[c]{1}{1}\otimes\ketbra[1]{0}{0}\otimes\ketbra[2]{\psi_1}{\psi_1}\otimes\cE_{\pi_1\cdots\pi_{i_1-1}}(\rho)\bigr) \\
        =& \left(\sum_{l=2}^{\infty}\mu(l)\right)\ketbra[c]{2}{2}\otimes\ketbra[1]{0}{0}\otimes\ketbra[2]{\psi_2}{\psi_2}\otimes \cE_{\pi_1\cdots\pi_{i_1}}(\rho)  \\
        & \  + \frac{\mu(1)}{2} \ketbra[c]{1}{1}\otimes\ketbra[1]{1}{1}\otimes\biggl( \\
        & \quad \ketbra[2]{0}{0}\otimes\Bigl(\cos(\frac{\pi}{4})^2 \cE_{\pi_1\cdots\pi_{i_1}}(\rho) + \sin(\frac{\pi}{4})^2 \cE_{\pi_{i_1}}(\tr_{\bar{q}_{i_1}}(\cE_{\pi_1\cdots\pi_{i_1-1}}(\rho))\otimes\sigma_{i_1}) \\
        &\quad + \frac{1}{2}\sin(\frac{\pi}{2})\cE_{\pi_{i_1}}(-i[[\sigma_{i_1}]_{\bar{q}_{i_1}}, \cE_{\pi_1\cdots\pi_{i_1-1}}(\rho)])\Bigr)\\
        & \quad \ketbra[2]{1}{1}\otimes\Bigl(\cos(-\frac{\pi}{4})^2 \cE_{\pi_1\cdots\pi_{i_1}}(\rho) + \sin(-\frac{\pi}{4})^2 \cE_{\pi_{i_1}}(\tr_{\bar{q}_{i_1}}(\cE_{\pi_1\cdots\pi_{i_1-1}}(\rho))\otimes\sigma_{i_1}) \\
        &\quad + \frac{1}{2}\sin(-\frac{\pi}{2})\cE_{\pi_{i_1}}(-i[[\sigma_{i_1}]_{\bar{q}_{i_1}}, \cE_{\pi_1\cdots\pi_{i_1-1}}(\rho)])\Bigr)\biggr).
    \end{align*}
    We find that in the second term $q_1$ is in $\ketbra{1}{1}$, then in the later execution, it will never go into $M_1$ or $M_2$, thus
    \begin{align*}
        &\sum_{h_2=0}^2 \sum_{h_1 = 0}^2\cE_{\eta_{m_{i_2-1}, m_{i_2}, h_2}}\bigl( \cE_{\pi_{i_1+1}\cdots\pi_{i_2-1}}(\cE_{\eta_{m_{i_1-1}, m_{i_1}, h_1}}(\cE_{\eta_{dinit}\pi_1\cdots\pi_{i_1-1}}(\rho)))\bigr) \\
        =& \left(\sum_{l=3}^{\infty}\mu(l)\right)\ketbra[c]{3}{3}\otimes\ketbra[1]{0}{0}\otimes\ketbra[2]{\psi_3}{\psi_3}\otimes\cE_{\pi_1\cdots\pi_{i_2}}(\rho)  \\
        & \  + \frac{\mu(2)}{2} \ketbra[c]{2}{2}\otimes\ketbra[1]{1}{1}\otimes\biggl( \\
        & \quad \ketbra[2]{0}{0}\otimes\Bigl(\cos(\frac{\pi}{4})^2 \cE_{\pi_1\cdots\pi_{i_2}}(\rho) + \sin(\frac{\pi}{4})^2 \cE_{\pi_{i_2}}(\tr_{\bar{q}_{i_2}}(\cE_{\pi_1\cdots\pi_{i_2-1}}(\rho))\otimes\sigma_{i_2}) \\
        &\quad + \frac{1}{2}\sin(\frac{\pi}{2})\cE_{\pi_{i_2}}(-i[[\sigma_{i_2}]_{\bar{q}_{i_2}}, \cE_{\pi_1\cdots\pi_{i_2-1}}(\rho)])\Bigr)\\
        & \quad \ketbra[2]{1}{1}\otimes\Bigl(\cos(-\frac{\pi}{4})^2 \cE_{\pi_1\cdots\pi_{i_2}}(\rho) + \sin(-\frac{\pi}{4})^2 \cE_{\pi_{i_2}}(\tr_{\bar{q}_{i_2}}(\cE_{\pi_1\cdots\pi_{i_2-1}}(\rho))\otimes\sigma_{i_2}) \\
        &\quad + \frac{1}{2}\sin(-\frac{\pi}{2})\cE_{\pi_{i_2}}(-i[[\sigma_{i_2}]_{\bar{q}_{i_2}}, \cE_{\pi_1\cdots\pi_{i_2-1}}(\rho)])\Bigr)\biggr) \\
        & \  + \frac{\mu(1)}{2} \ketbra[c]{1}{1}\otimes\ketbra[1]{1}{1}\otimes\biggl( \\
        & \quad \ketbra[2]{0}{0}\otimes\Bigl(\cos(\frac{\pi}{4})^2 \cE_{\pi_1\cdots\pi_{i_1}}(\rho) + \sin(\frac{\pi}{4})^2 \cE_{\pi_{i_1}}(\tr_{\bar{q}_{i_1}}(\cE_{\pi_1\cdots\pi_{i_1-1}}(\rho))\otimes\sigma_{i_1}) \\
        &\quad + \frac{1}{2}\sin(\frac{\pi}{2})\cE_{\pi_{i_1}}(-i[[\sigma_{i_1}]_{\bar{q}_{i_1}}, \cE_{\pi_1\cdots\pi_{i_1-1}}(\rho)])\Bigr)\\
        & \quad \ketbra[2]{1}{1}\otimes\Bigl(\cos(-\frac{\pi}{4})^2 \cE_{\pi_1\cdots\pi_{i_1}}(\rho) + \sin(-\frac{\pi}{4})^2 \cE_{\pi_{i_1}}(\tr_{\bar{q}_{i_1}}(\cE_{\pi_1\cdots\pi_{i_1-1}}(\rho))\otimes\sigma_{i_1}) \\
        &\quad + \frac{1}{2}\sin(-\frac{\pi}{2})\cE_{\pi_{i_1}}(-i[[\sigma_{i_1}]_{\bar{q}_{i_1}}, \cE_{\pi_1\cdots\pi_{i_1-1}}(\rho)])\Bigr)\biggr)
    \end{align*}
    step by step,
    \begin{align*}
        \sum_{\eta\in A_{\pi}}\cE_{\eta}(\rho) 
        =& \left(\sum_{l=k+1}^{\infty}\mu(l)\right)\ketbra[c]{k+1}{k+1}\otimes\ketbra[1]{0}{0}\otimes\ketbra[2]{\psi_{k+1}}{\psi_{k+1}}\otimes\cE_{\pi}(\rho)  \\
        &\  + \sum_{j=1}^k\biggl(\frac{\mu(j)}{2} \ketbra[c]{j}{j}\otimes\ketbra[1]{1}{1}\otimes\biggl( \\
        & \quad \ketbra[2]{0}{0}\otimes\Bigl(\frac{1}{2}\cE_{\pi}(\rho) + \frac{1}{2} \cE_{\pi_{i_j}\pi_{i_j+1}\cdots\pi_n}(\tr_{\bar{q}_{i_j}}(\cE_{\pi_1\cdots\pi_{i_j-1}}(\rho))\otimes\sigma_{i_j}) \\
        &\quad + \frac{1}{2}\cE_{\pi_{i_j}\pi_{i_j+1}\cdots\pi_n}(-i[[\sigma_{i_j}]_{\bar{q}_{i_j}}, \cE_{\pi_1\cdots\pi_{i_j-1}}(\rho)])\Bigr)\\
        & \quad \ketbra[2]{1}{1}\otimes\Bigl(\frac{1}{2} \cE_{\pi}(\rho) + \frac{1}{2} \cE_{\pi_{i_j}\pi_{i_j+1}\cdots\pi_n}(\tr_{\bar{q}_{i_j}}(\cE_{\pi_1\cdots\pi_{i_j-1}}(\rho))\otimes\sigma_{i_j}) \\
        &\quad - \frac{1}{2}\cE_{\pi_{i_j}\pi_{i_j+1}\cdots\pi_n}(-i[[\sigma_{i_j}]_{\bar{q}_{i_j}}, \cE_{\pi_1\cdots\pi_{i_j-1}}(\rho)])\Bigr)\biggr)\biggr).
    \end{align*}
    With $O_d = \sum_{j=1}^{\infty}\frac{2}{\mu(j)}\ketbra{j}{j}\otimes\ketbra{1}{1}\otimes Z$, an observable on $\cH_{q_c}\otimes\cH_{q_1}\otimes \cH_{q_2}$, we have
    \begin{align*}
        \tr\left(O_d\otimes O\sum_{\eta\in A_{\pi}}\cE_{\eta}(\rho)\right)
        =& \sum_{j=1}^k\Bigl( \tr\bigl(O\cE_{\pi_{i_j}\pi_{i_j+1}\cdots\pi_n}(-i[[\sigma_{i_j}]_{\bar{q}_{i_j}}, \cE_{\pi_1\cdots\pi_{i_j-1}}(\rho)])\bigr)\Bigr).
    \end{align*}
    
\end{proof}

We already know that
\[
    \sem{\qd{\theta}(P(\bm{\theta}))}(\rho) = \sum_{\pi\in\Pi_{\cS_{Q(\bm{\theta})}}} \cE_{\pi}(\rho)
    = \sum_{\pi\in \Pi_{\cS_{P(\bm{\theta})}}}\sum_{\eta\in A_{\pi}} \cE_{\eta}(\rho)
\]
then, 
\begin{equation}\label{eq:18}
    \begin{aligned}
        & \tr\left(O_d\otimes O \sem{\qd{\theta}(P(\bm{\theta}))}(\rho)\right) \\
        ={}& \sum_{\pi\in \Pi_{\cS_{P(\bm{\theta})}}}\tr\left(O_d\otimes O\sum_{\eta\in A_{\pi}} \cE_{\eta}(\rho)\right).
    \end{aligned}
\end{equation}
We should carefully consider the convergence of the above summation.

\begin{lemma}
    \label{lem:a10}
    As the same settings in Theorem~\ref{thm:exact} and\\  Lemma~\ref{lem:a2}, we fix $\bm{\theta} = \bm{\theta}^*$, for $x\in \mathbb{R}$, $n\in \mathbb{N}$, let
    \[h_n(x) = \sum_{\pi\in\Pi_{\cS_{P(\bm{\theta}^*[\theta\mapsto x])}}^{(n)}}\tr\left(O_d\otimes O\sum_{\eta\in A_{\pi}}\cE_{\eta}(\rho)\right)\]
    then, $\lim_{n\to \infty} h_n(x)$ exists, which is exactly
    \[ \tr\left(O_d\otimes O \sem{\qd{\theta}(P(\bm{\theta}^*[\theta\mapsto x]))}(\rho)\right)\]
    and $h_n(x)$ is uniform convergent on any close interval.
\end{lemma}
\begin{proof}
    Let $M$ be the largest eigenvalue of $\abs{O}$. For any $ \pi\in\Pi_{\cS_{P(\bm{\theta})}}$ with $\pi = l_{in}^{P(\bm{\theta})} \overto{\cE_1} m_1 \overto{\cE_2} m_2 \overto{\cE_3} \cdots \overto{\cE_n} m_n$, if symbol $\theta$ does not appear in $\pi$, then
    \[ \abs{\tr\left(O_d\otimes O \sum_{\eta\in A_{\pi}}\cE_{\eta}(\rho)\right)} = 0.\]
    otherwise, there exist $1\leq i_1 < i_2< \ldots<i_k \leq n$ such that $\cE_{i_j} = [e^{-i\theta\sigma_{i_j}}]_{\bar{q_{i_j}}}$ for $j = 1, 2,\ldots, k$, then by Lemma~\ref{lem:a5}, we have
    \begin{align*}
        \tr\left(O_d\otimes O \sum_{\eta\in A_{\pi}}\cE_{\eta}(\rho)\right) 
        =& \sum_{j=1}^k\Bigl( \tr\bigl(O\cE_{\pi_{i_j}\pi_{i_j+1}\cdots\pi_n}(-i[[\sigma_{i_j}]_{\bar{q}_{i_j}}, \cE_{\pi_1\cdots\pi_{i_j-1}}(\rho)])\bigr)\Bigr)
    \end{align*}
    then,
    \begin{align*}
        \abs{\tr\left(O_d\otimes O \sum_{\eta\in A_{\pi}}\cE_{\eta}(\rho)\right)} 
        \leq{} & \sum_{j=1}^k \abs{\tr\bigl(O\cE_{\pi_{i_j}\pi_{i_j+1}\cdots\pi_n}(-i[[\sigma_{i_j}]_{\bar{q}_{i_j}}, \cE_{\pi_1\cdots\pi_{i_j-1}}(\rho)])\bigr)} \\
        \leq{} & \sum_{j=1}^kM\abs{\tr\bigl(\cE_{\pi_{i_j}\pi_{i_j+1}\cdots\pi_n}(-i[[\sigma_{i_j}]_{\bar{q}_{i_j}}, \cE_{\pi_1\cdots\pi_{i_j-1}}(\rho)])\bigr)} \\
        \leq{} & \sum_{j=1}^k 2M \tr\bigl(\cE_{\pi_{i_j}\pi_{i_j+1}\cdots\pi_n}(\cE_{\pi_1\cdots\pi_{i_j-1}}(\rho))\bigr) \\
        ={} & \sum_{j=1}^k 2M\tr(\cE_{\pi}(\rho)) = 2kM\tr(\cE_{\pi}(\rho)).
    \end{align*}
     Thus for any $\pi\in \Pi_{\cS_{P(\bm{\theta})}}$ with length $n$,
    \begin{equation}\label{eq:17}
        \begin{aligned}
            \abs{\tr\left(O_d\otimes O \sum_{\eta\in A_{\pi}}\cE_{\eta}(\rho)\right)}
            \leq{}& 2nM\tr(\cE_{\pi}(\rho)).
        \end{aligned} 
    \end{equation}

    Let $M_1$ be the size of transition set of $\cS_{P(\bm{\theta})}$, $M_2$ be the number of occurrences of \textbf{while} statements in $P(\bm{\theta})$, which means that $P(\bm{\theta})$ contains following subprograms
    \[ P_{j}(\bm{\theta}) \equiv \qwhile M^{(j)}[\bar{q}_j] = 1 \qdo Q_{j}(\bm{\theta}) \qod, \enspace j = 1,\ldots, M_2.\] 
    these $P_j(\bm{\theta})$ can contain each others. In there we assume $M_2 \geq 1$, otherwise $P(\bm{\theta})$ doesn't contain \textbf{while} statement, then the conclusion is trivial. By Lemma~\ref{lem:a9}, we fix $\bm{\theta}$ and choose $\epsilon = \frac{1}{2}$, then for these $M_2$ subprograms, there exist $N_1,\ldots,N_{M_2}$ such that
    \begin{align*}
        & \forall n \geq 0, \forall \rho \in \cH_{P(\bm{\theta})}, 1\leq j \leq M_2, \\
        & \tr(\cE_0^{(j)}\circ(\sem{Q_j(\bm{\theta})}\circ\cE_1^{(j)})^n(\rho)) \leq \left(\frac{1}{2}\right)^{\lfloor\frac{n}{N_j}\rfloor}\tr(\rho).
    \end{align*}
    Let $N = \max_{j=1,\ldots, M_2}N_j$, then
    \begin{equation}\label{eq:11}
        \begin{aligned}
            &\forall n \geq 0, \forall \rho \in \cH_{P(\bm{\theta})}, 1 \leq j \leq M_2, \\
        &\tr(\cE_0^{(j)}\circ(\sem{Q_j(\bm{\theta})}\circ\cE_1^{(j)})^n(\rho)) \leq \left(\frac{1}{2}\right)^{\lfloor\frac{n}{N}\rfloor}\tr(\rho).
        \end{aligned}
    \end{equation}

    For any $n \geq 2$, we consider $\pi \in \Pi_{\cS_{P(\bm{\theta})}}^{((n+1)^{M_2} M_1-1)} \setminus\Pi_{\cS_{P(\bm{\theta})}}^{(n^{M_2} M_1-1)}$, the length of $\pi$ is at least $n^{M_2} M_1$, then $\pi$ has at least $n^{M_2}M_1/ M_1 = n^{M_2}$ locations in $\cS_{P(\bm{\theta})}$ repeatedly appear, which is caused by \textbf{while} statements. Then  $\pi$ has at least $n^{M_2}$ times runs into the loop bodies of these $P_j(\bm{\theta}), j=1,\ldots,M_2$. 
    
    Let $a(\lbrace P_1(\bm{\theta}), \ldots, P_{M_2}(\bm{\theta})\rbrace)$ denote the possible maximum times of $\pi$ runs into the loop bodys of these $P_j(\bm{\theta}), j=1,\ldots,M_2$ on the condition that $\pi$ only continuously runs into the loop body of each $P_{j}(\bm{\theta})$ with no more than $n-1$ times. We can assume that $P_1(\bm{\theta})$ is the first \textbf{while} statement appear in $P(\bm{\theta})$ and it contains $0\leq t \leq M_2-1$ \textbf{while} statement $P_{j_1}(\bm{\theta}), \ldots, P_{j_t}(\bm{\theta})$, $2\leq j_1 \leq \cdots \leq j_t \leq M_2$, then 
    \begin{align*}
        & a(\lbrace P_1(\bm{\theta}), \ldots, P_{M_2}(\bm{\theta})\rbrace) \\
        \leq{}& \underbrace{(n-1)}_{\text{for }P_1(\bm{\theta})} + \underbrace{(n-1)a(\lbrace P_{j_1}(\bm{\theta}), \ldots, P_{j_t}(\bm{\theta})\rbrace)}_{\text{for those \textbf{while} statements in the loop body of } P_1(\bm{\theta})}  \\
        & \quad + \underbrace{a(\lbrace P_2(\bm{\theta}), \ldots, P_{M_2}(\bm{\theta})\rbrace \setminus \lbrace P_{j_1}(\bm{\theta}), \ldots, P_{j_t}(\bm{\theta})\rbrace)}_{\text{for those \textbf{while} statements not in the loop body of } P_1(\bm{\theta})} \\
        \leq{} & (n-1) + n a(\lbrace P_{2}(\bm{\theta}), \ldots, P_{M_2}(\bm{\theta})\rbrace) \\
        & \cdots \\
        \leq{} & \sum_{j=0}^{M_2-1}(n-1)n^j = n^{M_2} - 1 < n^{M_2}.
    \end{align*}
    This conflicts to $\pi$ has at least $n^{M_2}$ times runs into the loop body of these $P_j(\bm{\theta}), j=1,\ldots,M_2$. Therefore, for any $\pi\in \Pi_{\cS_{P(\bm{\theta})}}^{((n+1)^{M_2} M_1-1)} \setminus\Pi_{\cS_{P(\bm{\theta})}}^{(n^{M_2} M_1-1)}$, there exists $1\leq j_0 \leq M_2$ such that $\pi$ continuously runs into the loop body of $P_{j_0}(\bm{\theta})$ with more than $n-1$ times, which is $\pi$ can be written as
    \[ \pi = \pi_1\eta_1\mu_1\eta_1\mu_2\cdots\eta_1\mu_{t}\eta_0\pi_2\]
    with $\eta_1 = l_{in}^{P_{j_0}(\bm{\theta})}\overto{\cE_{1}^{(j_0)}} l_{in}^{Q_{j_0}(\bm{\theta})}$, $\eta_0 = l_{in}^{P_{j_0}(\bm{\theta})}\overto{\cE_{0}^{(j_0)}} l_{out}^{P_{j_0}(\bm{\theta})}$, $\mu_j \in \Pi_{\cS_{Q_{j_0}(\bm{\theta})}}, 1\leq j\leq t, t \geq n$. 
    
    We define the following set for each $1\leq j\leq M_2$, $n\geq 2$ and $0\leq m\leq (n+1)^{M_2}M_1-1$,
    \begin{align*}
        A_j^{(n, m)} \equiv{}&\biggl\lbrace \pi \in \Pi_{\cS_{P(\bm{\theta})}}^{((n+1)^{M_2} M_1-1)} \setminus\Pi_{\cS_{P(\bm{\theta})}}^{(n^{M_2} M_1-1)} : \pi = \pi_1\eta_1\mu_1\eta_1\mu_2\cdots\eta_1\mu_{n}\eta_0\pi_2, \\
        & \qquad \eta_1 = l_{in}^{P_{j}(\bm{\theta})}\overto{\cE_{1}^{(j)}} l_{in}^{Q_{j}(\bm{\theta})}, \eta_0 = l_{in}^{P_{j}(\bm{\theta})}\overto{\cE_{0}^{(j)}} l_{out}^{P_{j}(\bm{\theta})}, \mu_k \in \Pi_{\cS_{Q_{j}(\bm{\theta})}}, 1\leq k\leq n, \\
        & \qquad \pi_1 \text{ contains $m$ times $\eta_1$}, \pi_2 \text{ may contain } \eta_1 \biggr\rbrace
    \end{align*}
    then
    \[ \Pi_{\cS_{P(\bm{\theta})}}^{((n+1)^{M_2} M_1-1)} \setminus\Pi_{\cS_{P(\bm{\theta})}}^{(n^{M_2} M_1-1)} = \bigcup_{j=1}^{M_2}\left(\cup_{m=0}^{n^{M_2}M_1-1}A_{j}^{(n,m)}\right).\]
    For each $A_{j}^{(n,m)}$, we define 
    \begin{align*}
        B_{j}^{(n,m)} \equiv{}& \biggl\lbrace \pi_1\eta_1:
        \pi_1\eta_1\mu_1\eta_1\mu_2\cdots\eta_1\mu_{n}\eta_0\pi_2 \in A_j^{(n,m)}, \\
        & \qquad \eta_1 = l_{in}^{P_{j}(\bm{\theta})}\overto{\cE_{1}^{(j)}} l_{in}^{Q_{j}(\bm{\theta})}, \eta_0 = l_{in}^{P_{j}(\bm{\theta})}\overto{\cE_{0}^{(j)}} l_{out}^{P_{j}(\bm{\theta})}, \mu_k \in \Pi_{\cS_{Q_{j}(\bm{\theta})}}, 1\leq k\leq n, \\
        & \qquad \pi_1 \text{ contains $m$ times $\eta_1$}, \pi_2 \text{ may contain } \eta_1 \biggr\rbrace,\\
        C_{j}^{(n,m)} \equiv{}& \biggl\lbrace \mu_1, \mu_2, \ldots, \mu_n:
        \pi_1\eta_1\mu_1\eta_1\mu_2\cdots\eta_1\mu_{n}\eta_0\pi_2 \in A_j^{(n,m)}, \\
        & \qquad \eta_1 = l_{in}^{P_{j}(\bm{\theta})}\overto{\cE_{1}^{(j)}} l_{in}^{Q_{j}(\bm{\theta})}, \eta_0 = l_{in}^{P_{j}(\bm{\theta})}\overto{\cE_{0}^{(j)}} l_{out}^{P_{j}(\bm{\theta})}, \mu_k \in \Pi_{\cS_{Q_{j}(\bm{\theta})}}, 1\leq k\leq n, \\
        & \qquad \pi_1 \text{contains $m$ times $\eta_1$}, \pi_2 \text{ may contain } \eta_1 \biggr\rbrace, \\
        D_{j}^{(n,m)} \equiv{}& \biggl\lbrace \pi_2:
        \pi_1\eta_1\mu_1\eta_1\mu_2\cdots\eta_1\mu_{n}\eta_0\pi_2 \in A_j^{(n,m)}, \\
        & \qquad \eta_1 = l_{in}^{P_{j}(\bm{\theta})}\overto{\cE_{1}^{(j)}} l_{in}^{Q_{j}(\bm{\theta})}, \eta_0 = l_{in}^{P_{j}(\bm{\theta})}\overto{\cE_{0}^{(j)}} l_{out}^{P_{j}(\bm{\theta})}, \mu_k \in \Pi_{\cS_{Q_{j}(\bm{\theta})}}, 1\leq k\leq n, \\
        & \qquad \pi_1 \text{ contains $m$ times $\eta_1$}, \pi_2 \text{ may contain } \eta_1 \biggr\rbrace. \\
    \end{align*}
    $B_{i}^{(n,m)}, C_{i}^{(n,m)}, D_{i}^{(n,m)}$ are all finite, and we have, 
    \begin{align*}
        A_{j}^{(n,m)} \subseteq E_{i}^{(n,m)} \equiv{}& \Bigl\lbrace \pi_1\eta_1\mu_1\eta_1\mu_2\cdots\eta_1\mu_{n}\eta_0\pi_2 :  \pi_1\eta_1\in B_{i}^{(n,m)}, \pi_2\in D_{i}^{(n,m)}, \mu_k \in C_{i}^{(n,m)}, 1\leq k\leq n\Bigr\rbrace,
    \end{align*}
    \[ C_{j}^{(n,m)} \subseteq \Pi_{\cS_{Q_{j}(\bm{\theta})}}.\]
    By Theorem~\ref{thm:a1}, for any $\rho\in \cD(\cH_{P(\bm{\theta})}), 1\leq j\leq M_2,$
    \begin{equation}\label{eq:14}
        \sum_{\mu\in C_{j}^{(n,m)}}\cE_{\mu}(\rho) \sqsubseteq \sum_{\mu\in \Pi_{\cS_{Q_{j}(\bm{\theta})}}}\cE_{\mu}(\rho) = \sem{Q_{j}(\bm{\theta})}(\rho).
    \end{equation}
    As regards to $B_{j}^{(n,m)}$, any $\pi, \pi' \in B_{j}^{(n,m)}, \pi\neq \pi'$, we have $\pi = \pi_1\eta_1, \pi' = \pi_1'\eta_1$ and $\pi_1, \pi_1'$ contains $m$ times of $\eta_1$, then $\pi, \pi'$ are not prefix of each other (otherwise $\pi = \pi'$).  Then $B_{j}^{(n,m)}$ satisfies the conditions of Lemma~\ref{lem:a11}, we have that for any $\rho\in \cD(\cH_{P(\bm{\theta})})$,
    \begin{equation}
        \label{eq:15}
        \sum_{\pi\in B_{j}^{(n,m)}}\tr(\cE_{\pi}(\rho)) \leq \tr(\rho), \enspace 1\leq j\leq M_2.
    \end{equation}
    For any $\pi, \pi' \in D_{j}^{(n,m)}, \pi \neq \pi'$, $\pi$ and $\pi'$ have $l_{out}^{P(\bm{\theta)}}$ as last location, which has no post-location, then $\pi, \pi'$ are not prefix of each other. $D_{j}^{(n,m)}$ also satisfies the conditions of Lemma~\ref{lem:a11}, then for any $\rho \in \cD(\cH_{P(\bm{\theta})})$,
    \begin{equation}
        \label{eq:16}
        \sum_{\pi\in D_{j}^{(n,m)}}\tr(\cE_{\pi}(\rho)) \leq \tr(\rho), \enspace 1\leq j\leq M_2.
    \end{equation}

    With Formula~(\ref{eq:14}), for any $\rho \in \cD(\cH_{P(\bm{\theta})})$ and any $1\leq j\leq M_2$, we have
    \begin{align*}
        &\phantom{=} \sum_{\pi\in A_j^{(n,m)}} \cE_{\pi}(\rho) \sqsubseteq \sum_{\pi\in E_{i}^{(n,m)}} \cE_{\pi}(\rho) \\
        &= \sum_{\pi_1\eta_1\mu_1\eta_1\mu_2\cdots\eta_1\mu_{n}\eta_0\pi_2\in E_{i}^{(n,m)}} \Bigl( \cE_{\pi_1\eta_1\mu_1\eta_1\mu_2\cdots\eta_1\mu_{n}\eta_0\pi_2}(\rho)\Bigr) \\
        &= \sum_{\pi_1\eta_1\mu_1\eta_1\mu_2\cdots\eta_1\mu_{n}\eta_0\pi_2\in E_{i}^{(n,m)}}\Bigl( \cE_{\pi_2}\circ\cE_{\eta_0}\circ\cE_{\mu_n}\circ\cE_{\eta_1}\circ \cdots\circ\cE_{\mu_1}\circ\cE_{\eta_1}\circ\cE_{\pi_1}(\rho)\Bigr) \\
        &= \left(\sum_{\pi_2\in D_{j}^{(n,m)}}\cE_{\pi_2}\right)\circ\cE_{\eta_0}\circ\left(\sum_{\mu_n\in C_{j}^{(n,m)}}\cE_{\mu_n}\right)\circ\cE_{\eta_1}\circ \cdots \circ \left(\sum_{\mu_1\in C_{j}^{(n,m)}}\cE_{\mu_1}\right)\circ\left(\sum_{\pi_1\eta_1\in B_{j}^{(n,m)}}\cE_{\pi_1\eta_1}\right)(\rho) \\
        &\sqsubseteq \left(\sum_{\pi_2\in D_{j}^{(n,m)}}\cE_{\pi_2}\right)\circ\cE_{\eta_0} \circ\underbrace{\sem{Q_{j}(\bm{\theta})}\circ\cE_{\eta_1}\circ \cdots \circ \sem{Q_{j}(\bm{\theta})}}_{n \text{ times } \sem{Q_{j}(\bm{\theta})}} \circ\left(\sum_{\pi_1\eta_1\in B_{j}^{(n,m)}}\cE_{\pi_1\eta_1}\right)(\rho) \\
        &= \left(\sum_{\pi_2\in D_{j}^{(n,m)}}\cE_{\pi_2}\right)\circ\cE_{\eta_0} \circ{\left(\sem{Q_{j}(\bm{\theta})}\circ\cE_{\eta_1}\right)^{n-1}}\circ\sem{Q_j(\bm{\theta})} \circ\left(\sum_{\pi_1\eta_1\in B_{j}^{(n,m)}}\cE_{\pi_1\eta_1}\right)(\rho)
    \end{align*}
    thus,
    \begin{align*}
        & \sum_{\pi\in A_j^{(n,m)}} \tr\left(\cE_{\pi}(\rho)\right) \\
        \leq{}& \tr\biggggl(\left(\sum_{\pi_2\in D_{j}^{(n,m)}}\cE_{\pi_2}\right)\circ\cE_{\eta_0} \circ{\left(\sem{Q_{j}(\bm{\theta})}\circ\cE_{\eta_1}\right)^{n-1}}\circ\sem{Q_j(\bm{\theta})} \circ\left(\sum_{\pi_1\eta_1\in B_{j}^{(n,m)}}\cE_{\pi_1\eta_1}\right)(\rho)\biggggr) \\
        \leq{}& \tr\biggggl(\cE_{\eta_0}\circ{\left(\sem{Q_{j}(\bm{\theta})}\circ\cE_{\eta_1}\right)^{n-1}}\circ\sem{Q_j(\bm{\theta})} \circ\left(\sum_{\pi_1\eta_1\in B_{j}^{(n,m)}}\cE_{\pi_1\eta_1}\right)(\rho)\biggggr) \tag{by Formula~(\ref{eq:16})} \\
        \leq{}& \left(\frac{1}{2}\right)^{\lfloor \frac{n-1}{N}\rfloor}\tr\left(\sem{Q_j(\bm{\theta})}\circ\left(\sum_{\pi_1\eta_1\in B_{j}^{(n,m)}}\cE_{\pi_1\eta_1}\right)(\rho)\right) \tag{by Formula~(\ref{eq:11})} \\
        \leq{}& \left(\frac{1}{2}\right)^{\lfloor \frac{n-1}{N}\rfloor}\tr\left(\sum_{\pi_1\eta_1\in B_{j}^{(n,m)}}\cE_{\pi_1\eta_1}(\rho)\right) \\
        \leq{}& \left(\frac{1}{2}\right)^{\lfloor \frac{n-1}{N}\rfloor} \tr(\rho). \tag{by Formula~(\ref{eq:15})}
    \end{align*}
    Then,
    \begin{equation}
        \label{eq:21}
        \begin{aligned}
            \sum_{\pi \in \Pi_{\cS_{P(\bm{\theta})}}^{((n+1)^{M_2} M_1-1)} \setminus\Pi_{\cS_{P(\bm{\theta})}}^{(n^{M_2} M_1-1)}}\tr(\cE_{\pi}(\rho)) 
            \leq{}& \sum_{j=1}^{M_2}\sum_{m=0}^{(n+1)^{M_2}M_1-1}\sum_{\pi\in A_{j}^{(n,m)}} \tr(\cE_{\pi}(\rho)) \\
            \leq{}& M_1M_2(n+1)^{M_2}\left(\frac{1}{2}\right)^{\lfloor \frac{n-1}{N}\rfloor} \tr(\rho).
        \end{aligned}
    \end{equation}
    with Formula~(\ref{eq:17}), we have for any $\rho \in \cD(\cH_{P(\bm{\theta})})$,
    \begin{align*}
        & \sum_{\pi \in \Pi_{\cS_{P(\bm{\theta})}}^{((n+1)^{M_2} M_1-1)} \setminus\Pi_{\cS_{P(\bm{\theta})}}^{(n^{M_2} M_1-1)}}\abs{\tr\left( O_d\otimes O\sum_{\eta\in A_{\pi}}\cE_{\eta}(\rho)\right)} \\
        \leq{}& \sum_{\pi \in \Pi_{\cS_{P(\bm{\theta})}}^{((n+1)^{M_2} M_1-1)} \setminus\Pi_{\cS_{P(\bm{\theta})}}^{(n^{M_2} M_1-1)}} 2\left((n+1)^{M_2}M_1 - 1\right)M\tr(\cE_{\pi}(\rho)) \\
        \leq{}& 2\left((n+1)^{M_2}M_1 - 1\right)MM_1M_2(n+1)^{M_2}\left(\frac{1}{2}\right)^{\lfloor \frac{n-1}{N}\rfloor} \tr(\rho) \\
        \leq{}& 4\left((n+1)^{M_2}M_1 - 1\right)MM_1M_2(n+1)^{M_2}\left(\frac{1}{2}\right)^{\frac{n-1}{N}} \tr(\rho).
    \end{align*}
    As the $N$ is dependent on $\bm{\theta}$, we can obtain that for any $\bm{\theta}$, there exists $N_{\bm{\theta}} > 0$ such that for any $\rho \in \cD(\cH_{P(\bm{\theta})})$ and $n\geq 2$,
    \begin{align*}
        & \sum_{\pi \in \Pi_{\cS_{P(\bm{\theta})}}^{((n+1)^{M_2} M_1-1)} \setminus\Pi_{\cS_{P(\bm{\theta})}}^{(n^{M_2} M_1-1)}}\abs{\tr\left( O_d\otimes O\sum_{\eta\in A_{\pi}}\cE_{\eta}(\rho)\right)} \\
        \leq{}& 4\left((n+1)^{M_2}M_1 - 1\right)MM_1M_2(n+1)^{M_2}\left(\frac{1}{2}\right)^{\frac{n-1}{N_{\bm{\theta}}}} \tr(\rho)
    \end{align*}
    then, it is easy to see that for any $x$, there exists $N_x >0$ such that for any $\rho \in \cD(\cH_{P(\bm{\theta})})$ and $n \geq 2$,
    \begin{align*}
        H_n(x) \equiv{}& \sum_{\pi \in \Pi_{\cS_{P(\bm{\theta}^*[\theta\mapsto x])}}^{((n+1)^{M_2} M_1-1)} \setminus\Pi_{\cS_{P(\bm{\theta}^*[\theta\mapsto x])}}^{(n^{M_2} M_1-1)}}\abs{\tr\left( O_d\otimes O\sum_{\eta\in A_{\pi}}\cE_{\eta}(\rho)\right)} \\
        \leq{}& 4\left((n+1)^{M_2}M_1 - 1\right)MM_1M_2(n+1)^{M_2}\left(\frac{1}{2}\right)^{\frac{n-1}{N_x}} \tr(\rho).
    \end{align*}
    With
    \begin{align*}
        \lim_{n\to\infty}\sqrt[n]{4\left((n+1)^{M_2}M_1 - 1\right)MM_1M_2(n+1)^{M_2}\left(\frac{1}{2}\right)^{\frac{n-1}{N_x}}} &= \left(\frac{1}{2}\right)^{\frac{1}{N_x}} \\
        &< 1
    \end{align*}
    we have that for any $x\in\mathbb{R}$, $\sum_{n=1}^{\infty} H_n(x)$ is convergent. Since each term of $H_n(x)$ is non-negative, $H_n(x), n\in \mathbb{N}$ is a monotone sequence of continuous functions, then by Dini's Theorem~\cite[Theorem 7.13]{rudin1976principles}, $H_n(x)$ is uniform convergent on any close interval.

    Now, for any $n \geq 3$ and any $x\in \mathbb{R}$, we have
    \begin{align*}
        \abs{h_{n^{M_2}M_1-1}(x)}
        &= \abs{\sum_{\pi\in \Pi_{\cS_{P(\bm{\theta}^*[\theta\mapsto x])}}^{(n^{M_2}M_1-1)}}\tr\left(O_d\otimes O \sum_{\eta\in A_{\pi}}\cE_{\eta}(\rho)\right)} \\
        &\leq \sum_{\pi\in \Pi_{\cS_{P(\bm{\theta}^*[\theta\mapsto x])}}^{(n^{M_2}M_1-1)}}\abs{\tr\left(O_d\otimes O \sum_{\eta\in A_{\pi}}\cE_{\eta}(\rho)\right)} \\
        &= \sum_{\pi\in \Pi_{\cS_{P(\bm{\theta}^*[\theta\mapsto x])}}^{(2^{M_2}M_1-1)}}\abs{\tr\left(O_d\otimes O \sum_{\eta\in A_{\pi}}\cE_{\eta}(\rho)\right)} \\
        & \quad + \sum_{k=2}^{n-1} \sum_{\pi \in \Pi_{\cS_{P(\bm{\theta}^*[\theta\mapsto x])}}^{((k+1)^{M_2} M_1-1)} \setminus\Pi_{\cS_{P(\bm{\theta}^*[\theta\mapsto x])}}^{(k^{M_2} M_1-1)}} \abs{\tr\left(O_d\otimes O \sum_{\eta\in A_{\pi}}\cE_{\eta}(\rho)\right)} \\
        &= \sum_{\pi\in \Pi_{\cS_{P(\bm{\theta}^*[\theta\mapsto x])}}^{(2^{M_2}M_1-1)}}\abs{\tr\left(O_d\otimes O \sum_{\eta\in A_{\pi}}\cE_{\eta}(\rho)\right)} + \sum_{k=2}^{n-1} H_{k}(x).
    \end{align*}
    Because $H_n(x)$ is uniform convergent on any close interval, then $h_n(x)$ is uniform convergent on any close interval. We can also check that
    \begin{align*}
        \lim_{n\to\infty}h_n(x) &= \lim_{n\to\infty} \sum_{\pi\in\Pi_{\cS_{P(\bm{\theta}^*[\theta\mapsto x])}}^{(n)}}\tr\left(O_d\otimes O\sum_{\eta\in A_{\pi}}\cE_{\eta}(\rho)\right) \\
        &= \sum_{\pi\in\Pi_{\cS_{P(\bm{\theta}^*[\theta\mapsto x])}}}\tr\left(O_d\otimes O\sum_{\eta\in A_{\pi}}\cE_{\eta}(\rho)\right) \\
        &= \tr\left(O_d\otimes O \sem{\qd{\theta}(P(\bm{\theta}^*[\theta\mapsto x]))}(\rho)\right). \tag{by Equation~(\ref{eq:18})}
    \end{align*}
\end{proof}

\begin{proof}
    [proof of Theorem~\ref{thm:exact}]
    As the same settings in Lemma~\ref{lem:a2}, for each $\pi\in\Pi_{P(\bm{\theta})}$, we define
    \begin{align*}
        f_{\pi}(\bm{\theta}) &= \tr(O\cE_{\pi}(\rho)) \\
        g_{\pi}(\bm{\theta}) &= \sum_{\eta\in A_{\pi}}\tr(O_d\otimes O \cE_{\eta}(\rho))
    \end{align*}
    With Lemma~\ref{lem:a5}, we assume $\pi = l_{in}^{P(\bm{\theta})} \overto{\cE_1} m_1 \overto{\cE_2} m_2 \overto{\cE_3} \cdots \overto{\cE_n} m_n$ and $1\leq i_1 < i_2< \ldots<i_k \leq n$ such that $\cE_{i_j} = [e^{-i\theta\sigma_{i_j}}]_{\bar{q_{i_j}}}$ for $j = 1, 2,\ldots, k$, then
    \begin{align*}
        g_{\pi}(\bm{\theta}) =& \sum_{j=1}^k\Bigl(\tr\bigl(O\cE_{\pi_{i_j}\cdots\pi_{n}}(-i[[\sigma_{i_j}]_{\bar{q}_{i_j}}, \cE_{\pi_1\cdots\pi_{i_j-1}}(\rho)])\bigr)\Bigr).
    \end{align*}
    With Lemma~\ref{lem:a4}, we have
    \begin{align*}
        \frac{\partial}{\partial \theta_j}(f_{\pi}(\bm{\theta}))
        ={} \tr\bigl(O\cE_{\pi_{i_j}\cdots\pi_{n}}(-i[[\sigma_{i_j}]_{\bar{q}_{i_j}}, \cE_{\pi_1\cdots\pi_{i_j-1}}(\rho)])\bigr).
    \end{align*}
    for $j = 1,2,\ldots, k$ (which is considered as the $j$-th occurrence of $\theta$), then
    \begin{align*}
        \frac{\partial}{\partial \theta}(f_{\pi}(\bm{\theta}))
        ={}\sum_{j=1}^{k}\tr\bigl(O\cE_{\pi_{i_j}\cdots\pi_{n}}(-i[[\sigma_{i_j}]_{\bar{q}_{i_j}}, \cE_{\pi_1\cdots\pi_{i_j-1}}(\rho)])\bigr).
    \end{align*}
    Thus,
    \begin{equation}\label{eq:6}
        \frac{\partial}{\partial \theta}(f_{\pi}(\bm{\theta})) = g_{\pi}(\bm{\theta}).
    \end{equation}

    For $n\geq 1$, let 
    \[ f_n(\bm{\theta}) = \sum_{\pi\in \Pi_{\cS_{P(\bm{\theta})}}^{(n)}} f_{\pi}(\bm{\theta})\]
    \[ g_n(\bm{\theta}) = \sum_{\pi\in \Pi_{\cS_{P(\bm{\theta})}}^{(n)}} g_{\pi}(\bm{\theta})\]
    we have
    \[ \lim_{n\to\infty}f_n(\bm{\theta}) = f(\bm{\theta})\]
    \[ \lim_{n\to\infty} g_n(\bm{\theta}) = g(\bm{\theta})\]
    the correctness and existence of the second equation are guaranteed by Lemma~{\ref{lem:a10}}. By Equation~(\ref{eq:6}), we can easily check that for any $n \geq 1$,
    \[\frac{\partial}{\partial \theta}f_n(\bm{\theta}) = g_{n}(\bm{\theta}).\]
    With Lemma~\ref{lem:a10}, $g_n(\bm{\theta})$ is uniform convergent on a close interval $[\theta-\epsilon, \theta+\epsilon]$ for any $\alpha \in \mathbb{R}$ and any $\epsilon > 0$, which means $\frac{\partial}{\partial \theta}(f_n(\bm{\theta}))$ is uniform convergent on $[\theta-\epsilon, \theta+\epsilon]$, then
    \[ \lim_{n\to\infty}\frac{\partial}{\partial \theta}f_n(\bm{\theta}) = \frac{\partial}{\partial \theta}(\lim_{n\to\infty}f_n(\bm{\theta})).\]
    Thus,
    \begin{align*}
        \frac{\partial}{\partial \theta}f(\bm{\theta}) &= \frac{\partial}{\partial \theta}\left(\lim_{n\to\infty}f_n(\bm{\theta})\right) = \lim_{n\to\infty}\frac{\partial}{\partial \theta}f_n(\bm{\theta}) = \lim_{n\to\infty}g_{n}(\bm{\theta}) = g(\bm{\theta}).
    \end{align*}
\end{proof}

\subsection{Proof of Theorem~\ref{thm:bound}}\label{prf:bound}
The proof is based on Appendix~\ref{prf:exact}. To obtain a more accurate estimation, we define the set $\Gamma_{\theta}^{(n)} \subseteq \Pi_{\cS_{P(\bm{\theta})}}$ for a given parameter symbol $\theta$ and $n\in \mathbb{N}$:
\begin{align*}
    \Gamma_{\theta}^{(n)} \equiv{}& \lbrace \pi \in \Pi_{\cS_{P(\bm{\theta})}} : \text{$\theta$ appears $k$ times on path $\pi$ and $0\leq k \leq n$}\rbrace.
\end{align*}
For the set $\Gamma_{\theta}^{(n)}$, we have a lemma similar to Lemma ~\ref{lem:a5}.
\begin{lemma}\label{lem:a12}
    For any $\pi\in\Gamma_{\theta}^{(n)}\setminus \Gamma_{\theta}^{(0)}, n \geq 1$, $\pi$ can be written as $\pi = l_{in}^{P(\bm{\theta})} \overto{\cE_1} l_1 \overto{\cE_2} l_2 \overto{\cE_3} \cdots \overto{\cE_m} l_m$ with $1\leq i_1 < \ldots<i_k \leq m$, $1\leq k\leq n$ such that $\cE_{i_j} = [e^{-i\theta\sigma_{i_j}}]_{\bar{q_{i_j}}}$ for $j = 1,\ldots, k$. These $\cE_{i_j}, 1\leq j\leq k$ correspond to $k$ times occurrences of $\theta$. Then for any observable $O$, we have
        \begin{align*}
            &\tr\left(O_d^2\otimes O^2 \sum_{\eta\in A_{\pi}}\cE_{\eta}(\rho)\right) \\
            =& \sum_{j=1}^k\biggl(\frac{2}{\mu(j)} \biggl( \tr(O^2\cE_{\pi}(\rho)) + \tr\bigl(O^2\cE_{\pi_{i_j}\pi_{i_j+1}\cdots\pi_n}(\tr_{\bar{q}_{i_j}}(\cE_{\pi_1\cdots\pi_{i_j-1}}(\rho))\otimes\sigma_{i_j})\bigr)\biggr)\biggr).
        \end{align*}
\end{lemma}
\begin{proof}
    The proof is similar to the proof of Lemma~\ref{lem:a5}.
\end{proof}

We also need another lemma that is similar to Lemma~\ref{lem:a9}.

\begin{lemma}\label{lem:a13}
    Consider a quantum loop $P \equiv \qqwhile{\bar{q}}{Q}$ with fixed parameters (omitted). Assume that the state space $\cH_P$ is finite-dimensional and $P$ terminates almost surely. We define superoperators $\cE_i: \cD(\cH_P) \to \cD(\cH_P)$ by $\cE_i(\rho) = M_i \rho M_i^{\dagger}$, $i=0,1$ and  $\cE: \cD(\cH_P) \to \cD(\cH_P)$ by $\cE(\rho)= \sem{Q}(\rho)$. Then for any $\epsilon \in (0, 1)$, there exists $N = N_{\epsilon} > 0$ such that $\forall n \in \mathbb{N}, \forall \rho \in \cD(\cH_P),$
    \[ \tr((\cE\circ\cE_1)^n(\rho)) \leq \epsilon^{\lfloor \frac{n}{N}\rfloor}\tr(\rho).\]
\end{lemma}
\begin{proof}
    Because $P$ terminates almost surely, we have that the operator $P_{\cY}$ in Appendix~\ref{prf:a9} is an identify operator $P_{\cY} = I$. By Lemma~\ref{lem:a8}, for any $\epsilon\in (0,1)$, there exists $N > 0$ such that $(\cG^*)^N(P_{\cY}) \sqsubseteq \epsilon P_{\cY}$, which is
    \[ (\cG^*)^N(I) \sqsubseteq \epsilon I\]
    Therefore, for any $n\in \mathbb{N}$, any $\rho \in \cD(\cH_P)$, we have
    \begin{align*}
        \tr((\cE\circ\cE_1)^n(\rho)) = \tr((\cG^*)^n(I)\rho) &\leq \tr(\epsilon^{\lfloor \frac{n}{N}\rfloor}I\rho) = \epsilon^{\lfloor \frac{n}{N}\rfloor}\tr(\rho)
    \end{align*}
\end{proof}

We then follow the previous proof of Lemma~\ref{lem:a10}.

\begin{proof}
    [Proof of Theorem~\ref{thm:bound}]
    With the Lemma~\ref{lem:a12} and a similar proof of Inequality~(\ref{eq:17}), we have that for any $\pi\in \Gamma_{\theta}^{(n)}$,
    \begin{align}\label{eq:19}
        & \abs{\tr\left(O^2\otimes O_c^2 \sum_{\eta\in A_{\pi}}\cE_{\eta}(\rho)\right)} \notag\\
        \leq{}& \sum_{j=1}^{k_{\pi}}\biggl(\frac{2}{\mu(j)} \Bigl( M^2\tr(\cE_{\pi}(\rho))  + M^2\tr\bigl(\cE_{\pi_{i_j}\pi_{i_j+1}\cdots\pi_n}(\sigma_{\pi_{i_j}})\bigr)\Bigr)\biggr) \notag\\
        ={}& \sum_{j=1}^{k_{\pi}} \biggl(\frac{2M^2}{\mu(j)}\tr(\cE_{\pi}(\rho)) + \frac{2M^2}{\mu(j)}\tr(\cE_{\pi_{i_j}\pi_{i_j+1}\cdots\pi_n}(\sigma_{\pi_{i_j}})) \biggr)
    \end{align}
    where $M$ is the largest eigenvalue of $\abs{O}$, ${k_{\pi}}$ is the $k$ in Lemma~\ref{lem:a12} and \[\sigma_{\pi_{i_j}} \equiv \tr_{\bar{q}_{i_j}}(\cE_{\pi_1\cdots\pi_{i_j-1}}(\rho))\otimes\sigma_{i_j}.\]

    Let $M_2 = LC(P(\bm{\theta}))$. For convenience, we assume $M_2\geq 1$ temporarily.
    According to the definition of $LC$ (in Definite~\ref{dfn:lc}), $P(\bm{\theta})$ contains $M_2$ subprograms of \textbf{while} statements:
    \[ P_{j}(\bm{\theta}) \equiv \qwhile M^{(j)}[\bar{q}_j] = 1 \qdo Q_{j}(\bm{\theta}) \qod, \enspace j = 1,\ldots, M_2.\] 
    With Lemma~\ref{lem:a9}, for any $\epsilon\in (0,1)$, there exists $N_{\epsilon} \geq 1$ satisfies the following formula that is similar to Formula~(\ref{eq:11}),
    \begin{equation}\label{eq:20}
        \begin{aligned}
            &\forall n \geq 0, \forall \rho \in \cH_{P(\bm{\theta})}, 1\leq j\leq M_2, \\ &\tr(\cE_0^{(j)}\circ(\sem{Q_j(\bm{\theta})}\circ\cE_1^{(j)})^n(\rho)) \leq \epsilon^{\lfloor\frac{n}{N_{\epsilon}}\rfloor}\tr(\rho).
        \end{aligned}
    \end{equation}

    Let $M_1 = RC_{\theta}(P(\bm{\theta})) \geq 1$. For any $n \geq 1$, we consider $\pi \in \Gamma_{\theta}^{((n+1)^{M_2}M_1-1)}\setminus \Gamma_{\theta}^{(n^{M_2}M_1-1)}$, the parameter symbol $\theta$ appears on $\pi$ at least $n^{M_2}M_1$ times, then $\pi$ has at least $n^{M_2}M_1 / M_1 = n^{M_2}$ locations of parameter symbol $\theta$ repeatedly appear, which is caused by \textbf{while} statements. Then $\pi$ has at least $n^{M_2}$ times running into into the loop bodies of above $P_j(\bm{\theta}), j= 1, \ldots, M_2$. 
    By the same discussion in the proof of Lemma~\ref{lem:a10}, there exists $1\leq j_0 \leq M_2$ such that $\pi$ continuously runs into the loop body of $P_{j_0}(\bm{\theta})$ with more than $n-1$ times.

    With the help of Formula~(\ref{eq:20}) and Lemma~\ref{lem:a11}, we can also obtain an inequality similar to Inequality~(\ref{eq:21}) in the same way of the proof of Lemma~\ref{lem:a10}
    (notice: we can get a tighter bound by limiting $0\leq m \leq l(k_j) $ for every $1\leq j\leq M_2$,
    where $k_j$ is the depth of subprogram $P(\bm{\theta})_j$ nested with other $M_2-1$ subprograms mentioned before
    and when $k \geq 2, l(k)= (n-1)^2n^{k-2}$, $l(1) = 1$)\footnote{In the proof of Lemma~\ref{lem:a10}, we define the set $A_{j}^{(n,m)}$ for $0\leq m \leq (n+1)^{M_2}M_1-1$.
    However, we can prove that for any $m > l(k_j)$ and any $\pi \in A_{j}^{n,m}$, there exists $j'$ such that $k_{j'} < k_j$ and $\pi \in A_{j'}^{n, m'}$ and $m' \leq l(k_{j'})$ by induction on $k_j$. Thus we only need consider those sets $A_{j}^{(n,m)}$ with $0\leq m \leq l(k_j)$.},
    \begin{equation}\label{eq:22}
        \begin{aligned}
            \sum_{\pi \in \Gamma_{\theta}^{((n+1)^{M_2} M_1-1)} \setminus \Gamma_{\theta}^{(n^{M_2} M_1-1)}}\tr(\cE_{\pi}(\rho))
            \leq{}& \sum_{j=1}^{M_2}(l(k_j)+1)\epsilon^{\lfloor \frac{n-1}{N_{\epsilon}}\rfloor} \tr(\rho) \\
            \leq{}& \sum_{j=1}^{M_2}(l(j)+1)\epsilon^{\lfloor \frac{n-1}{N_{\epsilon}}\rfloor} \tr(\rho) \\
            \leq{}& (M_2+ (n-1)(n^{M_2-1}-1)) \epsilon^{\lfloor \frac{n-1}{N_{\epsilon}}\rfloor} \tr(\rho) \\
            \leq{}& (M_2+ (n-1)(n^{M_2-1}-1)_+) \epsilon^{\lfloor \frac{n-1}{N_{\epsilon}}\rfloor} \tr(\rho)
        \end{aligned}
    \end{equation}
    where $(x)_+ = \max\lbrace 0,x\rbrace$. This inequality also holds for $M_2 = 0$, because if $M_2=0$, then 
    \begin{align*}
        \Gamma_{\theta}^{((n+1)^{M_2} M_1-1)} \setminus \Gamma_{\theta}^{(n^{M_2} M_1-1)} &= \Gamma_{\theta}^{((n+1)^{0} M_1-1)} \setminus \Gamma_{\theta}^{(n^{0} M_1-1)}
        = \Gamma_{\theta}^{(M_1-1)} \setminus \Gamma_{\theta}^{(M_1-1)}
    \end{align*}
    is an empty set and $P(\bm{\theta})$ doesn't contain \textbf{while} statement, then $\Pi_{\cS_{P(\bm{\theta})}} = \Gamma_{\theta}^{(M_1)}$.

    With Inequality~(\ref{eq:19}) and Inequality~(\ref{eq:22}), we have that for $n \geq 1$,
    \begin{align*}
        &\sum_{\pi \in \Gamma_{\bm{\theta}}^{((n+1)^{M_2}M_1-1)}} \abs{\tr\left( O_d^2\otimes O \sum_{\eta\in A_{\pi}}\cE_{\eta}(\rho)\right)} \\
        \leq{}& \sum_{\pi \in \Gamma_{\bm{\theta}}^{((n+1)^{M_2}M_1-1)}} \sum_{j=1}^{k_{\pi}} \frac{2M^2}{\mu(j)}\tr(\cE_{\pi}(\rho))
        + \sum_{\pi \in \Gamma_{\bm{\theta}}^{((n+1)^{M_2}M_1-1)}} \sum_{j=1}^{k_{\pi}} \frac{2M^2}{\mu(j)}\tr(\cE_{\pi_{i_j}\pi_{i_j+1}\cdots\pi_n}(\sigma_{\pi_{i_j}})) \\
        \equiv{}& \text{Term-}A + \text{Term-}B.
    \end{align*}
    For Term-$A$:
    \begin{align*}
        &\sum_{\pi \in \Gamma_{\bm{\theta}}^{((n+1)^{M_2}M_1-1)}} \sum_{j=1}^{k_{\pi}} \frac{2M^2}{\mu(j)}\tr(\cE_{\pi}(\rho)) \\
        \leq{}& \bigggl(\sum_{\pi \in \Gamma_{\bm{\theta}}^{(1^{M_2}M_1)}} \sum_{j=1}^{k_{\pi}} \frac{2M^2}{\mu(j)}\tr(\cE_{\pi}(\rho))\bigggr) + \bigggl(\sum_{k=1}^n\sum_{\pi \in \Gamma_{\bm{\theta}}^{((k+1)^{M_2}M_1-1)}\setminus \Gamma_{\bm{\theta}}^{((k)^{M_2}M_1-1)}}\sum_{j=1}^{k_{\pi}} \frac{2M^2}{\mu(j)}\tr(\cE_{\pi}(\rho))\bigggr) \\
        \leq{}& \bigggl(\sum_{j=1}^{M_1}\sum_{\pi \in \Gamma_{\bm{\theta}}^{(M_1)}}  \frac{2M^2}{\mu(j)}\tr(\cE_{\pi}(\rho))\bigggr) + \bigggl(\sum_{k=1}^n\sum_{j=1}^{(k+1)^{M_2}M_1-1}\sum_{\pi \in \Gamma_{\bm{\theta}}^{((k+1)^{M_2}M_1-1)}\setminus \Gamma_{\bm{\theta}}^{((k)^{M_2}M_1-1)}} \frac{2M^2}{\mu(j)}\tr(\cE_{\pi}(\rho))\bigggr) \\
        \leq{}& \bigggl(\sum_{j=1}^{M_1} \frac{2M^2}{\mu(j)}\tr(\rho)\bigggr) + \bigggl(\sum_{k=1}^n\sum_{j=1}^{(k+1)^{M_2}M_1-1} \frac{2M^2}{\mu(j)} (M_2+(k-1)(k^{M_2-1}-1)_+) \epsilon^{\lfloor \frac{k-1}{N_{\epsilon}}\rfloor} \tr(\rho)\bigggr) \tag{By Inequality~(\ref{eq:22})}\\
        \leq{} & 2M^2\biggl(S(M_1) + \sum_{k=1}^n \Bigl((M_2+(k-1)(k^{M_2-1}-1)_+) S\left((k+1)^{M_2}M_1-1\right)\epsilon^{\lfloor \frac{k-1}{N_{\epsilon}}\rfloor}\Bigr)\biggr) \tr(\rho),
    \end{align*}
    where $S(n) = \sum_{j=1}^n\frac{1}{\mu(j)}$.
    
    For Term-$B$:
    \begin{align*}
        &\sum_{\pi \in \Gamma_{\bm{\theta}}^{((n+1)^{M_2}M_1-1)}} \sum_{j=1}^{k_{\pi}} \frac{2M^2}{\mu(j)}\tr(\cE_{\pi_{i_j}\pi_{i_j+1}\cdots\pi_n}(\sigma_{\pi_{i_j}})) \\
        ={}& \bigggl(\sum_{\pi \in \Gamma_{\bm{\theta}}^{(M_1)}} \sum_{j=1}^{k_{\pi}} \frac{2M^2}{\mu(j)}\tr(\cE_{\pi_{i_j}\pi_{i_j+1}\cdots\pi_n}(\sigma_{\pi_{i_j}}))\bigggr) \\
        & + \bigggl(\sum_{k=1}^n\sum_{\pi \in \Gamma_{\bm{\theta}}^{((k+1)^{M_2}M_1-1)}\setminus\Gamma_{\bm{\theta}}^{((k)^{M_2}M_1-1)}} \sum_{j=1}^{k_{\pi}} \frac{2M^2}{\mu(j)}\tr(\cE_{\pi_{i_j}\pi_{i_j+1}\cdots\pi_n}(\sigma_{\pi_{i_j}}))\bigggr) \\
        \equiv{}& \text{Term-}C + \text{Term-}D
    \end{align*}

    For Term-$C$, we consider all the $\sigma_{\pi_{i_j}} = \tr_{\bar{q}_{i_j}}(\cE_{\pi_1\cdots\pi_{i_j-1}}(\rho))\otimes\sigma_{i_j}$ that have same $\cE_{\pi_1\cdots\pi_{i_j-1}}(\rho)$, let
    \[E_{j} = \left\lbrace \pi_1\cdots\pi_{i_j-1} : \exists \eta \text{ s.t. } \pi = \pi_1\cdots\pi_{i_j-1}\eta \in \Gamma_{\bm{\theta}}^{(M_1)}\right\rbrace\]
    for every $1\leq j \leq M_1$ and 
    \[ F_{\pi_1\cdots\pi_{i_j-1}} = \left\lbrace \pi \in \Gamma_{\bm{\theta}}^{(M_1)} : \exists \eta \text{ s.t. } \pi = \pi_1\cdots\pi_{i_j-1}\eta \right\rbrace\]
    then
    \begin{align*}
        &\sum_{\pi \in \Gamma_{\bm{\theta}}^{(M_1)}} \sum_{j=1}^{k_{\pi}} \frac{2M^2}{\mu(j)}\tr(\cE_{\pi_{i_j}\pi_{i_j+1}\cdots\pi_n}(\sigma_{\pi_{i_j}})) \\
        ={}& \sum_{j=1}^{M_1}\sum_{\pi_1\cdots\pi_{i_j-1}\in E_j} \sum_{\pi\in F_{\pi_1\cdots\pi_{i_j-1}}} \frac{2M^2}{\mu(j)} \tr(\cE_{\pi_{i_j}\pi_{i_j+1}\cdots\pi_n}(\sigma_{\pi_{i_j}})) \\
        \leq{}& \sum_{j=1}^{M_1}\sum_{\pi_1\cdots\pi_{i_j-1}\in E_j} \frac{2M^2}{\mu(j)}\tr(\sigma_{\pi_{i_j}}) \tag{apply Lemma~\ref{lem:a11} to $F_{\pi_1\cdots\pi_{i_j-1}}$} \\
        ={}& \sum_{j=1}^{M_1}\sum_{\pi_1\cdots\pi_{i_j-1}\in E_j} \frac{2M^2}{\mu(j)}\tr(\tr_{\bar{q}_{i_j}}(\cE_{\pi_1\cdots\pi_{i_j-1}}(\rho))\otimes\sigma_{i_j}) \\
        ={}& \sum_{j=1}^{M_1}\sum_{\pi_1\cdots\pi_{i_j-1}\in E_j} \frac{2M^2}{\mu(j)}\tr(\cE_{\pi_1\cdots\pi_{i_j-1}}(\rho)) \\
        \leq {}& \sum_{j=1}^{M_1} \frac{2M^2}{\mu(j)} \tr(\rho) \tag{apply Lemma~\ref{lem:a11} to $E_j$} \\
        ={}& 2M^2S(M_1)\tr(\rho)
    \end{align*}

    For Term-$D$, we have already known that any 
    \[\pi \in \Gamma_{\bm{\theta}}^{((k+1)^{M_2}M_1-1)}\setminus\Gamma_{\bm{\theta}}^{((k)^{M_2}M_1-1)}\]
    must continuously runs into a loop body $P_k(\bm{\theta}), 1\leq k\leq M_2$ at least $n$ times. For any $\pi$ above, $\cE_{\pi_{i_j}\pi_{i_j+1}\cdots\pi_n}(\sigma_{\pi_{i_j}})$ is 
    \[\cE_{\pi_{i_j}\pi_{i_j+1}\cdots\pi_n}(\tr_{\bar{q}_{i_j}}(\cE_{\pi_1\cdots\pi_{i_j-1}}(\rho))\otimes\sigma_{i_j})\]
    For any $1\leq j\leq ((k+1)^{M_2}M_1-1)$, the location $\pi_{i_j}$ will only split the continuously runs of a loop body at most once, thus we can get a scale $ \epsilon^{\frac{n-1}{N_{\epsilon}}-1}$ to replace $\epsilon^{\lfloor \frac{n-1}{N_{\epsilon}}\rfloor}$ in Inequality~(\ref{eq:22}) by Lemma~\ref{lem:a9} and Lemma~\ref{lem:a13}:
    \begin{equation}\label{eq:24}
        \begin{aligned}
            &\sum_{\pi \in \Gamma_{\theta}^{((k+1)^{M_2} M_1-1)} \setminus \Gamma_{\theta}^{(k^{M_2} M_1-1)}}\tr(\rho_{\pi, j}) \\
            \leq{}&  (M_2+(k-1)(k^{M_2-1}-1)_+) \epsilon^{ \frac{n-1}{N_{\epsilon}}-1} \tr(\rho),
        \end{aligned}
    \end{equation}
    where $\rho_{\pi, j} = \cE_{\pi_{i_j}\pi_{i_j+1}\cdots\pi_n}(\sigma_{\pi_{i_j}})$ if $\pi$ has $j$-th occurrence of $\theta$, or equal to $\cE_{\pi}$ if $\pi$ has no $j$-th occurrence of $\theta$.

    Thus
    \begin{align*}
        &\sum_{\pi \in \Gamma_{\bm{\theta}}^{((k+1)^{M_2}M_1-1)}\setminus\Gamma_{\bm{\theta}}^{((k)^{M_2}M_1-1)}} \sum_{j=1}^{k_{\pi}} \frac{2M^2}{\mu(j)}\tr(\cE_{\pi_{i_j}\pi_{i_j+1}\cdots\pi_n}(\sigma_{\pi_{i_j}})) \\
        =& \sum_{\pi \in \Gamma_{\bm{\theta}}^{((k+1)^{M_2}M_1-1)}\setminus\Gamma_{\bm{\theta}}^{((k)^{M_2}M_1-1)}} \sum_{j=1}^{k_{\pi}} \frac{2M^2}{\mu(j)}\tr(\rho_{\pi, j}) \\
        \leq{}& \sum_{\pi \in \Gamma_{\bm{\theta}}^{((k+1)^{M_2}M_1-1)}\setminus\Gamma_{\bm{\theta}}^{((k)^{M_2}M_1-1)}} \sum_{j=1}^{((k+1)^{M_2}M_1-1)} \frac{2M^2}{\mu(j)}\tr(\rho_{\pi,j}) \\
        ={}& \sum_{j=1}^{((k+1)^{M_2}M_1-1)}\sum_{\pi \in \Gamma_{\bm{\theta}}^{((k+1)^{M_2}M_1-1)}\setminus\Gamma_{\bm{\theta}}^{((k)^{M_2}M_1-1)}}  \frac{2M^2}{\mu(j)}\tr(\rho_{\pi,j}) \\
        \leq{}& \sum_{j=1}^{((k+1)^{M_2}M_1-1)} \frac{2M^2}{\mu(j)}(M_2+(k-1)(k^{M_2-1}-1)_+) \epsilon^{ \lfloor \frac{k-1}{N_{\epsilon}}\rfloor-1} \tr(\rho) \tag{by Inequality~(\ref{eq:24})}\\
        ={}& 2M^2(M_2+(k-1)(k^{M_2-1}-1)_+) S\left((k+1)^{M_2}M_1-1\right) \epsilon^{ \lfloor \frac{k-1}{N_{\epsilon}}\rfloor-1} \tr(\rho).
    \end{align*}
    Then, Term-$D$ can be bounded as follows:
    \begin{align*}
        &\bigggl(\sum_{k=1}^n\sum_{\pi \in \Gamma_{\bm{\theta}}^{((k+1)^{M_2}M_1-1)}\setminus\Gamma_{\bm{\theta}}^{((k)^{M_2}M_1-1)}} \sum_{j=1}^{k_{\pi}} \frac{2M^2}{\mu(j)}\tr(\cE_{\pi_{i_j}\pi_{i_j+1}\cdots\pi_n}(\sigma_{\pi_{i_j}}))\bigggr) \\
        \leq{}& 2M^2\sum_{k=1}^n  \Bigl((M_2+(k-1)(k^{M_2-1}-1)_+) S\left((k+1)^{M_2}M_1-1\right) \epsilon^{ \lfloor \frac{k-1}{N_{\epsilon}}\rfloor-1}\Bigr) \tr(\rho).
    \end{align*}
    Therefore, Term-$B$ is bounded as follows:
    \begin{align*}
        &\sum_{\pi \in \Gamma_{\bm{\theta}}^{((n+1)^{M_2}M_1-1)}} \sum_{j=1}^{k_{\pi}} \frac{2M^2}{\mu(j)}\tr(\cE_{\pi_{i_j}\pi_{i_j+1}\cdots\pi_n}(\sigma_{\pi_{i_j}})) \\
        \leq{} & 2M^2\biggl(S(M_1) + \sum_{k=1}^n \Bigl((M_2+(k-1)(k^{M_2-1}-1)_+) S\left((k+1)^{M_2}M_1-1\right)\epsilon^{\lfloor \frac{k-1}{N_{\epsilon}}\rfloor-1}\Bigr)\biggr) \tr(\rho).
    \end{align*}
    Finally, we have the following inequality:
    \begin{equation}\label{eq:23}
        \begin{aligned}
            & \sum_{\pi \in \Gamma_{\bm{\theta}}^{((n+1)^{M_2}M_1-1)}} \abs{\tr\left( O_d^2\otimes O \sum_{\eta\in A_{\pi}}\cE_{\eta}(\rho)\right)} \\
            \leq{}& M^2\bigggl(4S(M_1) + \sum_{k=1}^n\biggl((M_2+(k-1)(k^{M_2-1}-1)_+) \\
            & \qquad \cdot S\left((k+1)^{M_2}M_1-1\right)(2\epsilon^{\lfloor \frac{k-1}{N_{\epsilon}}\rfloor} + 2\epsilon^{\lfloor \frac{k-1}{N_{\epsilon}}\rfloor-1})\biggr)\bigggr)\tr(\rho)
        \end{aligned}
    \end{equation}
    that holds for any $n \geq 1$ and $M_1, M_2\geq 0$. 
    
    By the definition of $\Gamma_{\theta}^{(n)}$, we have
    \[\Pi_{\cS_{P(\bm{\theta})}} = \bigcup_{n=1}^{\infty}\Gamma_{\theta}^{(n)}.\]
    When $M_2\geq 1$, we also have
    \[\Pi_{\cS_{P(\bm{\theta})}} = \bigcup_{n=1}^{\infty}\Gamma_{\theta}^{(n^{M_2}M_1-1)}.\]
    
    With $\forall n \geq 1, \Gamma_{\theta}^{(n)} \subseteq \Gamma_{\theta}^{(n+1)}$ and Inequality~(\ref{eq:23}), we have
    \begin{align*}
        &\abs{\sum_{\pi\in\Pi_{\cS_{P(\bm{\theta})}}}\tr\left(O_d^2\otimes O^2 \sum_{\eta\in A_{\pi}}\cE_{\eta}(\rho)\right)} \\
        ={} & \abs{\sum_{\pi\in \bigcup_{n=1}^{\infty}\Gamma_{\theta}^{(n^{M_2}M_1)}}\tr\left(O_d^2\otimes O^2 \sum_{\eta\in A_{\pi}}\cE_{\eta}(\rho)\right)} \\
        \leq{} & \sum_{\pi\in \bigcup_{n=1}^{\infty}\Gamma_{\theta}^{(n^{M_2}M_1)}}\abs{\tr\left(O_d^2\otimes O^2 \sum_{\eta\in A_{\pi}}\cE_{\eta}(\rho)\right)} \\
        \leq{} & \lim_{n\to\infty} \sum_{\pi\in \Gamma_{\theta}^{(n^{M_2}M_1)}}\abs{\tr\left(O_d^2\otimes O^2 \sum_{\eta\in A_{\pi}}\cE_{\eta}(\rho)\right)} \\
        ={} & \lim_{n\to\infty} M^2\bigggl(4S(M_1) + \sum_{k=1}^n\biggl((M_2+(k-1)(k^{M_2-1}-1)_+) \\
        & \qquad \cdot S\left((k+1)^{M_2}M_1-1\right)(2\epsilon^{\lfloor \frac{k-1}{N_{\epsilon}}\rfloor} + 2\epsilon^{\lfloor \frac{k-1}{N_{\epsilon}}\rfloor-1})\biggr)\bigggr)\tr(\rho) \\
        ={}& M^2\bigggl(4S(M_1) + \sum_{k=1}^{\infty}\biggl((M_2+(k-1)(k^{M_2-1}-1)_+) \\
        & \qquad \cdot S\left((k+1)^{M_2}M_1-1\right)(2\epsilon^{\lfloor \frac{k-1}{N_{\epsilon}}\rfloor} + 2\epsilon^{\lfloor \frac{k-1}{N_{\epsilon}}\rfloor-1})\biggr)\bigggr)\tag{by $\tr(\rho) \leq 1$}
    \end{align*}
    When $M_2 = 0$, we already know $\Pi_{\cS_{P(\bm{\theta})}} = \Gamma_{\theta}^{(M_1)}$, then the above inequality also holds.
    
    An equation similar to Equation~(\ref{eq:18}) is
    \begin{align*}
        \langle O_d^2\otimes O^2\rangle &= \tr\left(O^2\otimes O_c^2 \sem{\qd{\alpha}{\theta}(P(\bm{\theta}))}(\rho)\right)
        = \sum_{\pi\in\Pi_{\cS_{P(\bm{\theta})}}}\tr\left(O_d^2\otimes O^2 \sum_{\eta\in A_{\pi}}\cE_{\eta}(\rho)\right).
    \end{align*}
    Thus, we finally obtain the following inequality
    \begin{align*}
        &\langle O_d^2\otimes O^2\rangle \\
        \leq{}& M^2\bigggl(4S(M_1) + \sum_{k=1}^{\infty}\biggl((M_2+(k-1)(k^{M_2-1}-1)_+) S\left((k+1)^{M_2}M_1-1\right)\left(2\epsilon^{\lfloor \frac{k-1}{N_{\epsilon}}\rfloor} + 2\epsilon^{\lfloor \frac{k-1}{N_{\epsilon}}\rfloor-1}\right)\biggr)\bigggr).
    \end{align*}
    
\end{proof}

\end{document}